\documentclass{article}
\usepackage{graphicx} % Required for inserting images
\usepackage[a4paper, left=2.5cm, right=2.5cm, top=2.5cm, bottom=3cm]{geometry}

\usepackage{enumitem}
\setlist[itemize]{leftmargin=4em, label=\textbullet}  % Set defaults for itemize

\usepackage{setspace}
\usepackage[bottom]{footmisc} % tell latex to put footnotes at the bottom of the page.

\usepackage[
  backend=biber,
  style=apa6,
  sortcites=false,
doi=false,
  url=false,
  eprint=false,
  isbn=false,
  alldates=year
]{biblatex}

\usepackage{pdflscape} % for specific pages in wide format
\usepackage{threeparttable}
\usepackage{amsmath,amsfonts,amssymb,mathtools, amsthm,adjustbox} % Ensures you have math symbol support
\allowdisplaybreaks % allows for a new page within an align environment
\usepackage{subcaption} % combine several figures in one figure
\usepackage{cancel}
\usepackage{nameref}
\usepackage{placeins}
\usepackage{hyperref}
\hypersetup{
    colorlinks=true,
    linkcolor=black,
    citecolor=black,
    urlcolor=black
}

\usepackage{xurl}

\usepackage{thmtools}

\newcommand{\E}{\mathbb{E}}

\newcommand{\UDYX}{U}

\newcommand{\indep}{\perp\!\!\!\perp}

\newcommand{\defeq}{\vcentcolon =}

\usepackage[most]{tcolorbox}
\newtcolorbox[auto counter]{procedurebox}[2][]{
  colback=white,
  colframe=black,
  fonttitle=\bfseries,
  title={Procedure~\thetcbcounter: #2},
  #1
}

\newtheorem{remark}{Remark} % remarks
\newtheorem{proposition}{Proposition}
\newtheorem{lemma}{Lemma}[section]
\newtheorem{definition}{Definition}

\newtheoremstyle{assumptionstyle}
  {3pt}   % space above
  {3pt}   % space below
  {\itshape} % body font
  {}      % indent
  {\bfseries} % head font
  {:}     % punctuation after head
  { }     % space after head
  {\thmname{#1} \thmnote{#3}} % <-- THIS removes parentheses

\theoremstyle{assumptionstyle}
\newtheorem*{assumption}{Assumption}
\newcommand{\assref}[1]{\text{\nameref{#1}}}
\newcommand{\tableref}[1]{\nameref{#1}}

\usepackage{pgf,tikz}
\usetikzlibrary{arrows,shapes.arrows,shapes.geometric,
shapes.multipart,decorations.pathmorphing,positioning,
swigs,calc}
\usetikzlibrary{arrows.meta, positioning}

\usepackage{xcolor}

\usepackage{tabularx} % more table options
\usepackage{multirow} % cells over several rows in tables

\usepackage{array}
\newcolumntype{Y}{>{\raggedright\arraybackslash}X} % wrapped, left-aligned X column

\usepackage{appendix} % for the appendix

\usepackage{centernot}

\usepackage{pifont}

\usepackage{subcaption}

\newcommand{\dzero}{\textcolor{blue}{0}}

\newcounter{model}

\definecolor{gray}{gray}{0.65}

\title{Causal Graphs for Conditional Parallel Trends}
\author{Michael C. Knaus\thanks{University of Tübingen, Mohlstraße 36, 72074 Tübingen, Germany. michael.knaus@uni-tuebingen.de} \and  Henri Pfleiderer\thanks{University of Tübingen, Mohlstraße 36, 72074 Tübingen, Germany. henri.pfleiderer@uni-tuebingen.de} }
\date{\today}

\begin{document}

\maketitle

\begin{abstract}
Difference-in-Differences (DiD) is a widely used research design that often relies on a \emph{conditional parallel trends} (CPT) assumption. 
In contrast to settings with unconfoundedness, where causal graphs provide powerful frameworks for reasoning about valid conditioning variables, general-purpose graphical tools for CPT are missing.
We introduce transformed Single World Intervention Graphs (SWIGs), the $\Delta$-SWIGs, and prove that they enable us to read off conditional independencies via $d$-separation that imply CPT.
Using $\Delta$-SWIGs, we study valid conditioning strategies for DiD in complex settings with multiple periods and time-varying covariates. We show that when time-varying covariates affect the outcome, controlling for post-treatment variables is required for identification. However, even when such controls are included, pre-treatment parallel trends are only informative about a subset of the assumptions required for unbiased post-treatment effects, highlighting the limitations of purely empirical justifications of CPT.
\end{abstract}

\newpage

\section{Introduction}

Difference-in-Differences (DiD) is one of the most widely used research designs for causal inference in economics and related fields \parencite[e.g.][]{goldsmith-pinkham_tracking_2024}. It is also studied in a rapidly growing methodological literature that examines its identifying assumptions, extensions, and limitations \parencite[see, e.g.][for reviews]{roth_whats_2023,dechaisemartin_credible_2023,arkhangelsky_causal_2024}. 
The key identifying assumption in most DiD applications is that outcomes for treated and untreated units would have evolved in parallel in a counterfactual world with no treatment. In empirical practice, this assumption is often imposed conditional on observed covariates, corresponding to conditional parallel trends (CPT). For example, 19 of the papers published in the \textit{American Economic Review} in 2024 and 2025 explicitly assume some form of parallel trends, 14 of which impose CPT, and 13 of which use time-varying controls. However, none of these papers provide a formal justification for the choice of conditioning variables beyond presenting evidence on parallel pre-trends.\footnote{Appendix \ref{app:literature} provides details about the literature review.} This indicates that guidance on how to justify appropriate conditioning variables remains limited, especially when conditioning on time-varying covariates.

In settings with unconfoundedness or instrumental variables, causal graphs provide well-established frameworks for reasoning about valid conditioning variables and causal identification (see, e.g., textbooks by \textcite{pearl_causal_2016}, \textcite{frolich_impact_2019}, \textcite{cunningham_causal_2021}, \textcite{huber_causal_2023}, \textcite{chernozhukov_applied_2024}, \textcite{hernan_causal_2024}, \textcite{wager_causal_2024}, or reviews by \textcite{heckman_econometric_2022}, \textcite{hunermund_causal_2025}, \textcite{cinelli_crash_2024}, \textcite{imbens_causal_2024}, \textcite{abbring_philip_2025}). However, these tools do not directly apply to DiD settings, where identification is typically justified by functional form assumptions on time-invariant unobserved confounders that enter outcomes in an additively separable way. In particular, the recent literature typically states these assumptions only for the untreated potential outcome, while leaving treated potential outcomes, and thus the individual treatment effect, unrestricted (see Appendix \ref{app:swas-lit} for a review). 
We build on Single World Intervention Graphs (SWIGs) developed by \textcite{richardson_single_2013} to represent this perspective graphically. In settings with unconfoundedness, SWIGs encode conditional independencies involving potential outcomes and treatment that justify standard identification of average treatment effects, which can be read off via $d$-separation, a standard tool in causal graphs \parencite[see][for a primer]{richardson_single_2013a}.

This paper introduces transformed SWIGs, the $\Delta$-SWIGs, and shows that they provide a graphical characterization of CPT via $d$-separation. In particular, $\Delta$-SWIGs encode conditional independencies involving \textit{differences in} potential outcomes that correspond to conditional parallel trends and justify identification in DiD settings.
This provides a general framework for reasoning about conditioning strategies in DiD settings. We use this framework to analyze how features of the causal structure—such as outcome dynamics, time-varying covariates, and treatment–covariate feedback—affect the existence and form of valid conditioning strategies in the standard 2x2 setting and the multi-period setting of \textcite{callaway_differenceindifferences_2021}. In particular, we characterize how CPT can be justified based on additive separability and the causal structure alone, show that outcome dynamics generally preclude it without further restrictions, and establish that when time-varying covariates affect the outcome, post-treatment variables must be controlled for. We also show that pre-treatment parallel trends are informative only about a subset of the assumptions required for unbiased post-treatment effects. Finally, we translate these insights into practical guidance for implementing conditional DiD with time-varying covariates.

To illustrate the main issues and findings, Figure~\ref{fig:simulations_intro} presents results from a simple simulation with unobserved time-invariant confounders and time-varying observed covariates. The latter may be unaffected by the treatment (Figure \ref{fig:sim-no-dx}) or affected by the treatment (Figure \ref{fig:sim-dx}) (see Appendix \ref{app:simulations} for details). The figure reports conditional DiD estimates from a large sample for the group first treated in period four under three conditioning strategies: using only pre-treatment covariates (left), using covariates measured prior to each outcome in the respective DiD comparison (``pre-outcome'' controls, middle), and using the full sequence of time-varying covariates (right). We highlight several patterns that may appear surprising but can all be rationalized within our framework: (i) using only pre-treatment covariates yields parallel pre-trends and unbiased short-term effects but biased dynamic effects, (ii) using pre-outcome controls provides the same results as using the full sequence of covariates, (iii) in the absence of treatment–covariate feedback, strategies involving post-treatment variables are unbiased, (iv) in the presence of treatment–covariate feedback, all dynamic effect estimates are biased, reflecting omitted variable bias or ``wrong world control bias'', and (v) pre-trends do not diagnose post-treatment violations of CPT, while short-term effects remain unbiased even with treatment–covariate feedback. \textcite{ghanem_when_2026} raise similar issues in a three-period setting.

\begin{figure}[t]
    \centering
    \caption{Simulation with time-varying covariates and different conditioning strategies}\label{fig:simulations_intro}
    \begin{subfigure}{\textwidth}
    \centering
            \includegraphics[width=0.8\linewidth]{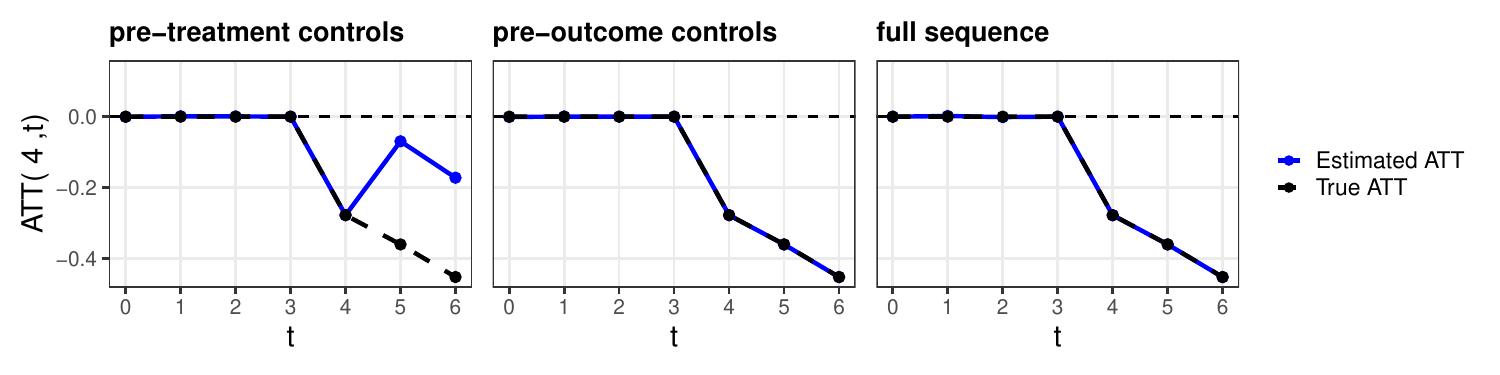}
            \caption{No treatment-covariate feedback}
        \label{fig:sim-no-dx}
    \end{subfigure}

\begin{subfigure}{\textwidth}
    \centering
            \includegraphics[width=0.8\linewidth]{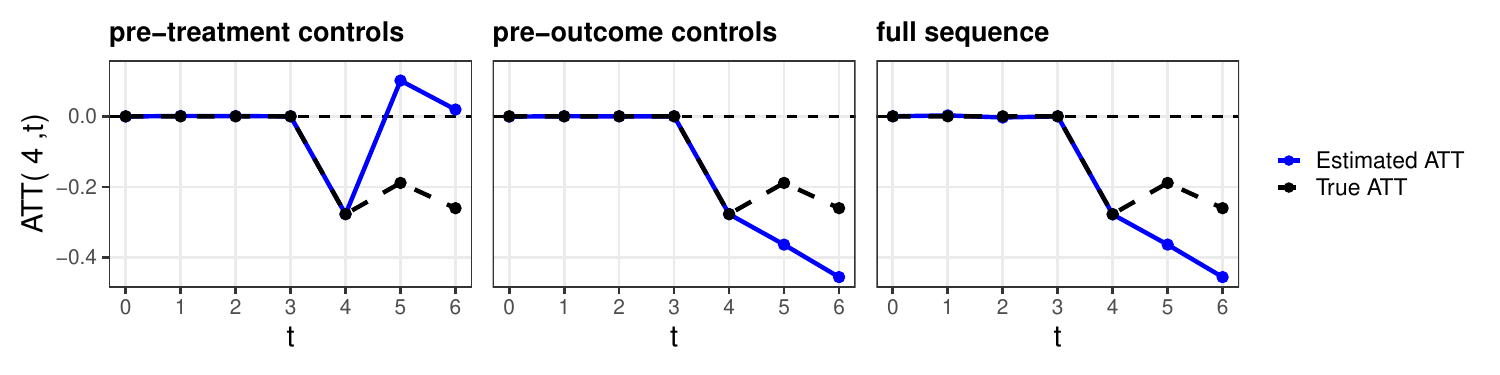}
            \caption{Treatment-covariate feedback}
        \label{fig:sim-dx}
    \end{subfigure}
\begin{minipage}{\linewidth}
\begin{small}
\textit{Note:} The sample size for the simulations is set to 10 million. Estimation is by conditional DiD with the never-treated as control group conditioning variables as specified in the subplot titles: pre-treatment covariates (left), covariates measured prior to each outcome in the respective DiD comparison (middle), and full sequence of time-varying covariates (right). The data-generating process with unobserved and additively separable time-invariant confounders and time-varying observed covariates and implementation details are provided in Appendix \ref{app:simulations}.   
\end{small}
\end{minipage}
\end{figure}

Overall, this suggests that identification should be derived from explicit structural considerations complemented by statistical evidence on parallel pre-trends, rather than treating CPT as a primitive assumption justified solely by such evidence. Our framework provides a transparent and principled way to conduct this reasoning and is particularly useful for analyses with time-varying covariates and multiple periods, especially when researchers are reluctant to rely on parametric or distributional assumptions beyond standard additive separability of unobserved confounders.

The remainder of the paper is organized as follows. Section~\ref{sec:related-lit} discusses the related literature and situates our contribution. Section~\ref{sec:prelim} introduces the necessary building blocks, including structural causal models, SWIGs, and $d$-separation. Section~\ref{sec:D-SWIG} develops the $\Delta$-SWIG framework and its connection to conditional parallel trends. Section~\ref{sec:time_varying_T2} and Section~\ref{sec:multiple} apply this framework to characterize valid conditioning strategies in 2x2 and multi-period settings. Section~\ref{sec:implications} discusses practical implications, and Section~\ref{sec:conclusion} concludes.

% Key ingredient of most DiD applications is a Conditional Parallel Trends (CPT) assumption like 
%     \begin{align*}\label{eq:CPT-X-2x2}\tag{CPT}
%         \underbrace{\E[Y_1(0) - Y_0(0) \mid X, D = 1]}_{\text{counterfactual trend of treated}} &=  \underbrace{\E[Y_1(0) - Y_0(0) \mid X, D = 0]}_{\text{factual trend of untreated}}  \\
%         % \Leftrightarrow \E[\Delta Y_1(0) \mid X, D = 1] &=  \E[\Delta Y_1(0)\mid X,D = 0]
%     \end{align*}

% Standard SWIGs are build to read off $Y(d) \indep D \mid X$ implying conditional mean independence $\E[Y(d) \mid X] = \E[Y(d) \mid X, D = d]$, which can be used to identify average potential outcome under unconfoundedness $\E[Y(d)] = \E[\E[Y(d) \mid X]] = \E[\E[Y(d) \mid X, D = d]] = \E[\E[Y \mid X, D = d]]$ (refs). The $\Delta$-SWIG allows to read off independencies like $Y_1(0) - Y_0(0) \indep D \mid X$ implying conditional parallel trends $\E[Y_1(0) - Y_0(0) \mid X, D=1] = \E[Y_1(0) - Y_0(0) \mid X,D=0]$, which are used to identify the counterfacutal time trend in DiD $\E[Y_1(0) - Y_0(0) \mid D=1] = \E[\E[Y_1(0) - Y_0(0) \mid X, D=1]\mid D=1] = \E[\E[Y_1(0) - Y_0(0) \mid X, D=0]\mid D=1] = \E[\E[Y_1 - Y_0 \mid X, D=1]\mid D=1]$ (ref).

\subsection{Related Literature} \label{sec:related-lit}

This paper is related and contributes to several strands of literature.
First, it contributes to the recent literature on justifying (conditional) parallel trends for DiD or DiD directly based on substantive knowledge about the data-generating process. 
Most prominently, \textcite{ghanem_selection_2024} provide selection-based necessary and sufficient conditions for parallel trends largely focusing on the two-period case without covariates. They discuss extensions to more time periods without covariates, or to two periods with time-varying covariates that are not affected by the treatment. We focus on sufficient conditions but also cover multiple time periods with time-varying covariates that may be affected by the treatment.
% Our graphical approach can deliver sufficient conditions for their main results on selection on unobservables.
% Their approach does, however, naturally accommodate parametric or distributional assumptions beyond the causal structure.
\textcite{ghanem_when_2026} study the informativeness of pre-trends in work conducted independently and concurrently with ours. In line with our findings, they show that with time-varying covariates, pre-trends based on pre-treatment controls can be misleading for causal effects in a three-period setting. Again, our results also hold for more time periods and provide more details about what can and cannot be deduced from pre-trends trends testing.
Therefore, our approach and results are complementary to those in \textcite{ghanem_selection_2024} and \textcite{ghanem_when_2026}.
\textcite{marx_parallel_2024} focus on specific models of dynamic choice and the parallel trends assumption.
\textcite{chabe-ferret_should_2025} considers a linear model with selection on pre-treatment outcomes. 
Both do not address the role of covariates.

The second strand of literature is on causal graphs for DiD.
\textcite{renson_using_2026} is most closely related using SWIGs to provide \textit{necessary} conditions for \textit{unconditional} parallel trends in \textit{two} time periods. % under a linear faithfulness assumption.
In contrast, we introduce $\Delta$-SWIGs to obtain \textit{sufficient} conditions for \textit{conditional} parallel trends with time-varying covariates and \textit{multiple} time periods. 
% We recover their insights that selection on pre-treatment outcomes and state-dependence are incompatible with parallel trends and extend the analysis to time-varying covariates and treatment covariate-feedback.
\textcite{weber_assumption_2015} compare identifying assumptions, including conditional parallel trends, relying on DAGs rather than SWIGs. 
% Also, their functional form restrictions apply to observed outcomes rather than only the untreated potential outcome, as is standard in the recent literature. 
\textcite{huber_joint_2024} uses DAGs to motivate a joint test of unconfoundedness and parallel trends.
Both studies are restricted to two time periods and cannot read off standard CPT involving differences in untreated potential outcomes because they are not represented in the DAGs.
\textcite{kim_gain_2021} consider ``gain scores'' - where the outcome of interest is the difference between pre- and post-treatment outcomes - in linear models and DAGs using path tracing rules.
\textcite{zhang_exploiting_2021} show how to incorporate equality constraints in linear models, with DiD as an important special case. 
Unlike these last two contributions, we do not restrict attention to linear models in two time periods and incorporate covariates.

Third, our work relates to the literature on DiD with time-varying covariates. 
\textcite{caetano_differenceindifferences_2024} establish identification under CPT conditional on the full sequence of time-varying covariates; our approach can be used to justify this assumption. We also show that controlling for fewer variables suffices under the same structural assumptions, thereby mitigating issues related to high dimensionality.
\textcite{caetano_difference_2024} consider CPT with post-treatment covariates in two periods, and provide conditions under which identification of treatment effects remains possible.
Again, our approach can provide sufficient conditions for these assumptions grounded in the causal structure, and further extends the analysis to multiple time periods and the informativeness of pre-trend testing.

% Fourth, \textcite{shahn_structural_2024} and \textcite{renson_identifying_2023} establish identification of causal effects under CPT with treatment-covariate feedback.
% Our graphical approach can be used to justify their CPT assumptions, which require ruling out unobserved confounding between the treatment and covariates.

Taken together, our work advances the literature along several dimensions. We introduce and show the validity of $\Delta$-SWIGs as a graphical tool to derive sufficient conditions for conditional parallel trends, accommodating time-varying covariates, treatment-covariate feedback, covariate dynamics, and multiple time periods — settings not jointly covered by existing approaches. Beyond identification, we provide additional testable implications and clarify when pre-trend tests are and are not informative about the validity of the identifying assumptions.

\section{Preliminaries} \label{sec:prelim}

% This section compactly introduces crucial building blocks for the main analysis below. 
Throughout the paper uppercase letters $V$ denote random variables, lowercase letters $v$ the corresponding realizations.
Bold letters denote sets of random variables $\mathbf{V}$ with realizations $\mathbf{v}$.

\subsection{2x2 Difference-in-Differences and Conditional Parallel Trends} 

The canonical 2x2 setting considers binary treatment $D$ and two periods $t=0,1$. It assumes potential outcomes $Y_t(d)$ in period $t$ under treatment $d$ and targets the average treatment effect on the treated $ATT := \E[Y_1(1) - Y_1(0) \mid D=1]$. It exploits that no unit is treated in period $t=0$ and some units are treated in $t=1$. Then, assuming parallel trends conditional on time-invariant control variables $X$ %\footnote{The unconditional case is implicitly covered as special case throughout the paper, but never explicitly discussed.}
\begin{align*}\label{eq:CPT-X-2x2}\tag{2x2-CPT}
        \underbrace{\E[\overbrace{Y_1(0) - Y_0(0)}^{\Delta Y_1(0)} \mid X, D = 1]}_{\text{counterfactual trend of treated}} &=  \underbrace{\E[\overbrace{Y_1(0) - Y_0(0)}^{\Delta Y_1(0)} \mid X, D = 0]}_{\text{factual trend of untreated}}
        % \Leftrightarrow \E[\Delta Y_1(0) \mid X, D = 1] &=  \E[\Delta Y_1(0)\mid X,D = 0]
\end{align*}
identifies the $ATT$ following standard arguments \parencite[see, e.g.][]{heckman_matching_1997, abadie_semiparametric_2005,lechner_estimation_2011}.%
\footnote{Identification also requires overlap. However, we abstract from overlap considerations in this paper and implicitly assume in our discussions that overlap holds. We also omit the no anticipation assumption at this stage. See Remark \ref{rem:no-anti} for a detailed discussion how it follows naturally once the graphical perspective is introduced.} 
We stick to this well-established setting until we have introduced all relevant components to consider more complex scenarios.

% The 2x2 setting with time-invariant controls is very closely related to unconfoundedness that usually assumes conditional independence $Y_1(0) \indep D \mid X$ implying conditional mean independence $\E[Y_1(0) \mid X, D=1] = \E[Y(d) \mid X,D = 0]$ to identify $ATT$ as $\E[Y(1) \mid D=1] - \E[\E[Y \mid X,D = 0] \mid D=1]$ (refs auch Rosenbaum Rubin ?, Heckman,..., discussion in Abadie). One leading application of SWIGs is to read off the conditional independence from a causal graph (cite Primer). The first step in this paper is to introduce $\Delta$-SWIGs to read off $Y_1(0) - Y_0(0) \indep D \mid X$ to justify \eqref{eq:CPT-X-2x2}. This allows us to introduce the required tools in a simple setting before applying the technology to more interesting scenarios.

\subsection{Structural Causal Models and Directed Acyclic Graphs} \label{sec:graphical_models}

A structural causal model (SCM) $\mathcal{M}$ consists of a set of endogenous variables $\mathbf V$,
exogenous variables $\mathbf U$, and functions $\mathcal{F} := \{f_{V_j}\}_{V_j \in \mathbf V}$. For each $V_j \in \mathbf V$, the model specifies a structural equation $V_j \defeq f_{V_j}(Pa(V_j), U_{V_j})$, where $Pa(V_j) \subseteq \mathbf{V} \setminus \{V_j\} $ denotes the parents (direct causes) of $V_j$, and $U_{V_j} \in \mathbf U$ are exogenous random variables.
The corresponding graph $\mathcal{G}$ contains a node for each $V_j \in \mathbf V$ and a directed edge from each member of $Pa(V_j)$ to $V_j$, with an arrowhead pointing to $V_j$.
% We use $Pa_\mathcal{G}(V_j)$ to denote the parent nodes of $V_j$ in $\mathcal{G}$, i.e., nodes with a direct edge to $V_j$.
If there exists a directed path $V_a \rightarrow ... \rightarrow V_d$, i.e.~with all arrows pointing towards $V_d$, then $V_d$ is call a descendant of $V_a$. We denote all descendants of a node $Desc(V_j)$.
If there are no cycles, i.e.~no directed sequences of edges from a node to itself, the graph is called a directed acyclic graph (DAG).\footnote{For a more in-depth introduction to SCMs and DAGs see, e.g., the textbooks by textbooks by \textcite{pearl_causal_2016}, \textcite{frolich_impact_2019}, \textcite{cunningham_causal_2021}, \textcite{huber_causal_2023}, \textcite{chernozhukov_applied_2024}, \textcite{hernan_causal_2024}, \textcite{wager_causal_2024}, or reviews by \textcite{heckman_econometric_2022}, \textcite{hunermund_causal_2025}, \textcite{cinelli_crash_2024}, \textcite{imbens_causal_2024}, \textcite{abbring_philip_2025}).}
% If the graph is acyclic and the variables $\mathbf U$ are mutually independent, the probability distribution over the variables $\mathbf V$, that is induced by the SCM and denoted $P(\mathbf{V})$, satisfies the Markov property with respect to the graph.
% This means that any variable $V_j$ is independent of its non-descendants conditional on its parents \parencite{pearl_causality_2009}.

Figure \ref{fig:ex-dag1} provides an example illustrating the conventions for causal graphs used throughout the paper: (i) exogenous variables like $U_{Y_0}$ are omitted unless their presence is useful; (ii) a missing circle around a variable indicates that it is unobservable. 
Figure \ref{fig:ex-dag1} depicts a setting with time-invariant observable confounders $X$ and unobservable time-invariant confounders $U$ where outcomes are measured in the pre- and post-treatment period. This is a setting where standard 2x2 DiD would often be considered to overcome unobserved confounding/heterogeneity indicated by the open backdoor path $D \leftarrow U \rightarrow Y_1$, which prevents identification of causal effects by controlling for $X$ \parencite[see, e.g.][Ch. 3.3]{pearl_causal_2016}. 

We note that common DAGs are useful to illustrate the problem in this setting, but not expressive enough to incorporate the DiD solution. In particular, they do not allow to read-off sufficient conditions for conditional parallel trends. This is in contrast to unconfoundedness settings where DAGs paired with the backdoor criterion are powerful to reason about good, bad and neutral controls \parencite{cinelli_crash_2024}. We build on Single World Intervention Graphs paired with $d$-separation to provide graph-based justifications for control variables in DiD. Both components are introduced in the following.

\begin{figure}[t]
    \caption{Causal graphs for a simple 2x2 setting and their corresponding SCMs}\label{fig:ex-dag}
    \centering
    % ================= Subfigure (a) =================
    \begin{subfigure}{0.49\textwidth}
        \centering
  \resizebox{0.6\linewidth}{!}{% 
        \begin{tikzpicture}
            \tikzset{line width=1.5pt, ell/.style={draw, fill = white, inner xsep=5pt,inner ysep=5pt, line width=1.5pt}, unobs/.style={ fill = none, inner xsep=5pt,inner ysep=5pt, line width=1.5pt}, swig vsplit={gap=8pt, inner line width right=0.5pt}};

            % Nodes
            \node[name=y0, ell, shape=ellipse]{$Y_0$};
            \node[name=y1, ell, shape=ellipse] at ($(y0)+(4cm,0cm)$) {$Y_1$};
            \node[name=d, ell, shape=ellipse] at ($(y1)+(-1cm,-2cm)$) {$D$};
            \node[name=u_dyx, unobs, shape=ellipse] at ($(y1)+(-2cm,-3.5cm)$) {$\UDYX$};
            \node[name=x, ell, shape=ellipse] at ($(y0)+(0cm,-2cm)$) {$X$};

            % Edges
            \draw[->,line width=1.5pt,>=stealth]
                (d) edge (y1)
                (u_dyx) edge (y0)
                (u_dyx) edge[bend right = 50, color=orange] (y1)
                (u_dyx) edge[color=orange] (d)
                (u_dyx) edge (x)
                (x) edge (y0)
                (x) edge (y1)
                (x) edge (d);
        \end{tikzpicture}
}
{\small
        \begin{align}
            Y_0 &\defeq f_{Y_0}(\UDYX,X,U_{Y_0})  \label{eq:scm-2x2-y0}\\
            Y_1 &\defeq f_{Y_1}(\UDYX,X,D,U_{Y_1}) \label{eq:scm-2x2-y} \\
            X & \defeq f_X (\UDYX,U_X) \\
            D   &\defeq f_{D}(\UDYX,X,U_D) \\
            U & \defeq f_U (U_U)
        \end{align}
}
        \caption{2x2 DAG}
        \label{fig:ex-dag1}
    \end{subfigure}
    % ================= Subfigure (b) =================
    \begin{subfigure}{0.49\textwidth}
        \centering
  \resizebox{0.65\linewidth}{!}{% 
    \begin{tikzpicture}
        \tikzset{line width=1.5pt, ell/.style={draw, fill = white, inner xsep=5pt,inner ysep=5pt, line width=1.5pt}, unobs/.style={ fill = none, inner xsep=5pt,inner ysep=5pt, line width=1.5pt}, swig vsplit={gap=8pt, inner line width right=0.5pt}};

% Nodes
    \node[name=y0, ell, shape=ellipse]{$Y_0\textcolor{gray}{(0)}$};
    \node[name=y1, ell, shape=ellipse] at ($(y0)+(4cm,0cm)$) {$Y_1(0)$};
    \node[name=d, ell, shape=swig vsplit] at ($(y1)+(-1cm,-2cm)$) {\nodepart{left}{$D$} \nodepart{right}{$0$}};
    \node[name=u_dyx, unobs, shape=ellipse] at ($(y1)+(-2.5cm,-3.5cm)$) {$\UDYX$};
    % \node[name=u_y0, unobs, shape=ellipse] at ($(y0)+(0cm,+1.5cm)$) {$U_{Y_0}$};
    % \node[name=u_y1, unobs, shape=ellipse] at ($(y1)+(0cm,+1.5cm)$) {$U_{Y_1}$};
    \node[name=x, ell, shape=ellipse] at ($(y0)+(0cm,-2cm)$) {$X$};

% Edges
\draw[->,line width=1.5pt,>=stealth]
(d) edge (y1)
(u_dyx) edge (y0)
(u_dyx) edge[bend right = 50] (y1)
(u_dyx) edge (d)
% (u_y0) edge (y0)
% (u_y1) edge (y1)
(u_dyx) edge (x)
(x) edge (y0)
(x) edge (y1)
(x) edge (d);
\end{tikzpicture}
}
{\small
\begin{align*}
    Y_0\textcolor{gray}{(0)} &\defeq f_{Y_0}(\UDYX,X,U_{Y_0}) \\
    % &= \alpha(\UDYX) + g_{Y_0}(X,U_{Y_0}) \\
    Y_1(0) &\defeq f_{Y_1}(\UDYX,X,0,U_{Y_1}) \\
    % &= \alpha(\UDYX) + g_{Y_1}(X,U_{Y_1}) \\
    X & \defeq f_X (\UDYX,U_X) \\
    D   &\defeq f_{D}(\UDYX,X,U_D) \\
            U & \defeq f_U (U_U) 
\end{align*}
}
        \caption{2x2 SWIG}
        \label{fig:ex-dag2}
    \end{subfigure}
\end{figure}
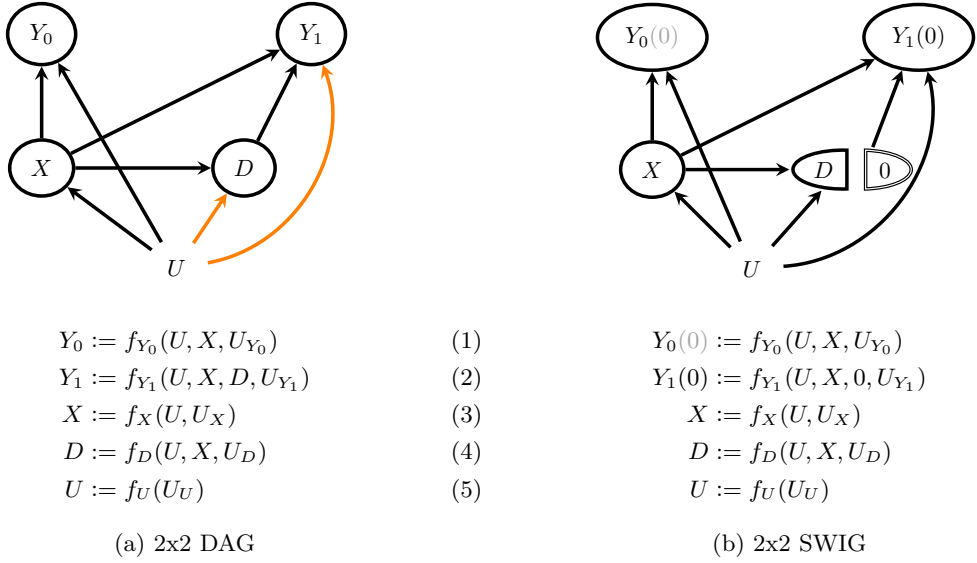

\subsection{Single World Intervention Graphs} \label{sec:swigs}

Single World Intervention Graphs (SWIGs) developed by \textcite{richardson_single_2013} unify DAGs and the potential outcomes framework.
In the underlying SCM, potential outcomes are obtained by fixing possibly multiple treatment variables $\mathbf D$ to a particular value $\mathbf d$ in the structural equations.\footnote{This \textit{fix} operator is discussed, e.g. in \textcite{heckman_causal_2015}. See Table 3.1 in \textcite{heckman_econometric_2022} for an illustration of the difference between the \textit{fix} and the \textit{do} operator, which is commonly applied in the context of DAGs.} This results in a new set of endogenous variables $\mathbf V (\mathbf d)$ consisting of observable and potential variables. For example, fixing $D = 0$ in Equation \eqref{eq:scm-2x2-y} of Figure \ref{fig:ex-dag1} yields $Y_1(0) := f_Y(U,X,0,U_{Y_1})$ representing the potential outcome in a counterfactual world where every unit remains untreated. All other variables remain unchanged as their structural equations do not depend on $D$ or its descendants.

The SWIG $\mathcal{G}(\mathbf{d})$ graphically mirrors the hypothetical intervention of fixing $\mathbf D = \mathbf d$ by transforming the original DAG in two steps. First, each of the treatment nodes $D_j \in \mathbf D$ are \textit{split} into a random part $D_j$ and a fixed part $d_j$. 
The random part of the split node inherits the incoming edges and the fixed part the outgoing edges.
Second, all descendants of the variables in $\mathbf D$ are \textit{relabeled} to become potential outcomes.\footnote{\textcite{hernan_causal_2024} and \textcite{chernozhukov_applied_2024} provide textbook introductions to SWIGs. \textcite{andersen_guide_2023} use SWIGs to discuss sample selection.}
We use a subscript to indicate if a specific expression refers to a SWIG, e.g.~$Pa_{\mathcal{G}(\mathbf{d})}(V_j (\mathbf d))$ denotes the parents of node $V_j (\mathbf d)$ in SWIG $\mathcal{G}(\mathbf{d})$.

Figure \ref{fig:ex-dag2} shows the SWIG for the two period case with unobserved confounding where the treatment node is split and its descendant node is relabeled as potential outcome $Y_1(0)$. The relabeling of the pre-treatment outcome in gray is optional but useful as explained in the following remark:
\begin{remark}(no anticipation \& \textcolor{gray}{color} convention) \label{rem:no-anti} \\
    The no anticipation assumption in the DiD literature states that potential outcomes in the pre-treatment period are unaffected by the treatment, i.e.~$Y_0(1) = Y_0(0) = Y_0$. This is encoded in DAG \ref{fig:ex-dag1} by a missing arrow $D \rightarrow Y_0$. Thus, $Y_0$ is not a descendant of $D$ and would not be relabeled as a potential outcome in the default relabeling step of \textcite{richardson_single_2013}. However, $\{Y_0\textcolor{gray}{(1)}, Y_0\textcolor{gray}{(0)}, Y_0\}$ represent an equivalence class in SWIG \ref{fig:ex-dag2}, which means that non-descendants of the treatment can be relabeled as potential outcomes if it is useful \parencite[Sec.~3.4]{richardson_single_2013}.
    Thus, we also relabel pre-treatment outcome nodes in our SWIGs to maintain compatibility with the DiD literature.\footnote{\textcite{richardson_single_2013} state on p.~30 "There may be situations in which we might find it useful to label a node with another variable from [the same equivalence class]". We think ensuring backwards compatibility with the previous literature is such a situation.} 
    To remind the reader that certain parts of the label are redundant, we display them in gray.
\end{remark}

\begin{remark}(consistency \& \textcolor{blue}{color} convention) \label{rem:consistency} \\
    Observable variable $V_j$ equals potential variable $V_j (\mathbf d)$ if $\mathbf D = \mathbf d$.\footnote{For example, $\E[Y_1 \mid D=0] = \E[f_{Y_1}(\UDYX,X,D,U_{Y_1}) \mid D=0] = \E[f_{Y_1}(\UDYX,X,0,U_{Y_1}) \mid D=0] = \E[Y_1(0) \mid D=0]$ in the setting of Figure \ref{fig:ex-dag}.} This consistency is used repeatedly in the following. To indicate when and why potential variables can be replaced by observable variables, we highlight the relevant match in blue. For example, $\E[Y_1(\dzero) \mid D=\dzero]$ is equivalent to $\E[Y_1 \mid D=0]$.
\end{remark}

\subsection{$d$-separation in SWIGs for Conditional Independencies} 

The benefit of SWIGs is that they allow to read off (conditional) independencies involving observable and potential variables via $d$-separation, which is a standard tool in the causal graph literature \parencite[e.g.][]{verma_causal_1988}. 
Two nodes $X,Y \in \mathbf{V}(\mathbf{d})$ are said to be $d$-separated by a set of nodes $\mathbf Z \subseteq \mathbf{V}(\mathbf{d}) \setminus \{X,Y\}$ in SWIG $\mathcal{G}(\mathbf{d})$ if every path between $X$ and $Y$ is blocked by the set $\mathbf Z$.
Practically, a path can be blocked i) by including the middle node of chains ($V_i \rightarrow V_j \rightarrow V_k$) or forks ($V_i \leftarrow V_j \rightarrow V_k$) in $\mathbf Z$, and/or ii) by neither including the middle node of a collider structure ($V_i \rightarrow V_j \leftarrow V_k$) nor its descendants in $\mathbf{Z}$.
We write $X \indep_{\mathcal{G}(\mathbf{d})} Y \mid \mathbf Z$ to express that two nodes $X$ and $Y$ are $d$-separated by the set of nodes $\mathbf Z$ in SWIG $\mathcal{G}(\mathbf{d})$. Appendix \ref{sec:app-dsep} provides a formal definition of $d$-separation and illustrates path blocking rules in an example that highlights rules frequently applied in the complex structures below.

Proposition 11 of \textcite{richardson_single_2013} shows that the distribution of $\mathbf{V}(\mathbf{d})$ satisfies the Markov property with respect to SWIG $\mathcal{G}(\mathbf d)$ if the data is generated by the underlying SCM. Most importantly for our purposes this implies that the graphical $d$-separation criterion can be applied to find (conditional) independencies among the variables displayed in a SWIG \parencite[][Theorem 12]{richardson_single_2013}. Concretely, let $X$, $Y$, and $\mathbf Z$ be nodes contained in SWIG $\mathcal{G}(\mathbf{d})$, then $X \indep_{\mathcal{G}(\mathbf{d})} Y \mid \mathbf Z, \mathbf{d} \Rightarrow X \indep Y \mid \mathbf Z$. In words, if the union of nodes $\mathbf Z$ and the fixed nodes $\mathbf{d}$ $d$-separate $X$ and $Y$, the corresponding random variables are conditionally independent given variables $\mathbf Z$ alone.%\footnote{The presence of $\mathbf{d}$ only matters for cases where treatment nodes emit more than one arrow like in the structures discussed in Section XXX.}

For example applying $d$-separation to Figure \ref{fig:ex-dag2} uncovers the conditional independence between the untreated potential outcome and the treatment $Y_1(0) \indep D \mid X, U$. This implies mean independence $\E[Y_1(0) \mid X,U, D=1] = \E[Y_1(0) \mid X,U, D = 0]$ that could be used in the standard unconfoundedness identification $ATT = \E[Y_1 \mid D=1] - \E[\E[Y_1 \mid X,U, D = 0] \mid D=1]$ if $U$ would be observable \parencite[e.g.,][]{rosenbaum_central_1983}.\footnote{See \textcite{richardson_single_2013a} for a more detailed introduction to SWIGs in the unconfoundedness setting.} However, $U$ is unobservable and DiD assumes conditional parallel trends \eqref{eq:CPT-X-2x2} instead. A conditional independence like $\Delta Y_1(0) \indep D \mid X$ would imply \eqref{eq:CPT-X-2x2}. However, this conditional independence does not follow from SWIG \ref{fig:ex-dag2} without further restrictions and the introduction of $\Delta$-SWIGs.

\begin{remark}(neutral conditioning \& \textcolor{gray}{color} convention) \label{rem:neutral} \\
In many relevant cases, $d$-separation given $\{\mathbf Z, \mathbf{d}\}$ and therefore independence $X \indep Y \mid \mathbf Z$ still holds after conditioning on neutral variables that we indicate in gray, i.e. we write~$X \indep Y \mid \mathbf Z, \textcolor{gray}{\mathbf N}$ if $\textcolor{gray}{\mathbf N}$ does not open paths already blocked by $\mathbf Z$.
\end{remark}

\subsection{Single World Additive Separability} \label{sec:swas}

DiD acknowledges the presence of time-invariant unobservable confounders but assumes that they enter additively separable in the untreated potential outcome $Y_t(0)$, as reviewed in Appendix \ref{app:general-swas}.
Translated to the SCM in Figure \ref{fig:ex-dag} this means that we restrict the functional form of the structural outcome equations:\footnote{For example, also \textcite{blundell_alternative_2009} and \textcite{bonhomme_back_2025} take this structural equation perspective instead of directly restricting untreated potential outcomes, which is most common in the recent DiD literature.}
    \begin{assumption}[$R_{2x2}^{\alpha}$] (2x2 single world additive separability) \label{ass:swas-2x2} \\
The outcome equations in Figure \ref{fig:ex-dag1} take the form
\begin{align*}
    Y_0&  :=  f_{Y_0}(\UDYX,X,U_{Y_0}) ~~~~ = \underbrace{\alpha(U,X) + g_{Y_0}(X,U_{Y_0})}_{Y_0(0)} & (\ref{eq:scm-2x2-y0}') \\
    Y_1&  :=  f_{Y_1}(\UDYX,X,D,U_{Y_1}) =  \underbrace{\alpha(U,X) + g_{Y_1}(X,U_{Y_1})}_{Y_1(0)} + D \cdot \underbrace{\tau(\UDYX,X,U_{Y_1})}_{Y_1(1) - Y_1(0)}. & (\ref{eq:scm-2x2-y}')
\end{align*}
\end{assumption}
The SCM perspective highlights a peculiarity of the DiD setting. While $U$ enters $Y_t(0)$ only via an additively separable and time-invariant function of time-invariant variables $\alpha(U,X)$, the individual treatment effect that is added in case of treatment $\tau(\UDYX,X,U_{Y_1})$ depends on $U$ in an arbitrary manner.\footnote{Note that in the unconditional case everything collapses to the standard two-way fixed-effects structure with $g_{Y_t}(U_{Y_t}) = \underbrace{\E[g_{Y_t}(U_{Y_t})]}_{\lambda_t} + \underbrace{g_{Y_t}(U_{Y_t}) - \E[g_{Y_t}(U_{Y_t})]}_{\varepsilon_t}$ such that $Y_t(0) = \alpha(U) + \lambda_t + \varepsilon_t$ (see also Appendix \ref{app:swas-lit}).}

\begin{figure}[t]
    \caption{DAG and SWIG for the 2x2 setting under single world additive separability}\label{fig:ex-dag-swas}
    \centering
    % ================= Subfigure (a) =================
    \begin{subfigure}{0.4\textwidth}
        \centering
  \resizebox{0.73\linewidth}{!}{%   
    \begin{tikzpicture}
        \tikzset{line width=1.5pt, ell/.style={draw, fill = white, inner xsep=5pt,inner ysep=5pt, line width=1.5pt}, unobs/.style={ fill = none, inner xsep=5pt,inner ysep=5pt, line width=1.5pt}, swig vsplit={gap=8pt, inner line width right=0.5pt}};

% Nodes
    \node[name=y0, ell, shape=ellipse]{$Y_0$};
    \node[name=y1, ell, shape=ellipse] at ($(y0)+(4cm,0cm)$) {$Y_1$};
    \node[name=d, ell, shape=ellipse] at ($(y1)+(-1cm,-2cm)$) {$D$};
    \node[name=u_dyx, unobs, shape=ellipse] at ($(y1)+(-2cm,-3.5cm)$) {$\UDYX$};
    \node[name=u_y0, unobs, shape=ellipse] at ($(y0)+(0cm,+1.5cm)$) {$U_{Y_0}$};
    \node[name=u_y1, unobs, shape=ellipse] at ($(y1)+(0cm,+1.5cm)$) {$U_{Y_1}$};
    \node[name=x, ell, shape=ellipse] at ($(y0)+(0cm,-2cm)$) {$X$};

% Edges

\draw[->,line width=1.5pt,>=stealth]
(d) edge (y1)
(u_dyx) edge[color = blue] node[pos = 0.9, right] {$+ \alpha$} (y0)
(u_dyx) edge[bend right = 50] (y1)
(u_dyx) edge (d)
(u_y0) edge (y0)
(u_y1) edge (y1)
(u_dyx) edge (x)
(x) edge (y0)
(x) edge (y1)
(x) edge (d);
    \end{tikzpicture}
}
{\small
    \begin{align*}
    Y_0  &\defeq f_{Y_0}(\UDYX,X,U_{Y_0}) \\
    & = \alpha(U,X) + g_{Y_0}(X, U_{Y_0}) \\
    Y_1  &\defeq f_{Y_1}(\UDYX,X,D,U_{Y_1}) \\
    & = \alpha(U,X) + g_{Y_1}(X,U_{Y_1}) + D \cdot \tau_1(\UDYX,X, U_{Y_1})\\
    X, & D,U \text{ like Figure \ref{fig:ex-dag}} 
\end{align*}
}
        \caption{2x2 DAG under Assumption \assref{ass:swas-2x2}}
        \label{fig:2x2-dag-swas}
    \end{subfigure}
    % ================= Subfigure (b) =================
    \begin{subfigure}{0.4\textwidth}
        \centering
  \resizebox{0.8\linewidth}{!}{%   
    \begin{tikzpicture}
        \tikzset{line width=1.5pt, ell/.style={draw, fill = white, inner xsep=5pt,inner ysep=5pt, line width=1.5pt}, unobs/.style={ fill = none, inner xsep=5pt,inner ysep=5pt, line width=1.5pt}, swig vsplit={gap=8pt, inner line width right=0.5pt}};

% Nodes
    \node[name=y0, ell, shape=ellipse]{$Y_0\textcolor{gray}{(0)}$};
    \node[name=y1, ell, shape=ellipse] at ($(y0)+(4cm,0cm)$) {$Y_1(0)$};
    \node[name=d, ell, shape=swig vsplit] at ($(y1)+(-1cm,-2cm)$) {\nodepart{left}{$D$} \nodepart{right}{$0$}};
    \node[name=u_dyx, unobs, shape=ellipse] at ($(y1)+(-2.5cm,-3.5cm)$) {$\UDYX$};
    \node[name=u_y0, unobs, shape=ellipse] at ($(y0)+(0cm,+1.5cm)$) {$U_{Y_0}$};
    \node[name=u_y1, unobs, shape=ellipse] at ($(y1)+(0cm,+1.5cm)$) {$U_{Y_1}$};
    \node[name=x, ell, shape=ellipse] at ($(y0)+(0cm,-2cm)$) {$X$};

% Edges

\draw[->,line width=1.5pt,>=stealth]
(d) edge (y1)
(u_dyx) edge[color = blue]  node[pos = 0.9, right] {$+ \alpha$}  (y0)
(u_dyx) edge[bend right = 50, color = blue] node[pos = 0.9, right] {$+ \alpha$} (y1)
(u_dyx) edge (d)
(u_y0) edge (y0)
(u_y1) edge (y1)
(u_dyx) edge (x)
(x) edge (y0)
(x) edge (y1)
(x) edge (d);
    \end{tikzpicture}
}
{\small
    \begin{align*}
    Y_0\textcolor{gray}{(0)}  &\defeq f_{Y_0}(\UDYX,X,U_{Y_0}) \\
    &= \alpha(U,X) + g_{Y_0}(X,U_{Y_0}) \\
    Y_1(0) &\defeq f_{Y_1}(\UDYX,X,0,U_{Y_1}) \\
    & = \alpha(U,X) + g_{Y_1}(X,U_{Y_1}) \\
    X, & D,U \text{ like Figure \ref{fig:ex-dag}} 
\end{align*}
}
        \caption{2x2 SWIG under Assumption \assref{ass:swas-2x2}}
        \label{fig:2x2-swig-swas}
    \end{subfigure}
\end{figure}
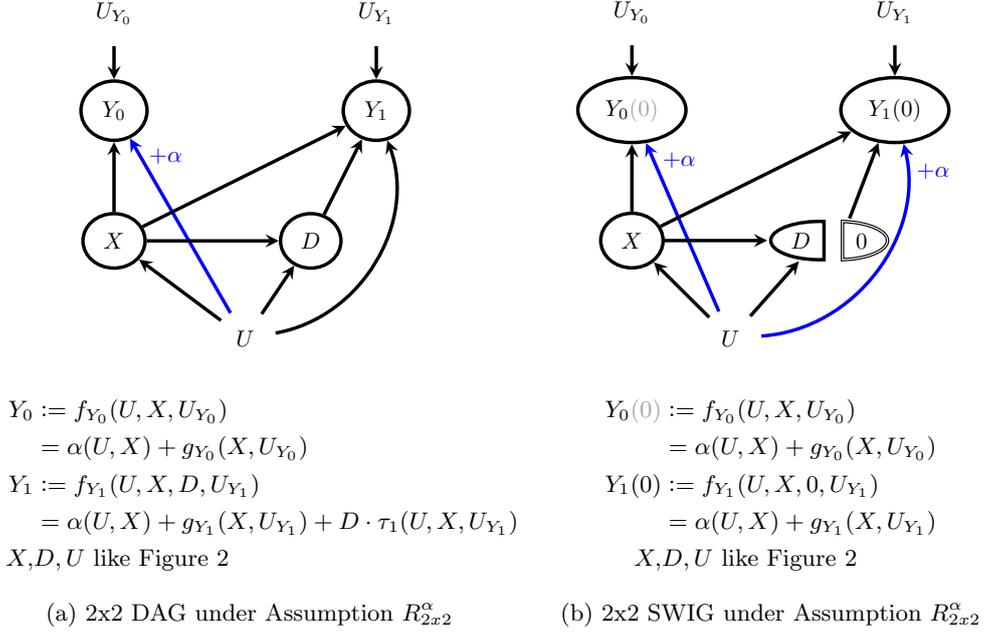

We graphically represent the assumption that $U$ enters $Y_t(0)$ only via the additively separable function by annotating the respective edges with $+\alpha$ in the DAG and SWIG of Figure \ref{fig:ex-dag-swas} that is otherwise identical to Figure \ref{fig:ex-dag}.\footnote{Note that we do not label the edges emitted from $X$ in the same manner, as $X$ enters $Y_t(0)$ via the time-varying functions $g_{Y_t}$ in addition to time-invariant function $\alpha$.}
This highlights that the DAG contains a single labeled edge $U \rightarrow Y_0$, indicating that $U$ enters $Y_0$ only via $\alpha(U,X)$, while the relation $U \rightarrow Y_1$ is nonparametric. It illustrates that $U$ is not additively separable in all \textit{observable} outcomes. However, $U$ enters all \textit{untreated potential} outcomes only via $\alpha(U,X)$ as depicted in Figure \ref{fig:2x2-swig-swas}.
The observation that this structure is implied by standard DiD modeling assumptions motivates the development of $\Delta$-SWIGs in the next section to directly read off \eqref{eq:CPT-X-2x2} from a causal graph.

% \subsection{Color and notation conventions}

% (Need to think aobut whether this sec is needed)

% We introduce the convention that the gray nodes in statements like $X \indep Y \mid \mathbf{Z_f}\textcolor{gray}{,\mathbf{Z_o}}$ are optional meaning that the independence holds in this context once the black nodes $\mathbf{Z_f}$ are conditioned on.

% First PO gray, then blue

\FloatBarrier

\section{The $\Delta$-SWIG} \label{sec:D-SWIG}

\subsection{2x2 Setting with Time-Invariant Control Variables} \label{sec:T2_example}

In this section, we outline how to construct a $\Delta$-SWIG that allows to read off $\Delta Y_1(0) \indep D \mid X$ and consequently \eqref{eq:CPT-X-2x2}. We build on the concrete 2x2 SWIG in Figure \ref{fig:2x2-swig-swas} to convey the main idea in a compact manner before Section \ref{sec:general_steps} provides the general procedure and proves its validity.

Note that Assumption \assref{ass:swas-2x2} implies that the first difference of untreated potential outcomes does not depend on $U$.\footnote{This follows because $\Delta Y_1(0) = f_{Y_1}(\UDYX,X,0,U_{Y_1}) - f_{Y_0}(\UDYX,X,U_{Y_0}) = \cancel{\alpha(\UDYX,X)} + g_{Y_1}(X,U_{Y_1}) - \cancel{\alpha(\UDYX,X)} - g_{Y_0}(X,U_{Y_0}) =: f_{\Delta Y_1(0)} (\cancel{U},X,U_{Y_1},U_{Y_0})$.} This motivates adding a "difference node" to the SWIG that inherits all incoming edges from its two levels \textit{except} the edge from $U$. Figure \ref{fig:D-swig-2x2} illustrates the "$\Delta$-SWIG" that results from augmenting SWIG \ref{fig:2x2-swig-swas} accordingly.

\begin{figure}[h]
    \caption{Full and pruned $\Delta$-SWIGs under \assref{ass:swas-2x2}}\label{fig:ex-d-swig}
    \centering
    % ================= Subfigure (a) =================
    \begin{subfigure}{0.48\textwidth}
        \centering
  \resizebox{0.8\linewidth}{!}{%        
   \begin{tikzpicture}
        \tikzset{line width=1.5pt, ell/.style={draw, fill = white, inner xsep=5pt,inner ysep=5pt, line width=1.5pt}, unobs/.style={ fill = none, inner xsep=5pt,inner ysep=5pt, line width=1.5pt}, swig vsplit={gap=8pt, inner line width right=0.5pt}};

% Nodes
    \node[name=y0, ell, shape=ellipse]{$Y_0\textcolor{gray}{(0)} $};
    \node[name=y1, ell, shape=ellipse] at ($(y0)+(5cm,0cm)$) {$Y_1(0)$};
    \node[name=dy1, ell, shape=ellipse] at ($(y0)+(2.5cm,2cm)$) {$\Delta Y_1(0)$};
    \node[name=d, ell, shape=swig vsplit] at ($(y1)+(-1cm,-2cm)$) {\nodepart{left}{$D$} \nodepart{right}{$0$}};
    \node[name=u_dyx, unobs, shape=ellipse] at ($(y1)+(-2.5cm,-3.5cm)$) {$\UDYX$};
    \node[name=u_y0, unobs, shape=ellipse] at ($(y0)+(-1cm,+1.5cm)$) {$U_{Y_0}$};
    \node[name=u_y1, unobs, shape=ellipse] at ($(y1)+(1cm,+1.5cm)$) {$U_{Y_1}$};
    \node[name=x, ell, shape=ellipse] at ($(y0)+(-1cm,-2cm)$) {$X$};

% Edges

\draw[->,line width=1.5pt,>=stealth]
(u_dyx) edge[color = blue] node[pos = 0.82, right] {$+ \alpha$}  (y0)
(u_dyx) edge[bend right = 50, color = blue] node[pos = 0.8, right] {$+ \alpha$} (y1)
(u_dyx) edge (d)
(u_y0) edge (y0)
(u_y1) edge (y1)
(u_y0) edge (dy1)
(u_y1) edge (dy1)
(u_dyx) edge (x)
(x) edge (y0)
(x) edge (y1)
(x) edge[bend right = 20] (dy1)
(x) edge (d);

\draw[->,line width=1.5pt,>=stealth] (d.40) -- (y1);
\draw[->,line width=1.5pt,>=stealth] (d.70) -- (dy1);
\end{tikzpicture}
}
        \caption{Full $\Delta$-SWIG}
        \label{fig:D-swig-2x2}
    \end{subfigure}
    \hfill
    % ================= Subfigure (b) =================
    \begin{subfigure}{0.48\textwidth}
        \centering
        \resizebox{0.43\linewidth}{!}{%
        \begin{tikzpicture}
        \tikzset{line width=1.5pt, ell/.style={draw, fill = white, inner xsep=5pt,inner ysep=5pt, line width=1.5pt}, unobs/.style={ fill = none, inner xsep=5pt,inner ysep=5pt, line width=1.5pt}, swig vsplit={gap=5pt, inner line width right=0.5pt}};

% Nodes
    \node[name=dy1, ell, shape=ellipse] {$\Delta Y_1(0)$};
    \node[name=d, ell, shape=swig vsplit] at ($(dy1)+(1cm,-4cm)$) {\nodepart{left}{$D$} \nodepart{right}{$0$}};
    \node[name=u_dyx, unobs, shape=ellipse] at ($(dy1)+(0cm,-5.5cm)$) {$\UDYX$};
    \node[name=x, ell, shape=ellipse] at ($(dy1) + (-2cm, -4cm)$) {$X$};

% Edges

\draw[->,line width=1.5pt,>=stealth]
% (d) edge (dy1)
(u_dyx) edge (d)
(u_dyx) edge (x)
(x) edge (d)
(x) edge (dy1);
\draw[->,line width=1.5pt,>=stealth] (d.70) -- (dy1);
\end{tikzpicture}
}
\caption{Pruned $\Delta$-SWIG}
        \label{fig:D-swig-2x2-pruned}
    \end{subfigure}
%     \begin{align*}
%     % Y_1\textcolor{gray}{(0)}  &\defeq f_{Y_0}(\UDYX,X,U_{Y_0}) \\
%     % &= \alpha(\UDYX) + g_{Y_0}(X,U_{Y_0}) \\
%     % Y_1(0) &\defeq f_{Y_1}(\UDYX,X,0,U_{Y_1}) \\
%     % &= \alpha(\UDYX) + g_{Y_1}(X,U_{Y_1}) \\
%     % D   &\defeq f_{D}(\UDYX,X,U_D) \\
%     % X & \defeq f_X (\UDYX,U_X) \\
%     \Delta Y_1(0) & \overset{\assref{ass:swas-2x2}}{=} g_{Y_1}(X,U_{Y_1}) - g_{Y_0}(X,U_{Y_0}) \\
%     & =:  f_{\Delta Y_1(0)} (X,U_{Y_1},U_{Y_0})     
% \end{align*}
\end{figure}
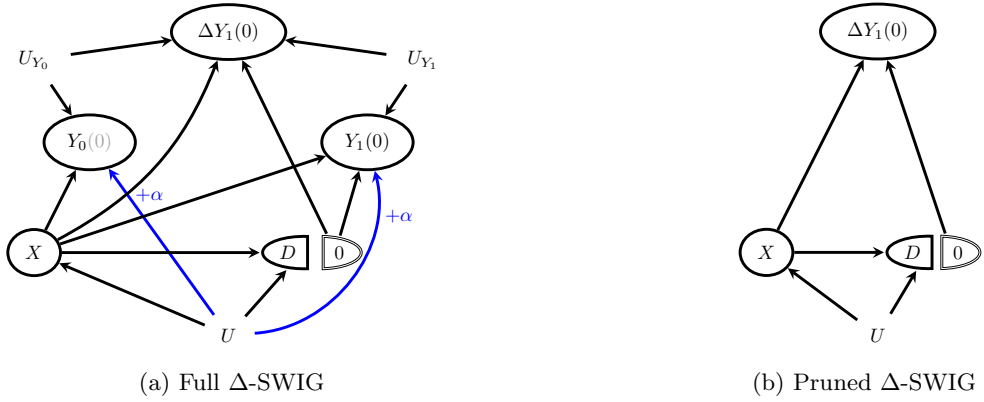

Most importantly, Theorem \ref{thm:Markov} below proves that $d$-separation also implies conditional independencies in such $\Delta$-SWIGs. 
Further, note that all paths between $D$ and $\Delta Y_1(0)$ in Figure \ref{fig:D-swig-2x2} are either blocked by colliders $Y_t(0)$ (e.g.~$D\leftarrow U \rightarrow Y_t(0) \leftarrow U_{Y_t} \rightarrow \Delta Y_1(0)$) and/or by conditioning on $X$ (e.g.~$D \leftarrow U \rightarrow X \rightarrow \Delta Y_1(0)$). Therefore, the $\Delta$-SWIG implies $\Delta Y_1(0) \indep D \mid X$, which implies \eqref{eq:CPT-X-2x2} and therefore the standard identification result $ATT = \E[\Delta Y_1 \mid D=1] - \E[\E[\Delta Y_1 \mid X, D = 0] \mid D=1]$ holds.

As $Y_0(0)$ and $Y_1(0)$ have no descendants in Figure \ref{fig:D-swig-2x2}, we can "prune" the full $\Delta$-SWIG by dropping $Y_0(0)$, $Y_1(0)$ and their incoming edges. Additionally, $U_{Y_0}$ and $U_{Y_1}$ do not need to be explicitly depicted as they only affect one variable in the pruned graph. The resulting pruned $\Delta$-SWIG in Figure \ref{fig:D-swig-2x2-pruned} establishes $\Delta Y_1(0) \indep D \mid X$ in a more compact manner. The benefit of pruning is that paths that are already blocked by colliders $Y_t(0)$ disappear and we can focus on blocking the two confounding paths $D \leftarrow X \rightarrow \Delta Y_1(0)$ and $D \leftarrow U \rightarrow X \rightarrow \Delta Y_1(0)$. 
Theorem \ref{thm:Markov} below shows that the described pruning keeps the validity of $d$-separation intact. Thus, unless the full $\Delta$-SWIG is instructive, we display appropriately pruned $\Delta$-SWIGs by default.

\subsection{General Procedure} \label{sec:general_steps}

Procedure \ref{proc:d-swig} generalizes the transformation steps of the previous section to construct a full or a pruned $\Delta$-SWIG from any SCM:
\begin{procedurebox}[label={proc:d-swig}]{$\Delta$-SWIG Construction}
\begin{enumerate}
    \item Assume SCM $\mathcal{M}$ with endogenous variables $\mathbf{V}$, exogenous variables $\mathbf{U}$, and functions $\mathcal{F}$ that potentially incorporate functional form assumptions.
    \item \emph{(optional)} Draw the corresponding DAG $\mathcal{G}$ with nodes $\mathbf{V}$ and egdes corresponding to the structural relations specified in $\mathcal{M}$.
    \item Draw the SWIG $\mathcal{G}(\mathbf{d})$ as described in Section \ref{sec:swigs} with nodes $\mathbf{V}(\mathbf{d})$. Explicitly include the exogeneous variables of the endogeneous variables that are differenced in step 6 (and 7).
    \item \emph{(optional)} Relabel pre-treatment outcomes as potential outcomes (see Remark \ref{rem:no-anti}).
    \item \emph{(optional)} Label edges that enter only via an additively separable function (see Section \ref{sec:swas}).
    \item Augment $\mathcal{G}(\mathbf{d})$ by a node representing the difference between two of its nodes $\Delta_{j,k}(\mathbf{d}) := V_j(\mathbf{d}) - V_k(\mathbf{d})$. This new difference node has \textbf{no outgoing edges} but \textbf{inherits incoming edges} from level nodes $V_j(\mathbf{d})$ and $V_k(\mathbf{d})$ \textit{unless} their edges exactly cancel due to the assumed functional form in $\mathcal{M}$.\footnote{Appendix \ref{app:edge-cancel} formalizes the underlying mechanism at the level of the structural equations.} The resulting $\Delta$-SWIG is denoted by $\mathcal{G}_\Delta^{full}(\mathbf{d})$.
    \item \emph{(optional)} Repeat step 6 with different level nodes to add additional difference nodes.
    \item \emph{(optional)} Prune those nodes in $\mathbf{V}(\mathbf{d})$ that have no descendants in $\mathcal{G}_\Delta^{full} (\mathbf{d})$ (and are not of interest for conditional independencies) by deleting them together with their incoming edges. The resulting graph is the pruned $\Delta$-SWIG denoted by $\mathcal{G}^{prune}_\Delta (\mathbf{d})$.
\end{enumerate}
\end{procedurebox}

The results provided in the following sections exploit that $\mathcal{G}_\Delta^{full}(\mathbf{d})$ or $\mathcal{G}^{prune}_\Delta(\mathbf{d})$ obtained via Procedure \ref{proc:d-swig} allow to read off conditional independencies involving difference nodes via $d$-separation, as Theorem \ref{thm:Markov} shows under standard assumptions for causal graphs:

\begin{restatable}{theorem}{MarkovTheorem}
\label{thm:Markov}
    Assume the distribution $P(\mathbf{V}(\mathbf{d}))$ satisfies the Markov property \parencite[][or Definition \ref{def:markov} in the Appendix]{kiiveri_recursive_1984} with respect to the subgraph $(\mathcal{G}(\mathbf{d}))_{\mathbf{V}(\mathbf{d})}$ of $\mathcal{G}(\mathbf{d})$ that is obtained by removing all fixed nodes $\mathbf{d}$ \parencite[][section 3.5]{richardson_single_2013}. Let $\mathcal{G}_\Delta (\mathbf{d})$ be a full or pruned $\Delta$-SWIG representing variables $\mathbf{V}_{\Delta}(\mathbf{d})$ obtained via Procedure \ref{proc:d-swig}.
    $\mathbf{X}$, $\mathbf{Y}$, and $\mathbf{Z}$ are disjoint subsets of $\mathbf{V}_{\Delta}(\mathbf{d})$. If $\mathbf{X}$ and $\mathbf{Y}$ are $d$-separated given $\mathbf{Z} , \mathbf{d}$ in $\mathcal{G}_\Delta(\mathbf{d})$, they are conditionally independent given $\mathbf{Z}$, i.e.,
    \begin{equation*}
        \mathbf{X} \indep_{\mathcal{G}_\Delta(\mathbf{d})} \mathbf{Y} \mid \mathbf{Z} , \mathbf{d}  \Rightarrow  \mathbf{X} \indep \mathbf{Y} \mid \mathbf{Z}.
    \end{equation*}

\end{restatable}

The assumption that $P(V(\mathbf{d}))$ satisfies the Markov property with respect to $(\mathcal{G}(\mathbf{d}))_{\mathbf{V}(\mathbf{d})}$ is the standard adaptation of the Markov property to SWIGs and is satisfied if the data is generated by the underlying SCM. %\footnote{For the Markov property see e.g. \textcite{kiiveri_recursive_1984}, \textcite{lauritzen_graphical_1996}, \textcite{peters_elements_2017}, \textcite{ pearl_causality_2009}.}
Remark \ref{rem:subgraph_markov} in Appendix \ref{app:definitions} discusses how the standard Markov property is connected to factorization according to a SWIG in \textcite{richardson_single_2013}, ensuring that $d$-separation implies conditional independence. 

The proof of Theorem \ref{thm:Markov} is provided in Appendix \ref{sec:thm_markov_proof} and exploits that a $\Delta$-SWIG is constructed by adding nodes $\Delta_{j,k}(\mathbf{d})$ that have no descendants. Thus, the variables represented by $\Delta_{j,k}(\mathbf{d})$ are deterministic functions of the variables represented by their parents in  $\mathcal{G}_\Delta^{full}(\mathbf{d})$, which ensures that the probability distribution of the variables $\mathbf{V}_\Delta(\mathbf{d})$ satisfies the Markov property with respect to $\mathcal{G}_\Delta^{full}(\mathbf{d})$.
The pruning step uses that removing nodes with no descendants leaves $d$-separation among the \emph{other} nodes in a graph unchanged, as it can only remove paths containing unconditioned colliders or unconditioned descendants of nodes on other paths. 

While we focus on difference nodes, the same arguments can be applied for other combinations of nodes, e.g.~ratios of nodes to remove multiplicative separable unobservables. We leave such extensions for future research.

\section{2x2 Setting with Time-Varying Covariates} \label{sec:time_varying_T2}

The example in Section \ref{sec:T2_example} illustrates how to deduce conditional parallel trends from a $\Delta$-SWIG in the 2x2 setting with time-invariant observable control variables, denoted $X$.
Now, we introduce time-varying control variables $X_t$ in the 2x2 framework. This allows to demonstrate the application of the $\Delta$-SWIG technology in a relatively compact setting, while discussing implications of outcome dynamics and endogenous control variables for CPT that also occur in multi-period settings.

\subsection{Conditional Parallel Trends} \label{sec:time_varying_T2_standard}

The DAG in Figure \ref{fig:dag-2x2-xt} depicts a setting with time-varying but pre-treatment covariates, $X_0$ and $X_1$.\footnote{Time-invariant covariates can be subsumed in $X_0$ because an explicit node $X$ would emit the same arrows as node $X_0$.} 
The SCM spelled out below the figure incorporates an adapted single world additive separability assumption. Therefore, the $\Delta$-SWIG in Figure \ref{fig:d-swig-2x2-xt} resulting from Procedure \ref{proc:d-swig} contains no edge $\UDYX \rightarrow \Delta Y_1(0)$.

\begin{figure}[t]
    \centering
    \caption{2x2 DAG and $\Delta$-SWIG with time-varying control variables}
    \label{fig:DAG_DSWIG_time_varying_T2}
\begin{subfigure}[t]{0.48\textwidth}
\centering
\makebox[\linewidth][c]{%
  \resizebox{0.7\linewidth}{!}{%
        \begin{tikzpicture}
        \tikzset{line width=1.5pt, ell/.style={draw, fill = white, inner xsep=5pt,inner ysep=5pt, line width=1.5pt}, unobs/.style={ fill = none, inner xsep=5pt,inner ysep=5pt, line width=1.5pt}, swig vsplit={gap=5pt, inner line width right=0.5pt}};

% Nodes
    \node[name=y0, ell, shape=ellipse]{$Y_0$};
    \node[name=y1, ell, shape=ellipse] at ($(y0)+(4cm,0cm)$) {$Y_1$};
    % \node[name=d, ell, shape=swig vsplit] at ($(y2)+(-0.5cm,-1.5cm)$) {\nodepart{left}{$D$} \nodepart{right}{$0$}};
    \node[name=d, ell, shape=ellipse] at ($(y1)+(-0.5cm,-1.5cm)$) {$D$};
    \node[name=x0, ell, shape=ellipse] at ($(y0)+(-1cm,-3cm)$) {$X_0$};
    \node[name=x1, ell, shape=ellipse] at ($(y1)+(-1cm,-3cm)$) {$X_1$};
    \node[name=u_dyx, unobs, shape=ellipse] at ($(y1)+(-2.5cm,-4.5cm)$) {$\UDYX$};
    \node[name=u_y0, unobs, shape=ellipse] at ($(y0)+(0cm,+1.5cm)$) {$U_{Y_0}$};
    \node[name=u_y1, unobs, shape=ellipse] at ($(y1)+(0cm,+1.5cm)$) {$U_{Y_1}$};

% Edges

\draw[->,line width=1.5pt,>=stealth]
(d) edge (y1)
% (u_dyx) edge[color = blue] node[midway, right] {$+ \alpha$} (y0)
% (u_dyx) edge[bend right = 60,color = blue] node[midway, right] {$+ \alpha$} (y1)
(u_dyx) edge node[pos = 0.9, right] {$+ \alpha$} (y0)
(u_dyx) edge[bend right = 60] (y1)
(u_dyx) edge[bend left = 20] (d)
(x0) edge (y0)
(x0) edge (y1)
(x1) edge[bend right= 35] (y1)
(x0) edge (x1)
(x0) edge (d)
% (x1) edge (d.left south)
(x1) edge (d)
(u_dyx) edge (x0)
(u_dyx) edge (x1)
(u_y0) edge (y0)
(u_y1) edge (y1);
    \end{tikzpicture}
    }
    }
{\small
    \begin{align*}
    Y_0  &\defeq f_{Y_0}(\UDYX,X_0,U_{Y_0}) = \alpha(\UDYX,X_0) + g_{Y_0}(X_0,U_{Y_0}) \\
    Y_1 &\defeq f_{Y_1}(\UDYX,X_0,X_1,D,U_{Y_1})  \\
    & =\alpha(\UDYX,X_0) + g_{Y_1}(X_0,X_1,U_{Y_1}) \\
    & \quad + D \cdot \tau_1(\UDYX, X_0,X_1, U_{Y_1}) 
    \end{align*} 
}
\caption{2x2 DAG with $X_t$}
        \label{fig:dag-2x2-xt}
\end{subfigure}
\hfill
\begin{subfigure}[t]{0.48\textwidth}
\centering
\makebox[\linewidth][c]{%
  \resizebox{0.7\linewidth}{!}{%
        \begin{tikzpicture}
        \tikzset{line width=1.5pt, ell/.style={draw, fill = white, inner xsep=5pt,inner ysep=5pt, line width=1.5pt}, unobs/.style={ fill = none, inner xsep=5pt,inner ysep=5pt, line width=1.5pt}, swig vsplit={gap=5pt, inner line width right=0.5pt}};

% Nodes
    \node[name=y0, ell, shape=ellipse]{$Y_0$};
    \node[name=dy1, ell, shape=ellipse] at ($(y0)+(4cm,0cm)$) {$\Delta Y_1(0)$};
    \node[name=d, ell, shape=swig vsplit] at ($(dy1)+(-0.5cm,-1.5cm)$) {\nodepart{left}{$D$} \nodepart{right}{$0$}};
    \node[name=x0, ell, shape=ellipse] at ($(y0)+(-1cm,-3cm)$) {$X_0$};
    \node[name=x1, ell, shape=ellipse] at ($(dy1)+(-1cm,-3cm)$) {$X_1$};
    \node[name=u_dyx, unobs, shape=ellipse] at ($(dy1)+(-2.5cm,-4.5cm)$) {$\UDYX$};
    \node[name=u_y1, unobs, shape=ellipse] at ($(y0)+(0cm,+1.5cm)$) {$U_{Y_0}$};

% Edges

\draw[->,line width=1.5pt,>=stealth]
(d) edge (dy1)
(u_dyx) edge[color = orange] node[pos = 0.9, right] {$+ \alpha$} (y0)
(u_dyx) edge[bend left = 20, color = orange]  (d)
(x0) edge (y0)
(x1) edge[bend right= 60] (dy1)
(x0) edge (x1)
(x0) edge (dy1)
(x0) edge (d)
(x1) edge (d.240)
(u_dyx) edge (x0)
(u_dyx) edge (x1)
(u_y1) edge[color = orange] (y0)
(u_y1) edge[color = orange] (dy1);
    \end{tikzpicture}
    }
    }
{\small
\begin{align*} 
   Y_0 \textcolor{gray}{(0)} &\defeq f_{Y_0}(\UDYX,X_0,U_{Y_0})= \alpha(\UDYX,X_0) + g_{Y_0}(X_0,U_{Y_0}) \\
    \Delta Y_1(0) &\defeq f_{\Delta Y_1(0)}(X_0,X_1,U_{Y_0},U_{Y_1})  \\
    & \defeq g_{Y_1}(X_0, X_1,U_{Y_1}) - g_{Y_0}(X_0,U_{Y_0}) \\
\end{align*}
}
\caption{2x2 $\Delta$-SWIG with $X_t$}
        \label{fig:d-swig-2x2-xt}
\end{subfigure}
\begin{minipage}{\linewidth}
\begin{small}
\textit{Note:} The intermediate SWIG and all structural equations are provided in Figure \ref{fig:SWIG_time_varying_T2} in the Appendix.    
\end{small}
\end{minipage}
\end{figure}
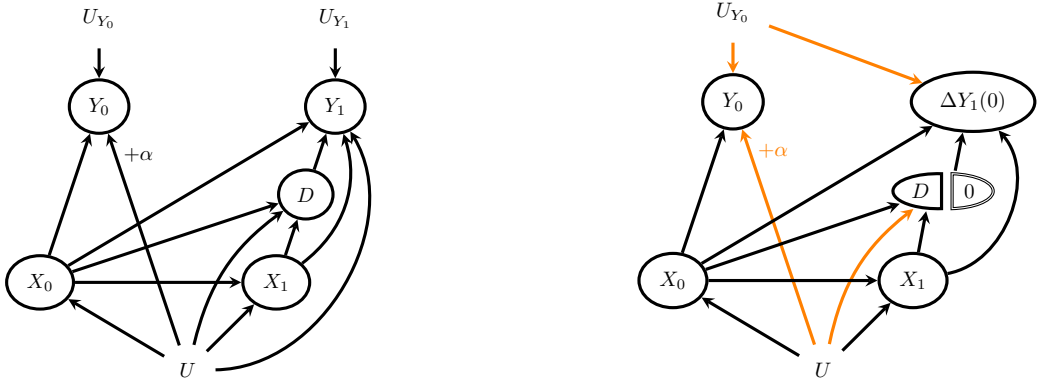

In this $\Delta$-SWIG, all paths between $D$ and $\Delta Y_1(0)$ are blocked given $X_0,X_1$.\footnote{For $t = 0,1$, the paths $D \leftarrow X_t \rightarrow \Delta Y_1 (0)$ and $D \leftarrow \UDYX \rightarrow X_t \rightarrow \Delta Y_1(0)$ are all blocked by conditioning on $X_t$. Additionally, the paths $D \leftarrow U (\rightarrow X_0) \rightarrow Y_0 \leftarrow U_{Y_1} \rightarrow \Delta Y_1(0)$ are blocked by collider $Y_0$.} Therefore, they are $d$-separated and we can deduce $\Delta Y_1(0) \indep D \mid X_0,X_1$ following Theorem \ref{thm:Markov}, which implies conditional parallel trends $\E[\Delta Y_1(0) \mid X_0, X_1, D = 0] = \E[\Delta Y_1(0) \mid X_0, X_1, D = 1]$ that enable standard identification of the $ATT$.

$\Delta$-SWIG \ref{fig:d-swig-2x2-xt} explicitly includes node $Y_0$, although it does not have descendants and therefore could be pruned from the graph. However, it is instructive to revisit graphically why pre-treatment outcomes are considered bad controls although they are measured before the treatment \parencite[see e.g.][]{daw_matching_2018,chabe-ferret_should_2025}.
Conditioning on collider $Y_0$ would open the path $D \leftarrow \UDYX \rightarrow Y_0 \leftarrow U_{Y_0} \rightarrow \Delta Y_1(0)$ and precludes reading off parallel trends conditional on $Y_0$. \textcite{cinelli_crash_2024} refer to this type of bad control in the unconfoundedness setting as inducing M-bias.

%%%%%%%%%%%%%%%%%%%%%%%%%%%%% Outcome Dynamics%%%%%%%%%%%%%%%%%%%%%%%%%%%%%%%%%%%%%%%%%%%%%%%%%%%%%%%%%%%%

\subsection{No Parallel Trends with Outcome Dynamics} \label{sec:2x2-ydyn}

Next, we graphically revisit why outcome dynamics are usually not compatible with CPT. Figure \ref{fig:dag-2x2-dyn} shows the DAG with the three possible instances of outcome dynamics depicted as dotted edges: state dependence ($Y_0 \rightarrow Y_1$), outcome-treatment feedback  ($Y_0 \rightarrow D$), and outcome-covariate feedback  ($Y_0 \rightarrow X_1$). We discuss them in separate $\Delta$-SWIGs for clarity, but the same conclusions follow from one $\Delta$-SWIG with all dynamics.\footnote{Note that $Y_0$ could not be pruned from the following $\Delta$-SWIGs because it now has a descendant. In contrast, level node $Y_1(0)$ can be pruned and is pruned.}

\textbf{State-dependence:}
Following Procedure \ref{proc:d-swig}, the difference node in $\Delta$-SWIG \ref{fig:d-swig-y-y} now inherits an arrow $Y_0 \rightarrow \Delta Y_1(0)$. This creates open path $D \leftarrow \UDYX \rightarrow Y_0 \rightarrow \Delta Y_1(0)$ that could be blocked by conditioning on $Y_0$. However, doing so opens path $D \leftarrow \UDYX \rightarrow Y_0 \leftarrow U_{Y_0} \rightarrow \Delta Y_1(0)$ and creates M-Bias, as discussed in the previous section. Therefore, $Y_0$ is a ``dilemma node'' that can block only one of two paths that can not be blocked otherwise. Consequently, $D$ and $\Delta Y_1(0)$ can not be $d$-separated and CPT can not be established in this structure \parencite[see e.g.][for a similar discussion]{renson_using_2026}.

\textbf{Outcome-treatment feedback:} 
$\Delta$-SWIG \ref{fig:d-swig-y-d} displays selection on pre-treatment outcomes. The $\Delta$-SWIG now contains the open path $D \leftarrow Y_0 \leftarrow U_{Y_0} \rightarrow \Delta Y_1(0)$. $Y_0$ is again a dilemma node. Conditioning on it would block this additional path while opening path $D \leftarrow \UDYX \rightarrow Y_0 \leftarrow U_{Y_0} \rightarrow \Delta Y_1(0)$. Again, $d$-separation can not be established in this setting, which is in line with previous results showing that selection on pre-treatment outcomes is in conflict with parallel trends \parencite[see e.g.][]{ashenfelter_using_1985,bonhomme_back_2025,marx_parallel_2024, renson_using_2026} unless a particular martingale condition holds, as discussed in \textcite{ghanem_selection_2024}.

\textbf{Outcome-covariate feedback:} 
$X_1$ is the dilemma node in $\Delta$-SWIG \ref{fig:d-swig-y-x}. Conditioning on $X_1$ is required to block the path $D \leftarrow X_1 \rightarrow \Delta Y_1(0)$. However, this opens path $D \leftarrow \UDYX \rightarrow Y_0 \leftarrow U_{Y_0} \rightarrow \Delta Y_1(0)$ because $X_1$ is a descendant of collider $Y_0$. Again, there is no set of observable variables that $d$-separates $D$ and $\Delta Y_1(0)$ \parencite[see][for a similar discussion using DAGs]{weber_assumption_2015}. 

\begin{figure}[t]
    \centering
    \caption{DAGs and $\Delta$-SWIGs with outcome dynamics, T = 2}
    \label{fig:DAG_DSWIG_outcome_dyn}
\begin{subfigure}[b]{0.48\textwidth}
\centering
\makebox[\linewidth][c]{%
  \resizebox{0.7\linewidth}{!}{%
        \begin{tikzpicture}
        \tikzset{line width=1.5pt, ell/.style={draw, fill = white, inner xsep=5pt,inner ysep=5pt, line width=1.5pt}, unobs/.style={ fill = none, inner xsep=5pt,inner ysep=5pt, line width=1.5pt}, swig vsplit={gap=5pt, inner line width right=0.5pt}};

% Nodes
    \node[name=y0, ell, shape=ellipse]{$Y_0$};
    \node[name=y1, ell, shape=ellipse] at ($(y0)+(4cm,0cm)$) {$Y_1$};
    % \node[name=d, ell, shape=swig vsplit] at ($(y2)+(-0.5cm,-1.5cm)$) {\nodepart{left}{$D$} \nodepart{right}{$0$}};
    \node[name=d, ell, shape=ellipse] at ($(y1)+(-0.5cm,-1.5cm)$) {$D$};
    \node[name=x0, ell, shape=ellipse] at ($(y0)+(-1cm,-3cm)$) {$X_0$};
    \node[name=x1, ell, shape=ellipse] at ($(y1)+(-1cm,-3cm)$) {$X_1$};
    \node[name=u_dyx, unobs, shape=ellipse] at ($(y1)+(-2.5cm,-4.5cm)$) {$\UDYX$};
    \node[name=u_y0, unobs, shape=ellipse] at ($(y0)+(0cm,+1.5cm)$) {$U_{Y_0}$};
    \node[name=u_y1, unobs, shape=ellipse] at ($(y1)+(0cm,+1.5cm)$) {$U_{Y_1}$};

% Edges

\draw[->,line width=1.5pt,>=stealth]
(d) edge (y1)
% (u_dyx) edge[color = blue] node[midway, right] {$+ \alpha$} (y0)
% (u_dyx) edge[bend right = 60,color = blue] node[midway, right] {$+ \alpha$} (y1)
(u_dyx) edge node[pos = 0.75, right] {$+ \alpha$} (y0)
(u_dyx) edge[bend right = 55] (y1)
(u_dyx) edge[bend left = 30] (d.230)
(x0) edge (y0)
(x0) edge (y1)
(x1) edge[bend right= 40] (y1)
(x0) edge (x1)
(x0) edge (d)
% (x1) edge (d.left south)
(x1) edge (d)
(u_dyx) edge (x0)
(u_dyx) edge (x1)
(u_y0) edge (y0)
(u_y1) edge (y1)
(y0) edge[dotted]  (y1)
(y0) edge[dotted]  (d)
(y0) edge[dotted]  (x1);
    \end{tikzpicture}
    }
    }
% \begin{minipage}{\linewidth}
% %
% {\small
% \centering (\ref{eq:T2_basic_Y0}) \& (\ref{eq:T2_basic_D}) \& (\ref{eq:T2_basic_X0}) \& (\ref{eq:T2_basic_X1})
% \begin{align*}
% \begin{split}
%     Y_1 &\defeq f_{Y_1}(\UDYX,X_0,X_1,Y_0,D,U_{Y_1}) \\
%     & =\alpha(\UDYX) + g_{Y_1}(X_0,X_1,Y_0,U_{Y_1}) \\
%     & \qquad + D \cdot \tau_1(\UDYX, X_0,X_1,Y_0, U_{Y_1})
%     \end{split} \\
% \end{align*}
% }
% \end{minipage}
\caption{2x2 DAG with outcome dynamics}
        \label{fig:dag-2x2-dyn}
\end{subfigure}
\hfill
\begin{subfigure}[b]{0.48\textwidth}
\centering
\makebox[\linewidth][c]{%
  \resizebox{0.7\linewidth}{!}{%
        \begin{tikzpicture}
        \tikzset{line width=1.5pt, ell/.style={draw, fill = white, inner xsep=5pt,inner ysep=5pt, line width=1.5pt}, unobs/.style={ fill = none, inner xsep=5pt,inner ysep=5pt, line width=1.5pt}, swig vsplit={gap=5pt, inner line width right=0.5pt}};

% Nodes
    \node[name=y0, ell, shape=ellipse]{$Y_0$};
    \node[name=dy1, ell, shape=ellipse] at ($(y0)+(4cm,0cm)$) {$\Delta Y_1(0)$};
    \node[name=d, ell, shape=swig vsplit] at ($(dy1)+(-0.5cm,-1.5cm)$) {\nodepart{left}{$D$} \nodepart{right}{$0$}};
    \node[name=x0, ell, shape=ellipse] at ($(y0)+(-1cm,-3cm)$) {$X_0$};
    \node[name=x1, ell, shape=ellipse] at ($(dy1)+(-1cm,-3cm)$) {$X_1$};
    \node[name=u_dyx, unobs, shape=ellipse] at ($(dy1)+(-2.5cm,-4.5cm)$) {$\UDYX$};
    \node[name=u_y1, unobs, shape=ellipse] at ($(y0)+(0cm,+1.5cm)$) {$U_{Y_0}$};

% Edges

\draw[->,line width=1.5pt,>=stealth]
(d) edge (dy1)
% (u_dyx) edge node[pos = 0.9, right] {$+ \alpha$} (y0)
% (u_dyx) edge[bend left = 30]  (d)
(u_dyx) edge[color = orange] node[pos = 0.9, right] {$+ \alpha$} (y0)
(u_dyx) edge[bend left = 27, color = orange]  (d.220)
(x0) edge (y0)
(x1) edge[bend right= 60] (dy1)
(x0) edge (x1)
(x0) edge (dy1)
(x0) edge (d)
(x1) edge (d)
(u_dyx) edge (x0)
(u_dyx) edge (x1)
% (u_y1) edge (y0)
% (u_y1) edge (dy1)
(u_y1) edge[color = orange] (y0)
(u_y1) edge[color = orange] (dy1)
(y0) edge[color = blue] (dy1);

    \end{tikzpicture}
    }
    }
% \begin{minipage}{\linewidth}
% %
% {\small
%    \centering (\ref{eq:T2_basic_Y0}) \& (\ref{eq:T2_basic_D}) \& (\ref{eq:T2_basic_X0}) \& (\ref{eq:T2_basic_X1})
%     \begin{align*}
%     \Delta Y_1(0) &\defeq f_{\Delta Y_1(0)}(X_0,X_1,Y_0,U_{Y_0},U_{Y_1}) \\
%     & \defeq g_{Y_1}(X_0,X_1,Y_0,U_{Y_1}) - g_{Y_0}(X_0,U_{Y_0})
% \end{align*}
% }
% \vspace{0.5cm}
% \end{minipage}
\caption{$\Delta$-SWIG with state dependence}
        \label{fig:d-swig-y-y}
\end{subfigure}
\begin{subfigure}[b]{0.48\textwidth}
\centering
\makebox[\linewidth][c]{%
  \resizebox{0.7\linewidth}{!}{%
        \begin{tikzpicture}
        \tikzset{line width=1.5pt, ell/.style={draw, fill = white, inner xsep=5pt,inner ysep=5pt, line width=1.5pt}, unobs/.style={ fill = none, inner xsep=5pt,inner ysep=5pt, line width=1.5pt}, swig vsplit={gap=5pt, inner line width right=0.5pt}};

% Nodes
    \node[name=y0, ell, shape=ellipse]{$Y_0$};
    \node[name=dy1, ell, shape=ellipse] at ($(y0)+(4cm,0cm)$) {$\Delta Y_1(0)$};
    \node[name=d, ell, shape=swig vsplit] at ($(dy1)+(-0.5cm,-1.5cm)$) {\nodepart{left}{$D$} \nodepart{right}{$0$}};
    \node[name=x0, ell, shape=ellipse] at ($(y0)+(-1cm,-3cm)$) {$X_0$};
    \node[name=x1, ell, shape=ellipse] at ($(dy1)+(-1cm,-3cm)$) {$X_1$};
    \node[name=u_dyx, unobs, shape=ellipse] at ($(dy1)+(-2.5cm,-4.5cm)$) {$\UDYX$};
    \node[name=u_y1, unobs, shape=ellipse] at ($(y0)+(0cm,+1.5cm)$) {$U_{Y_0}$};

% Edges

\draw[->,line width=1.5pt,>=stealth]
(d) edge (dy1)
% (u_dyx) edge node[pos = 0.9, right] {$+ \alpha$} (y0)
% (u_dyx) edge[bend left = 30]  (d)
(u_dyx) edge[color = orange] node[pos = 0.9, right] {$+ \alpha$} (y0)
(u_dyx) edge[bend left = 27, color = orange]  (d.220)
(x0) edge (y0)
(x1) edge[bend right= 60] (dy1)
(x0) edge (x1)
(x0) edge (dy1)
(x0) edge (d)
(x1) edge (d)
(u_dyx) edge (x0)
(u_dyx) edge (x1)
% (u_y1) edge (y0)
% (u_y1) edge (dy1)
(u_y1) edge[color = orange] (y0)
(u_y1) edge[color = orange] (dy1)
(y0) edge[color = blue] (d);

    \end{tikzpicture}
    }
    }
% \begin{minipage}{\linewidth}
% %
% {\small
%  \centering   (\ref{eq:T2_basic_Y0})  \& (\ref{eq:T2_basic_X0}) \& (\ref{eq:T2_basic_X1}) \& (\ref{eq:T2_basic_DY1}) 
% \begin{align*}
%     D   &\defeq f_{D}(X_0,X_1,Y_0,\UDYX,U_D)
% \end{align*}
% }
% \end{minipage}
\caption{$\Delta$-SWIG with outcome-treatment feedback}
        \label{fig:d-swig-y-d}
\end{subfigure}
\hfill
\begin{subfigure}[b]{0.48\textwidth}
\centering
\makebox[\linewidth][c]{%
  \resizebox{0.7\linewidth}{!}{%
        \begin{tikzpicture}
        \tikzset{line width=1.5pt, ell/.style={draw, fill = white, inner xsep=5pt,inner ysep=5pt, line width=1.5pt}, unobs/.style={ fill = none, inner xsep=5pt,inner ysep=5pt, line width=1.5pt}, swig vsplit={gap=5pt, inner line width right=0.5pt}};

% Nodes
    \node[name=y0, ell, shape=ellipse]{$Y_0$};
    \node[name=dy1, ell, shape=ellipse] at ($(y0)+(4cm,0cm)$) {$\Delta Y_1(0)$};
    \node[name=d, ell, shape=swig vsplit] at ($(dy1)+(-0.5cm,-1.5cm)$) {\nodepart{left}{$D$} \nodepart{right}{$0$}};
    \node[name=x0, ell, shape=ellipse] at ($(y0)+(-1cm,-3cm)$) {$X_0$};
    \node[name=x1, ell, shape=ellipse] at ($(dy1)+(-1cm,-3cm)$) {$X_1$};
    \node[name=u_dyx, unobs, shape=ellipse] at ($(dy1)+(-2.5cm,-4.5cm)$) {$\UDYX$};
    \node[name=u_y1, unobs, shape=ellipse] at ($(y0)+(0cm,+1.5cm)$) {$U_{Y_0}$};

% Edges

\draw[->,line width=1.5pt,>=stealth]
(d) edge (dy1)
% (u_dyx) edge node[pos = 0.8, right] {$+ \alpha$} (y0)
% (u_dyx) edge[bend left = 30]  (d)
(u_dyx) edge[color = orange] node[pos = 0.8, right] {$+ \alpha$} (y0)
(u_dyx) edge[bend left = 27, color = orange]  (d.220)
(x0) edge (y0)
(x1) edge[bend right= 60] (dy1)
(x0) edge (x1)
(x0) edge (dy1)
(x0) edge (d)
(x1) edge (d)
(u_dyx) edge (x0)
(u_dyx) edge (x1)
% (u_y1) edge (y0)
% (u_y1) edge (dy1)
(u_y1) edge[color = orange] (y0)
(u_y1) edge[color = orange] (dy1)
(y0) edge[color = blue] (x1);

    \end{tikzpicture}
    }
    }
% \begin{minipage}{\linewidth}
% %
% {\small
%    \centering (\ref{eq:T2_basic_Y0})  \& (\ref{eq:T2_basic_D}) \& (\ref{eq:T2_basic_X0}) \& (\ref{eq:T2_basic_DY1})
%     \begin{align*}
% X_1  &\defeq f_{X_1}(X_0,Y_0,\UDYX,U_{X_1}) 
% \end{align*}
% }
% \end{minipage}
\caption{$\Delta$-SWIG with outcome-covariate feedback}
        \label{fig:d-swig-y-x}
\end{subfigure}
\begin{minipage}{\linewidth}
\begin{small}
\textit{Note:} The intermediate SWIG and all structural equations are provided in Figure \ref{fig:SWIG_outcome_dyn} in the Appendix.
\end{small}
\end{minipage}
\end{figure}
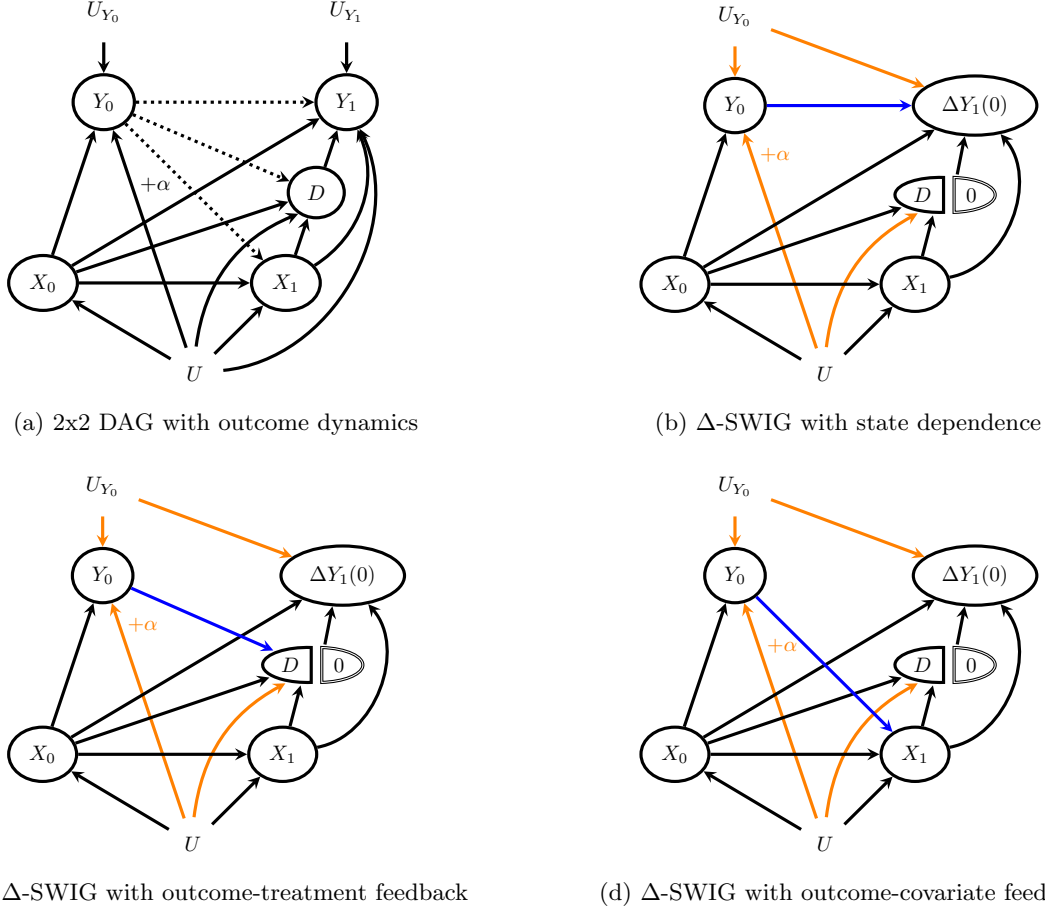

\textbf{Summary:} The presence of path $D \leftarrow U \rightarrow Y_0 \leftarrow U_{Y_0} \rightarrow \Delta Y_1(0)$ prevents $d$-separation if collider node $Y_0$ emits an arrow because it (i) either creates an additional chain where $Y_0$ is the only observable variable that could block it, making itself a dilemma node, or (ii) it becomes a parent of an otherwise crucial control variable, making its descendant $X_1$ a dilemma node. This prevents the deduction of CPT in 2x2 settings with outcome dynamics but also in multi-period settings without imposing additional restrictions. Therefore, we abstract from outcome dynamics in the following.

\begin{remark} \label{rem:strict-exog}
Note that outcome-treatment feedback and outcome-covariate feedback are similar to a violation of the \emph{strict exogeneity} assumption in the panel data literature. In \textcite{wooldridge_econometric_2010}, the strict exogeneity assumption is stated in different ways, e.g., $Cov(U_{Y_s},X_t) = 0$ ($Cov(U_{Y_s},D) = 0$) for all $s,t$. A sufficient graphical condition for this is unconditional $d$-separation of all $U_{Y_s}$ and all $X_t$ ($D$), which is not fulfilled in Figures \ref{fig:d-swig-y-d} and \ref{fig:d-swig-y-x} due to the direct edge from $Y_0$ to $X_1$ ($D$). Note that the standard DAG \ref{fig:dag-2x2-dyn} suffices to see this point \parencite[see e.g.][]{imai_when_2019, bonhomme_back_2025} but it is also visible in $\Delta$-SWIGs \ref{fig:d-swig-y-d} and \ref{fig:d-swig-y-x}. 
\end{remark}

\subsection{Post-Treatment Covariates} \label{sec:post-treatment}

Next, we discuss the role of post-treatment covariates in the 2x2 setting as considered in \textcite{caetano_difference_2024}. 
A corresponding DAG and $\Delta$-SWIG is depicted in Figure \ref{fig:DAG_DSWIG_post}. In contrast to Figure \ref{fig:DAG_DSWIG_outcome_dyn}, outcome dynamics are absent and $X_1$ is now affected by the treatment ($D \rightarrow X_1$).\footnote{Time-invariant covariates denoted $Z$ in \textcite{caetano_difference_2024} are subsumed into $X_0$. Drawing the $Z$ node explicitly would not change the discussion.}
Therefore, the $\Delta$-SWIG contains untreated potential covariate $X_1(0)$. $\Delta$-SWIG \ref{fig:d-swig-2x2-d-x} implies independence
\begin{equation}
    \Delta Y_1(0) \indep D \mid X_0,X_1(0)
    \label{eq:ccpr-cpt}
\end{equation}
and CPT $\E[\Delta Y_1(0) \mid X_0,X_1(0), D = 0] = \E[\Delta Y_1(0) \mid X_0,X_1(0), D = 1]$, which corresponds to Assumption 2 in \textcite{caetano_difference_2024}. However, $X_1(0)$ is unobservable for the treated units and identification based on this CPT is not feasible.
Thus, \textcite{caetano_difference_2024} discuss additional covariate unconfoundedness assumptions of the form
\begin{align}
    X_1(0) \indep D & \mid X_0 \label{eq:ccpr-cu1} \\
    X_1(0) \indep D & \mid X_0, Y_0 \label{eq:ccpr-cu2}    
\end{align}
in their Corollary 1. Graphically, both \eqref{eq:ccpr-cu1} and \eqref{eq:ccpr-cu2} are implied by $\Delta$-SWIG \ref{fig:d-swig-2x2-d-x} after removing the edge $\UDYX \rightarrow X_1(0)$, i.e.~assuming that $X_1$ is only indirectly affected by $U$ via $X_0$.\footnote{This represents the minimal modification to achieve \eqref{eq:ccpr-cu1} and \eqref{eq:ccpr-cu2}. Additionally, also $\UDYX \rightarrow X_0$ could be removed without changing the result. Alternatively, $\UDYX \rightarrow D$ could be removed, while keeping the $\UDYX \rightarrow X_t$ edges. This would create a standard unconfoundedness scenario where $Y_1(0) \indep D \mid X_0$ and a CPT is not required for identification.}\footnote{To see why \eqref{eq:ccpr-cu2} holds, note that path $D \leftarrow U \rightarrow Y_0 \leftarrow U_{Y_1} \rightarrow \Delta Y_1(0) \leftarrow X_1(0)$ is blocked by collider $\Delta Y_1(0)$ even conditional on the other collider $Y_0$.} Then, conditional independencies \eqref{eq:ccpr-cpt} and \eqref{eq:ccpr-cu1} collapse to $\Delta Y_1(0) \indep D \mid X_0$,\footnote{To see why, note that conditional independencies \eqref{eq:ccpr-cpt} and \eqref{eq:ccpr-cu1} imply $(X_1(0),\Delta Y_1(0)) \indep D \mid X_0$ by contraction and $\Delta Y_1(0) \indep D \mid X_0$ by decomposition.} which justifies standard DiD identification using CPT with pre-treatment covariates only, i.e.~$\E[\Delta Y_1(0) \mid X_0, D = 0] = \E[\Delta Y_1(0) \mid X_0, D = 1]$. This observation provides an alternative way to prove Corollary 1 (1) of \textcite{caetano_difference_2024}.

Additionally, \textcite{caetano_difference_2024} show in Corollary 1 (2) that combining \eqref{eq:ccpr-cpt} and \eqref{eq:ccpr-cu2} permits the non-standard identification result involving nested conditional expectations $ATT = \E[\Delta Y_1 | D=1] - \E[\E[\E[\Delta Y_1 \mid X_0, X_1,D = 0]  \mid X_0, Y_0, D = 0] \mid D=1]$. This is conceptually interesting because it demonstrates that identification is possible through a particular conditioning strategy involving $Y_0$, even though CPT conditional on $Y_0$ is not attainable (see Section \ref{sec:time_varying_T2_standard}). However, the practical relevance of this identification result appears limited. Appendix \ref{app:post-treatment} shows that there exists no $\Delta$-SWIG with unobserved confounding in which \eqref{eq:ccpr-cpt} and \eqref{eq:ccpr-cu2} hold but \eqref{eq:ccpr-cpt} and \eqref{eq:ccpr-cu1} do not hold. Thus, there is at least no graphical motivation to prefer the alternative identification result over standard conditional DiD identification. 

\textcite{caetano_difference_2024} also consider a covariate exogeneity assumption $(X_1(0) \mid X_0, D=1) \sim (X_1(0) \mid X_0, D=1)$ that is combined with CPT \eqref{eq:ccpr-cpt} to obtain additional identification results. A sufficient condition for covariate exogeneity is to remove edge $D \rightarrow X_1$ in DAG \ref{fig:dag-2x2-d-x}. However, this would be stronger than required. Covariate exogeneity can therefore not be represented in a $\Delta$-SWIG without modifications that are beyond the scope of the paper. This shows that $\Delta$-SWIGs are not compatible with all identifying assumptions that appear in the DiD literature. In particular, distributional assumptions like covariate exogeneity or parametric functional form assumptions are not naturally captured by $\Delta$-SWIGs. However, in many cases we expect that $\Delta$-SWIGs are applicable to justify at least parts of the identification arguments. For example, independence \eqref{eq:ccpr-cpt} can be justified by $\Delta$-SWIG \ref{fig:d-swig-2x2-d-x} but covariate exogeneity requires arguments going beyond the causal structure.

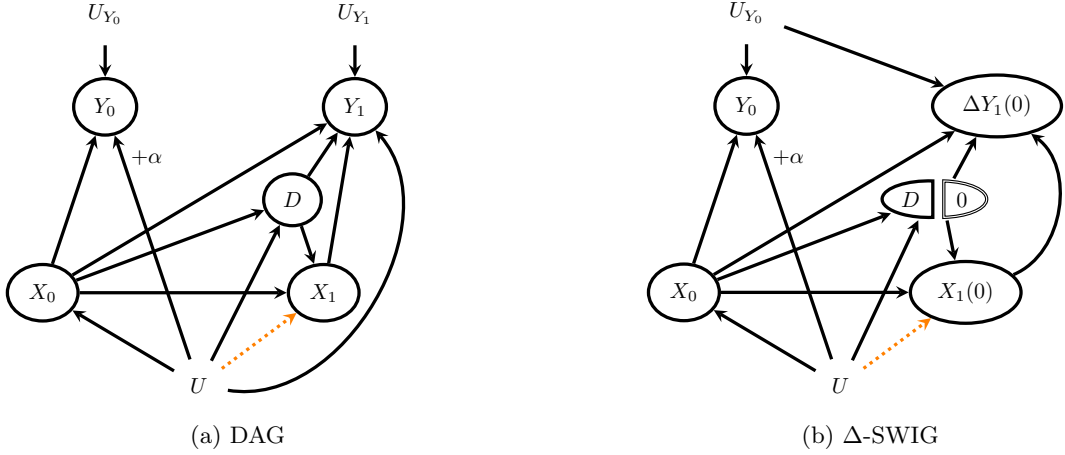
\begin{figure}[t]
    \centering
    \caption{DAG and $\Delta$-SWIG with post-treatment control variables}
    \label{fig:DAG_DSWIG_post}
\begin{subfigure}[b]{0.48\textwidth}
\centering
\makebox[\linewidth][c]{%
  \resizebox{0.78\linewidth}{!}{%
        \begin{tikzpicture}
        \tikzset{line width=1.5pt, ell/.style={draw, fill = white, inner xsep=5pt,inner ysep=5pt, line width=1.5pt}, unobs/.style={ fill = none, inner xsep=5pt,inner ysep=5pt, line width=1.5pt}, swig vsplit={gap=5pt, inner line width right=0.5pt}};

% Nodes
    \node[name=y0, ell, shape=ellipse]{$Y_0$};
    \node[name=y1, ell, shape=ellipse] at ($(y0)+(4cm,0cm)$) {$Y_1$};
    % \node[name=d, ell, shape=swig vsplit] at ($(y2)+(-0.5cm,-1.5cm)$) {\nodepart{left}{$D$} \nodepart{right}{$0$}};
    \node[name=d, ell, shape=ellipse] at ($(y1)+(-1cm,-1.5cm)$) {$D$};
    \node[name=x0, ell, shape=ellipse] at ($(y0)+(-1cm,-3cm)$) {$X_0$};
    \node[name=x1, ell, shape=ellipse] at ($(y1)+(-0.5cm,-3cm)$) {$X_1$};
    \node[name=u_dyx, unobs, shape=ellipse] at ($(y1)+(-2.5cm,-4.5cm)$) {$\UDYX$};
    \node[name=u_y0, unobs, shape=ellipse] at ($(y0)+(0cm,+1.5cm)$) {$U_{Y_0}$};
    \node[name=u_y1, unobs, shape=ellipse] at ($(y1)+(0cm,+1.5cm)$) {$U_{Y_1}$};

% Edges

\draw[->,line width=1.5pt,>=stealth]
(d) edge (y1)
% (u_dyx) edge[color = blue] node[midway, right] {$+ \alpha$} (y0)
% (u_dyx) edge[bend right = 60,color = blue] node[midway, right] {$+ \alpha$} (y1)
(u_dyx) edge node[pos = 0.9, right] {$+ \alpha$} (y0)
(u_dyx) edge[bend right = 70] (y1)
(u_dyx) edge (d)
(x0) edge (y0)
(x0) edge (y1)
(x1) edge (y1)
(x0) edge (x1)
(x0) edge (d)
% (x1) edge (d.left south)
(d) edge (x1)
(u_dyx) edge (x0)
(u_dyx) edge[color = orange, dotted]  (x1)
(u_y0) edge (y0)
(u_y1) edge (y1);
    \end{tikzpicture}
    }
    }
\caption{DAG}
        \label{fig:dag-2x2-d-x}
\end{subfigure}
\hfill
\begin{subfigure}[b]{0.48\textwidth}
\centering
\makebox[\linewidth][c]{%
  \resizebox{0.75\linewidth}{!}{%
        \begin{tikzpicture}
        \tikzset{line width=1.5pt, ell/.style={draw, fill = white, inner xsep=5pt,inner ysep=5pt, line width=1.5pt}, unobs/.style={ fill = none, inner xsep=5pt,inner ysep=5pt, line width=1.5pt}, swig vsplit={gap=5pt, inner line width right=0.5pt}};

% Nodes
    \node[name=y0, ell, shape=ellipse]{$Y_0$};
    \node[name=dy1, ell, shape=ellipse] at ($(y0)+(4cm,0cm)$) {$\Delta Y_1(0)$};
    \node[name=d, ell, shape=swig vsplit] at ($(dy1)+(-1cm,-1.5cm)$) {\nodepart{left}{$D$} \nodepart{right}{$0$}};
    \node[name=x0, ell, shape=ellipse] at ($(y0)+(-1cm,-3cm)$) {$X_0$};
    \node[name=x1, ell, shape=ellipse] at ($(dy1)+(-0.5cm,-3cm)$) {$X_1(0)$};
    \node[name=u_dyx, unobs, shape=ellipse] at ($(dy1)+(-2.5cm,-4.5cm)$) {$\UDYX$};
    \node[name=u_y1, unobs, shape=ellipse] at ($(y0)+(0cm,+1.5cm)$) {$U_{Y_0}$};

% Edges

\draw[->,line width=1.5pt,>=stealth]
(d) edge (dy1)
(u_dyx) edge node[pos = 0.9, right] {$+ \alpha$} (y0)
(u_dyx) edge (d)
(x0) edge (y0)
(x1) edge[bend right= 60] (dy1)
(x0) edge (x1)
(x0) edge (dy1)
(x0) edge (d)
(d) edge (x1)
(u_dyx) edge (x0)
(u_dyx) edge[color = orange, dotted] (x1)
(u_y1) edge (y0)
(u_y1) edge (dy1);

    \end{tikzpicture}
    }
    }
\caption{$\Delta$-SWIG}
        \label{fig:d-swig-2x2-d-x}
\end{subfigure}

\begin{minipage}{\linewidth}
\begin{small}
\textit{Note:} The intermediate SWIG and all structural equations are provided in Figure \ref{fig:SWIG_post} in the Appendix.
\end{small}
\end{minipage}
\end{figure}

\section{Multiple Time Periods} \label{sec:multiple}

The previous section largely revisits known results in the well-studied 2x2 setting through the $\Delta$-SWIG lens. This section uses the graphical perspective to establish new results and to generalize or refine existing ones in the multi-period DiD setting of \textcite{callaway_differenceindifferences_2021}. The following Section \ref{sec:implications} discusses then practical implications.

\subsection{Notation and Setting}
Consider $T = \mathcal{T} + 1$ time periods indexed by $t = 0,...,\mathcal{T}$. 
Let $\overline{V} = \{V_0,...,V_\mathcal{T}\}$ denote a sequence of time-varying random variables, and write $\overline{V}_t = \{V_0,...,V_t\}$,  $\underline{V}_t = \{V_t,...,V_\mathcal{T}\}$ for $t > 0$, $\overline V_{t,t'} = \{V_t,\ldots,V_{t'}\}$ for $0 < t \leq t'$ and $\varnothing$ otherwise. 
The same notation applies to realizations, e.g.~$\overline{v}$. Denote general differences as $\Delta V_{s,t} = V_{t} - V_{s}$ such that $\Delta V_{t} = \Delta V_{t-1,t}$.

\textcite{callaway_differenceindifferences_2021} consider the setting with staggered treatment, where no unit is treated in period 0 and treatment is irreversible:
\begin{assumption}[ST] \label{ass:staggered-trmnt}
    (staggered treatment) $D_0 = 0$ and $D_t \geq D_{t-1} ~\forall~t>0$.\footnote{The latter could be formalized in an SCM as $D_t  \defeq f_{D_t}(Pa(D_t),U_{D_t}) = (1-D_{t-1}) g_{D_t}(Pa(D_t), U_{D_t}) + D_{t-1} ~\forall~ t > 0$.}
\end{assumption}
As $D_0 = 0$ in ``all worlds'', it plays no role in the following arguments and is absent from all graphs. Therefore, we define treatment sequences as excluding $D_0$, i.e.~$\overline{D} := \{D_1,...,D_\mathcal{T}\}$, $\overline{D}_t := \{D_1,...,D_t\}$, and $\overline{D}_0 := \varnothing$. This has two advantages: (i) the 2x2 setting of the previous section is nested as the special case with $T=2$ and $D = D_1$, (ii) potential variables obtained by fixing a treatment sequence have an intuitive structure, e.g.~the potential outcome for being first treated in period 3 is denoted $Y_t(\overline{0}_{2},\underline{1}_{3}) = Y_t(0,0,1,1,\dots)$.

The target parameter in \textcite{callaway_differenceindifferences_2021} is the average treatment effect of the group first treated in period $g$ at time $t$:
\vspace{-0.5cm}
\begin{align} \label{eq:att-gt}
    ATT(g,t) := \E[\underbrace{Y_t(\overline{0}_{g-1},\underline{1}_{g})}_{Y_t(g)} - \underbrace{Y_t(\overline{0})}_{Y_t(\infty)} \mid \underbrace{\overline{D}_{g-1} =\overline{0}_{g-1}, \underline{D}_{g} = \underline{1}_{g}}_{G = g}]
\end{align}
The definition using the sequence notation is less common in the DiD literature. However, it follows the mechanics of SWIGs and is used throughout. The curly brackets in \eqref{eq:att-gt} illustrate how it maps to the common group notation used, e.g., in \textcite{roth_whats_2023} where $G := min\{t : D_t = 1\}$.

Identification of $ATT(g,t)$ is based on conditional parallel trends with the never-treated \eqref{eq:cpt-nt} or the not-yet-treated \eqref{eq:cpt-nyt}:
\begin{align}
\E[\Delta Y_t(\overline{0}) \mid \mathbf Z, \overline{D}_{m-1}=\overline{0}_{m-1}, \underline{D}_m=\underline{1}_m]
&= \E[\Delta Y_t(\overline{0}) \mid \mathbf Z, \overline{D}=\overline{0}], \ \forall\, 1 \le t \le \mathcal{T},\ m \le t \label{eq:cpt-nt} \\
\E[\Delta Y_t(\overline{0}) \mid \mathbf Z, \overline{D}_{m-1}=\overline{0}_{m-1}, \underline{D}_m=\underline{1}_m]
& =\E[\Delta Y_t(\overline{0}) \mid \mathbf Z, \overline{D}_s=\overline{0}_s], \ \forall\, m \le t \le s \le \mathcal{T},\ t \ge 2  \label{eq:cpt-nyt}
\end{align}
$\mathbf Z$ serves here as a placeholder.
%and may contain time-invariant and/or time-varying covariates. 
Like in the previous section, we use $\Delta$-SWIGs to obtain concrete sets of valid control variables or to explain the mechanisms that prevent reading off CPTs.

Generalizing the previous settings, we focus on the class of SCMs $\mathfrak{M} := \{\mathcal{M} : \mathbf{V} = \{ \overline{X},\overline{D}, \overline{Y}, U \} \}$, where $\overline{X}$, $\overline{D}$, and $\overline{Y}$ denote the covariate, binary treatment, and outcome sequences, respectively, and, $U$ denotes unobservable time-invariant confounders.\footnote{We abstract from time-invariant covariates $X$ to avoid notational clutter. All expressions that require conditioning on time-varying covariates can additionally condition on $X$ without changing the result.}
We also generalize single world additive separability to avoid stating a tailored version for each scenario:
\begin{assumption}[$R^\alpha$] \label{ass:SWAS-staggered}
    (single world additive separability - staggered treatment) \\
Additionally to Assumption \assref{ass:staggered-trmnt}, the structural equation for $Y_t$ in SCM $\mathcal{M} \in \mathfrak{M}$ has the form
\begin{align}
    Y_t  &\defeq  f_{Y_t}(Pa(Y_t), U_{Y_t}) =\underbrace{\alpha (U,C) + g_{Y_t}(Pa(Y_t) \setminus \{ U, \overline{D}_t \}, U_{Y_t})}_{Y_t(\overline{0})} + D_t \cdot \tau_t(Pa(Y_t)\setminus D_t, U_{Y_t}) ~\forall~t
\end{align}
where $C \defeq \left(\bigcap_{t = 0}^\mathcal{T} Pa(Y_t) \right) \setminus \{U,\overline{D} \}$ are parents of all $Y_t$, e.g.~$X_0$.
\end{assumption}
Like in Assumption \assref{ass:swas-2x2}, unobservable confounders enter never-treated potential outcomes only via an additively separable function and individual treatment effects are unrestricted. In particular, note that the latter may depend on the previous treatment sequence $\overline{D}_{t-1}$, which implies under Assumption \assref{ass:staggered-trmnt} that the effects may depend on the timing of the first treatment.\footnote{A general single world additive separability without Assumption \assref{ass:staggered-trmnt} is provided in Appendix \ref{app:general-swas}.} Assumption \assref{ass:SWAS-staggered} not only nests Assumption \assref{ass:swas-2x2} but also the untreated potential outcome models used in the literature. We collect them and show how they are special cases of Assumption \assref{ass:SWAS-staggered} in Appendix \ref{app:swas-lit}.

We also streamline the discussion by abstracting from outcome dynamics because they prevent reading off CPTs in model class $\mathfrak{M}$ (see discussion in Section \ref{sec:2x2-ydyn}):
\begin{assumption}[$R^Y$] \label{ass:no-y-dyn}
    (no outcome dynamics) $Desc(Y_t) = \varnothing$ in SCM $\mathcal{M} \in \mathfrak{M}$.
\end{assumption}

\subsection{Three Time Periods ($T=3$)} \label{sec:t3}

For expositional simplicity, we begin by considering the case $T=3$. The mechanics extend directly to $3 < T < \infty$, as discussed in Section \ref{sec:general-t}. DAG \ref{fig:dag-2x3} shows the $T=3$ setting under Assumptions \assref{ass:SWAS-staggered} and \assref{ass:no-y-dyn} with dotted edge $D_{1}\rightarrow X_2$ because we first consider the common setting without treatment-covariate feedback $D_{1}\not\rightarrow X_2$ and then allow for edge $D_{1}\rightarrow X_2$.

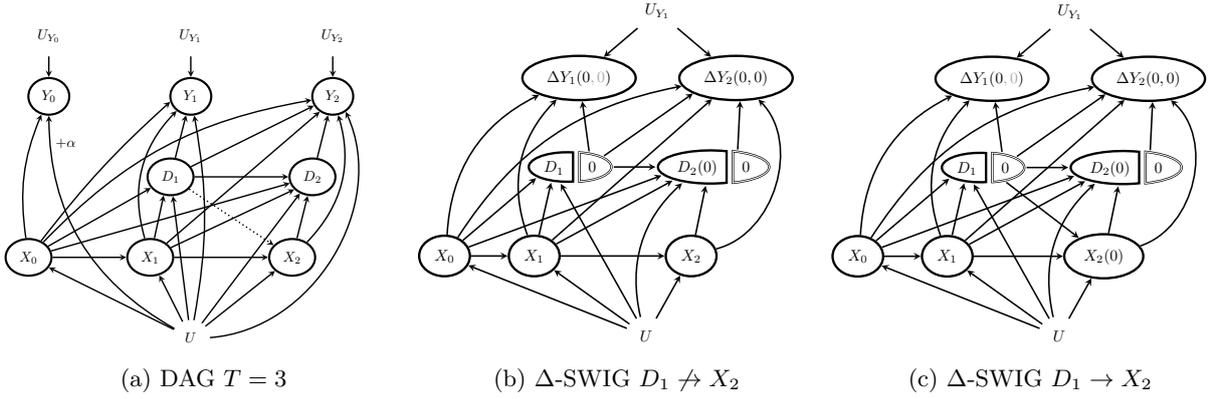
\begin{figure}[t]
    \centering
    \caption{$T = 3$ with and without treatment-covariate feedback}
    \label{fig:2x3}
\begin{subfigure}[b]{0.32\textwidth}
\centering
\makebox[\linewidth][c]{%
  \resizebox{1\linewidth}{!}{%
    \begin{tikzpicture}
        \tikzset{line width=1.5pt, ell/.style={draw, fill = white, inner xsep=5pt,inner ysep=5pt, line width=1.5pt}, unobs/.style={ fill = none, inner xsep=5pt,inner ysep=5pt, line width=1.5pt}, swig vsplit={gap=5pt, inner line width right=0.5pt}};

% Nodes
    \node[name=y0, ell, shape=ellipse]{$Y_0$};
    \node[name=y1, ell, shape=ellipse] at ($(y0)+(3.5cm,0cm)$) {$Y_1$};
    \node[name=y2, ell, shape=ellipse] at ($(y1)+(3.5cm,0cm)$) {$Y_2$};
    % \node[name=y4, ell, shape=ellipse] at ($(y2)+(3.5cm,0cm)$) {$Y_4$};
    \node[name=d1, ell, shape=ellipse] at ($(y1)+(-0.5cm,-2cm)$) {$D_1$} ;
    \node[name=d2, ell, shape=ellipse] at ($(y2)+(-0.5cm,-2cm)$) {$D_2$} ;
    % \node[name=d4, ell, shape=ellipse] at ($(y4)+(-0.5cm,-2.5cm)$) {$D_4$};
    \node[name=x0, ell, shape=ellipse] at ($(y0)+(-0.5cm,-4cm)$) {$X_0$};
    \node[name=x1, ell, shape=ellipse] at ($(y1)+(-1cm,-4cm)$) {$X_1$};
    \node[name=x2, ell, shape=ellipse] at ($(y2)+(-1cm,-4cm)$) {$X_2$};
    % \node[name=x4, ell, shape=ellipse] at ($(y4)+(-1cm,-4.5cm)$) {$X_4$};
    \node[name=u, unobs, shape=ellipse] at ($(y1)+(0cm,-6cm)$) {$U$};
    \node[name=e0, unobs, shape=ellipse] at ($(y0)+(0cm,+1.5cm)$) {$U_{Y_0}$};
    \node[name=e1, unobs, shape=ellipse] at ($(y1)+(0cm,+1.5cm)$) {$U_{Y_1}$};
    \node[name=e2, unobs, shape=ellipse] at ($(y2)+(0cm,+1.5cm)$) {$U_{Y_2}$};
    % \node[name=e4, unobs, shape=ellipse] at ($(y4)+(0cm,+1.5cm)$) {$U_{Y_4}$};

% Edges

\draw[->,line width=1pt,>=stealth]
(d1) edge (y1)
(d1) edge (y2)
(d1) edge (d2)
(d2) edge (y2)
(u) edge[bend left = 30] node[pos = 0.9, right] {$+ \alpha$}  (y0)
(u) edge[bend right = 10] (y1)
(u) edge[bend right = 55]  (y2.310)
(u) edge (d1)
(u) edge[bend left = 0] (d2)
(u) edge (x0)
(u) edge (x1)
(u) edge (x2)
(e0) edge (y0)
(e1) edge (y1)
(e2) edge (y2)
(x1) edge[bend left = 30] (y1)
(x2) edge[bend right = 30] (y2)
(x1) edge (d1.250)
(x2) edge (d2)
(x0) edge (d1)
(x1) edge (d2.215)
(x0) edge (x1)
(x1) edge (x2)
(x0) edge[bend left = 15] (y0)
(x0) edge[bend left = 5] (y1.195)
(x0) edge[bend left = 18] (y2)
(x0) edge (d2)
(x1) edge (y2.230)
(d1) edge[dotted] (x2)
;
\end{tikzpicture}
    }
    }
\caption{DAG $T=3$}
        \label{fig:dag-2x3}
\end{subfigure}
\hfill
\begin{subfigure}[b]{0.32\textwidth}
\centering
\makebox[\linewidth][c]{%
  \resizebox{1\linewidth}{!}{%
    \begin{tikzpicture}
        \tikzset{line width=1.5pt, ell/.style={draw, fill = white, inner xsep=5pt,inner ysep=5pt, line width=1.5pt}, unobs/.style={ fill = none, inner xsep=5pt,inner ysep=5pt, line width=1.5pt}, swig vsplit={gap=4pt, inner line width right=0.5pt}};

% Nodes
    \node[name=dy1, ell, shape=ellipse] {$\Delta Y_1(0\textcolor{gray}{,0})$};
    \node[name=dy2, ell, shape=ellipse] at ($(dy1)+(3.5cm,0cm)$) {$\Delta Y_2(0,0)$};
    \node[name=d1, ell, shape=swig vsplit] at ($(dy1)+(-0.2cm,-2cm)$) {\nodepart{left}{$D_1$} \nodepart{right}{$0$}};
    \node[name=d2, ell, shape=swig vsplit] at ($(dy2)+(-0.5cm,-2cm)$) {\nodepart{left}{$D_2(0)$} \nodepart{right}{$0$}};
    \node[name=x0, ell, shape=ellipse] at ($(dy1)+(-3cm,-4cm)$) {$X_0$};
    \node[name=x1, ell, shape=ellipse] at ($(dy1)+(-1cm,-4cm)$) {$X_1$};
    \node[name=x2, ell, shape=ellipse] at ($(dy2)+(-1cm,-4cm)$) {$X_2$};
    \node[name=u, unobs, shape=ellipse] at ($(dy1)+(1.5cm,-5.8cm)$) {$U$};
    \node[name=e1, unobs, shape=ellipse] at ($(dy1)+(1.75cm,+1.5cm)$) {$U_{Y_1}$};
    
% Edges

\draw[->,line width=1pt,>=stealth]
(d1) edge (d2)
(u) edge[bend right = 0] (d1.250)
(u) edge[bend left = 30] (d2.210)
(u) edge (x0)
(u) edge (x1)
(u) edge (x2)
(x0) edge[bend left = 30] (dy1)
(x0) edge[bend left = 23] (dy2)
(x1) edge[bend right = 0] (dy2)
(x1) edge[bend left = 32] (dy1)
(x2) edge[bend right = 60] (dy2)
(x1) edge (d1.225)
(x2) edge (d2)
(x0) edge (d1.200)
(x0) edge (d2.190)
(x1) edge (d2.200)
(x0) edge (x1)
(x1) edge (x2)
(e1) edge (dy1)
(e1) edge (dy2)
;

\draw[->,line width=1pt,>=stealth] (d1.45) -- (dy1);
\draw[->,line width=1pt,>=stealth] (d1.20) -- (dy2.200);
\draw[->,line width=1pt,>=stealth] (d2.40) -- (dy2.280);
\end{tikzpicture}
    }
    }
\caption{$\Delta$-SWIG $D_{1}\not\rightarrow X_2$}
        \label{fig:d-swig-2x3}
\end{subfigure}
\hfill
\begin{subfigure}[b]{0.32\textwidth}
\centering
\makebox[\linewidth][c]{%
  \resizebox{1\linewidth}{!}{%
    \begin{tikzpicture}
        \tikzset{line width=1.5pt, ell/.style={draw, fill = white, inner xsep=5pt,inner ysep=5pt, line width=1.5pt}, unobs/.style={ fill = none, inner xsep=5pt,inner ysep=5pt, line width=1.5pt}, swig vsplit={gap=4pt, inner line width right=0.5pt}};

% Nodes
    \node[name=dy1, ell, shape=ellipse] {$\Delta Y_1(0\textcolor{gray}{,0})$};
    \node[name=dy2, ell, shape=ellipse] at ($(dy1)+(3.5cm,0cm)$) {$\Delta Y_2(0,0)$};
    \node[name=d1, ell, shape=swig vsplit] at ($(dy1)+(-0.2cm,-2cm)$) {\nodepart{left}{$D_1$} \nodepart{right}{$0$}};
    \node[name=d2, ell, shape=swig vsplit] at ($(dy2)+(-0.5cm,-2cm)$) {\nodepart{left}{$D_2(0)$} \nodepart{right}{$0$}};
    \node[name=x0, ell, shape=ellipse] at ($(dy1)+(-3cm,-4cm)$) {$X_0$};
    \node[name=x1, ell, shape=ellipse] at ($(dy1)+(-1cm,-4cm)$) {$X_1$};
    \node[name=x2, ell, shape=ellipse] at ($(dy2)+(-1cm,-4cm)$) {$X_2(0)$};
    \node[name=u, unobs, shape=ellipse] at ($(dy1)+(1.5cm,-5.8cm)$) {$U$};
    \node[name=e1, unobs, shape=ellipse] at ($(dy1)+(1.75cm,+1.5cm)$) {$U_{Y_1}$};
    
% Edges

\draw[->,line width=1pt,>=stealth]
(d1) edge (d2)
(u) edge[bend right = 0] (d1.250)
(u) edge[bend left = 30] (d2.210)
(u) edge (x0)
(u) edge (x1)
(u) edge (x2)
(x0) edge[bend left = 30] (dy1)
(x0) edge[bend left = 23] (dy2)
(x1) edge[bend right = 0] (dy2)
(x1) edge[bend left = 32] (dy1)
(x2) edge[bend right = 60] (dy2)
(x1) edge (d1.225)
(x2) edge (d2)
(x0) edge (d1.200)
(x0) edge (d2.190)
(x1) edge (d2.200)
(x0) edge (x1)
(x1) edge (x2)
(e1) edge (dy1)
(e1) edge (dy2)
(d1) edge (x2)
;

\draw[->,line width=1pt,>=stealth] (d1.45) -- (dy1);
\draw[->,line width=1pt,>=stealth] (d1.20) -- (dy2.200);
\draw[->,line width=1pt,>=stealth] (d2.40) -- (dy2.280);
\end{tikzpicture}
    }
    }
\caption{$\Delta$-SWIG $D_{1}\rightarrow X_2$}
        \label{fig:d-swig-2x3-d-x}
\end{subfigure}
\begin{minipage}{\linewidth}
\begin{small}
\textit{Note:} The intermediate SWIGs are in Figure \ref{fig:swigs_3}. All structural equations are provided in Figure \ref{fig:SCMs_T3} in the Appendix.
\end{small}
\end{minipage}
\end{figure}

\subsubsection{Without Treatment-Covariate Feedback} \label{sec:2x3-wo-d-x}

The absence of treatment-covariate feedback resembles the setting discussed, e.g., in \textcite{caetano_differenceindifferences_2024}. The respective $\Delta$-SWIG is depicted in Figure \ref{fig:d-swig-2x3}.

\textbf{Period 1:}
Consider the short-term effect for the group that is first treated in period 1:
\begin{align}
    ATT(1,1) & = \E[Y_1(1\textcolor{gray}{,1}) - Y_1(0\textcolor{gray}{,0}) \mid D_1 = 1, D_2 = 1]\\
     & = \underbrace{\E[\Delta Y_1(\textcolor{blue}{1}\textcolor{gray}{,1})\mid D_1 = \textcolor{blue}{1}, D_2 = 1]}_{\text{observable trend}} - \underbrace{\E[\Delta Y_1(0\textcolor{gray}{,0}) \mid D_1 = 1, D_2 = 1]}_{\text{counterfactual trend}} \label{eq:att11-rearrange}
\end{align}
where the second line follows by adding $Y_0\textcolor{gray}{(1,1)} - Y_0\textcolor{gray}{(0,0)} = 0$ and rearranging.

Note that $\Delta$-SWIG \ref{fig:d-swig-2x3} contains multiple open paths between $\Delta Y_1(0\textcolor{gray}{,0})$ and $D_1, D_2(0)$ that preclude reading off unconditional joint independence $\Delta Y_1(0\textcolor{gray}{,0}) \indep D_1, D_2$, which would be sufficient to identify the counterfactual trend in \eqref{eq:att11-rearrange}. However, the graph implies two useful independencies
and the second one is further processed to prepare for the next step:\footnote{\eqref{eq:2x3-cia1} follows because $X_t$ blocks $D_1 \leftarrow (U \rightarrow) X_t \rightarrow \Delta Y_1(0\textcolor{gray}{,0})$, collider $\Delta Y_2(0,0)$ blocks paths like $D_1 \leftarrow U \rightarrow X_2 \rightarrow \Delta Y_2(0,0) \leftarrow U_{Y_1} \rightarrow \Delta Y_1(0\textcolor{gray}{,0})$, and fixed node 0 and collider $D_2(0)$ block paths like $D_1 \leftarrow U \rightarrow D_2(0) \leftarrow 0 \rightarrow \Delta Y_1(0\textcolor{gray}{,0})$. $X_2$ does not open any path conditional on $X_0,X_1$ and is therefore optional. \eqref{eq:2x3-cia2} follows from similar arguments. Notably, $D_1$ is a collider on paths like $D_2(0) \leftarrow X_t \rightarrow D_1 \leftarrow X_1 \rightarrow Y_1(0\textcolor{gray}{,0})$ but conditional on $X_t$ the path is always blocked and $D_1$ is optional.} 
\begin{align}
    \Delta Y_1(0\textcolor{gray}{,0}) & \indep D_1 \mid X_0,X_1\textcolor{gray}{,X_2} \label{eq:2x3-cia1} \\
    \Delta Y_1(0\textcolor{gray}{,0}) & \indep D_2(0) \mid X_0,X_1\textcolor{gray}{,X_2,D_1}  \label{eq:2x3-cia2} \\
    \Rightarrow \Delta Y_1(0\textcolor{gray}{,0}) & \indep D_2(\dzero) \mid X_0,X_1\textcolor{gray}{,X_2},D_1 = \dzero  & \label{eq:2x3-cia3} \\
    \Rightarrow \Delta Y_1(0\textcolor{gray}{,0}) & \indep D_2 \mid X_0,X_1\textcolor{gray}{,X_2},D_1 = 0  & \text{(consistency)} \label{eq:2x3-cia4}\\
    \Rightarrow \Delta Y_1(0\textcolor{gray}{,0}) & \indep D_2 \mid X_0,X_1\textcolor{gray}{,X_2},D_1  & \text{(\assref{ass:staggered-trmnt})} \label{eq:2x3-cia5}
\end{align}
While \eqref{eq:2x3-cia2} contains potential treatment $D_2(0)$, it becomes observable $D_2$ by conditioning on $D_1=0$ and consistency. The final line uses that $\Delta Y_1(0\textcolor{gray}{,0}) \indep D_2 \mid X_0,X_1\textcolor{gray}{,X_2},D_1 = d$ for all $d \in \{0,1\}$ holds by Assumption \assref{ass:staggered-trmnt}.\footnote{Then, $D_2 = 1$ - and is thus degenerate - conditional on $D_1 = 1$ such that $\Delta Y_1(0\textcolor{gray}{,0}) \indep D_2 \mid X_0,X_1\textcolor{gray}{,X_2},D_1 = 1$ holds. Combined with \eqref{eq:2x3-cia4}, this observation yields \eqref{eq:2x3-cia5}.}
Combining \eqref{eq:2x3-cia1} and \eqref{eq:2x3-cia5} by contraction yields then joint conditional independence
\begin{align}
    \Delta Y_1(0\textcolor{gray}{,0}) & \indep D_1, D_2 \mid X_0,X_1\textcolor{gray}{,X_2} \label{eq:2x3-joint-indep}
\end{align}
and a variety of CPTs for different control groups:
\begin{align}
 & \E[\Delta Y_1(0\textcolor{gray}{,0}) \mid X_0,X_1\textcolor{gray}{,X_2}, D_1 = 1, D_2 = 1] &  \text{counterfactual trend for }ATT(1,1) \\
=~ &\E[\Delta Y_1(\dzero\textcolor{gray}{,0}) \mid X_0,X_1\textcolor{gray}{,X_2}, D_1 = \dzero, D_2 = 0] & \text{observable never-treated trend} \label{eq:nt-trend}\\
=~ &\E[\Delta Y_1(\dzero\textcolor{gray}{,0}) \mid X_0,X_1\textcolor{gray}{,X_2}, D_1 = \dzero, D_2 = 1] & \text{observable later-treated trend} \label{eq:lt-trend} \\
=~ &\E[\Delta Y_1(\dzero\textcolor{gray}{,0}) \mid X_0,X_1\textcolor{gray}{,X_2}, D_1 = \dzero] & \text{observable not-yet-treated trend} \label{eq:nyt-trend}
\end{align}
Therefore, the counterfactual trend in \eqref{eq:att11-rearrange} can be identified, e.g., using the never-treated as control group $\E[\E[\Delta Y_1 \mid X_0,X_1\textcolor{gray}{,X_2}, D_1 = 0, D_2 = 0] \mid D_1 = 1, D_2 = 1]$. This means that $ATT(1,1)$ is identified and can be estimated by controlling for $X_0,X_1$ (and $X_2$) via standard regression or (augmented) inverse probability estimators.

Additionally, conditional \textit{pre-trends} for $ATT(2,2)$ are parallel. In particular, never-treated trend \eqref{eq:nt-trend} and later-treated trend \eqref{eq:lt-trend} coincide.
Consequently, pre-trend analyses are not expected to reject parallel pre-trends if they control for $X_0,X_1$ (and $X_2$).

\textbf{Period 2:} $\Delta$-SWIG \ref{fig:d-swig-2x3} further implies $\Delta Y_2(0,0) \indep D_1 \mid X_0,X_1,X_2$ and $\Delta Y_2(0,0) \indep D_2(0) \mid X_0,X_1,X_2\textcolor{gray}{,D_1}$. Following similar arguments as before they contract to joint conditional independence 
\begin{align}
    \Delta Y_2(0,0) & \indep D_1, D_2 \mid X_0,X_1,X_2 \label{eq:2x3-indep-y2}
\end{align}
and justify identification of dynamic effect $ATT(1,2)$ and short-term effect $ATT(2,2)$ using the never-treated as control group:
\begin{align}
&\E[\Delta Y_2(0,0) \mid X_0,X_1,X_2, D_1 = 1, D_2 = 1] & \text{counterfactual trend for }ATT(1,2)\\
=~ &\E[\Delta Y_2(0,0) \mid X_0,X_1,X_2, D_1 = 0, D_2 = 1] & \text{counterfactual trend for }ATT(2,2) \\
=~ &\E[\Delta Y_2(\dzero,\dzero) \mid X_0,X_1,X_2, D_1 = \dzero, D_2 = \dzero] & \text{observable never-treated trend} 
\end{align}

\subsubsection{With Treatment-Covariate Feedback}

Figure \ref{fig:d-swig-2x3-d-x} shows the $\Delta$-SWIG with treatment-covariate feedback. The main difference to Figure \ref{fig:d-swig-2x3} is that it includes now potential covariate $X_2(0)$. This has minor consequences for CPTs in period 1 but major consequences in period 2.

\textbf{Period 1:}
First, note that $\Delta$-SWIG \ref{fig:d-swig-2x3-d-x} implies
\begin{align} \label{eq:2x3-d-x-cia1}
    \Delta Y_1(0\textcolor{gray}{,0}) & \indep D_1, D_2 \mid X_0,X_1
\end{align}
which is nearly identical to \eqref{eq:2x3-joint-indep} but $X_2$ is not optional anymore. At first glance this suggests that CPT does not hold anymore conditional on the full covariate sequence $X_0,X_1,X_2$. However, the additional independence
\begin{align}
   \Delta Y_1(0\textcolor{gray}{,0}) & \indep X_2(0) \mid X_0,X_1 \textcolor{gray}{, D_1, D_2(0)} \\
   \Rightarrow \Delta Y_1(0\textcolor{gray}{,0}) & \indep X_2 \mid X_0,X_1, D_1 = 0, D_2 & \text{(consistency)} \label{eq:2x3-d-x-cov-uc}
\end{align}
implies that all results regarding identifiability of $ATT(1,1)$ and parallel pre-trends conditional on $X_0,X_1$ (and $X_2$) from the previous section apply identically with treatment-covariate feedback.\footnote{To see this formally, note that $\E[\Delta Y_1(0\textcolor{gray}{,0}) \mid X_0,X_1, D_1 = 1, D_2 = 1] = \E[\Delta Y_1(\dzero\textcolor{gray}{,0}) \mid X_0,X_1, D_1 = \dzero, D_2 = d] = \E[\Delta Y_1(\dzero\textcolor{gray}{,0}) \mid X_0,X_1,X_2, D_1 = \dzero, D_2 = d]$ where the first equality follows by \eqref{eq:2x3-d-x-cia1} and the second by \eqref{eq:2x3-d-x-cov-uc}. Therefore, the counterfactual trend for $ATT(1,1)$ is $\E[\Delta Y_1(0\textcolor{gray}{,0}) \mid X_0,X_1, D_1 = 1, D_2 = 1] = \eqref{eq:nt-trend} = \eqref{eq:lt-trend} = \eqref{eq:nyt-trend}$.}

\textbf{Period 2:} $\Delta$-SWIG \ref{fig:d-swig-2x3-d-x} implies two crucial independencies that could not be simplified further:
\begin{align}
    \Delta Y_2(0,0) & \indep D_2 \mid X_0,X_1,X_2,D_1=0 \label{eq:2x3-d-x-cia21} \\
    \Delta Y_2(0,0) & \indep D_1,D_2 \mid X_0, X_1, X_2(0) \label{eq:2x3-d-x-cia2}
\end{align}
Independence \eqref{eq:2x3-d-x-cia21} conditions on observable variables but only shows marginal independence regarding $D_2$. This leads to the following (non-)CPTs:
\begin{align}
& \E[\Delta Y_2(0,0) \mid X_0,X_1,X_2, D_1 = 0, D_2 = 1] & \text{counterfactual trend for }ATT(2,2) \\
= &\E[\Delta Y_2(\dzero,\dzero) \mid X_0,X_1,X_2, D_1 = \dzero, D_2 = \dzero] & \text{observable never-treated trend} \\
\neq &\E[\Delta Y_2(0,0) \mid X_0,X_1,X_2, D_1 = 1, D_2 = 1] & \text{counterfactual trend for }ATT(1,2)
\end{align}
The short-term effect $ATT(2,2)$ can still be identified using the never-treated trend. However, the dynamic effect $ATT(1,2)$ is not identified in the presence of treatment-covariate feedback.

Identification of $ATT(1,2)$ would require a joint independence conditional on observable variables. However, the minimal conditioning set to achieve joint independence is shown in \eqref{eq:2x3-d-x-cia2} and contains unobservable $X_2(0)$. This poses a dilemma where researchers can basically choose between two biases: only controlling for $X_0,X_1$ leads to confounding bias, but controlling additionally for $X_2 = D_1 X_2(1) + (1-D_1) X_2(0)$ instead of $X_2(0)$ leads to ``wrong world control bias''. % We provide explicit bias expressions and further discussion in Section XXX below.

\textbf{Summary:} Short-term effects $ATT(g,g)$ are still identified in the presence of treatment-covariate feedback. However, dynamic effects are not identified without further assumptions. Strikingly, pre-trends hold by construction in this setting and would not flag problems with the dynamic effect estimates. This rationalizes the patterns shown in the introduction.

\subsection{Multiple Time Periods and Long Differences} \label{sec:general-t}

This section extends the analysis beyond $T = 3$ and first differences. The latter is motivated by the conditional DiD estimand targeting $ATT(g,t)$ that involves general differences $\Delta Y_{g-1,t}$:\footnote{We focus here on the estimand using never-treated but note that everything discussed in this section applies identically for the not-yet-treated, as we show in Appendix \ref{sec:NYT_app}.}
{\small
\begin{align} \label{eq:did-estimand}
     DiD_{g,t}(\mathbf{Z}) & := \underbrace{\E[\Delta Y_{g-1,t}\mid \overline{D}_{g-1} = \overline{0}_{g-1}, \underline{D}_g = \underline{1}_g]}_{\text{observable trend of group g}}  - \underbrace{\E[\E[\Delta Y_{g-1,t} \mid \mathbf{Z}, \overline{D} = \overline{0}]\mid \overline{D}_{g-1} = \overline{0}_{g-1}, \underline{D}_g = \underline{1}_g]}_{\text{covariate adjusted never-treated trend}}
\end{align} }

for some adjustment set $\mathbf{Z}$.
This section discusses which adjustment sets ensure, under which model restrictions, that $\Delta Y_{g-1,t} (\overline{0}) \indep \underline{D}_{g} \mid \mathbf{Z}, \overline{D}_{g-1} = \overline{0}_{g-1}$ and therefore $ATT(g,t) = DiD_{g,t}\left(\mathbf{Z}\right)$.\footnote{The conditional independence implies CPT $\E[\Delta Y_{g-1,t} (\overline{0}) \mid \mathbf{Z}, \overline{D} = \overline{0}] = \E[\Delta Y_{g-1,t} (\overline{0}) \mid \mathbf{Z}, \overline{D}_{g-1} = \overline{0}_{g-1}, \underline{D}_g = \underline{1}_g]$, which is an essential part of the standard identification proof that is provided in Appendix \ref{app:attgt-ident} for completeness.}
To this end, we first show how to graphically obtain the minimal sufficient adjustment set, which can contain unobservable potential covariates and might be infeasible. Nevertheless, this set provides a constructive starting point for discussing the feasible adjustment sets in the following.

\subsubsection{Minimal Sufficient Adjustment Set} \label{sec:min-sas}

The previous sections establish how independencies leading to sufficient adjustment sets can be deduced from $\Delta$-SWIGs for $T \leq 3$, e.g. \eqref{eq:ccpr-cpt}, \eqref{eq:2x3-joint-indep}, \eqref{eq:2x3-indep-y2}, \eqref{eq:2x3-d-x-cia2}. The same strategies apply to larger $T$ but require to either draw many $\Delta$-nodes within one $\Delta$-SWIG or many $\Delta$-SWIGs. To circumvent this complexity, we show that the single SWIG from step 3 in Procedure \ref{proc:d-swig} suffices to compactly reason about minimal sufficient adjustment sets under Assumptions \assref{ass:SWAS-staggered} and \assref{ass:no-y-dyn}, building on the following result:
\begin{proposition} (minimal sufficient adjustment set) \label{prop:min-z0}
    Consider SCM $\mathcal{M} \in \mathfrak{M}$ and let $\mathcal{G}(\overline{0})$ be the corresponding SWIG constructed in step 3 of Procedure \ref{proc:d-swig}.
Then, under Assumptions \assref{ass:SWAS-staggered} and \assref{ass:no-y-dyn} the following conditional independence holds for all $g,t$:
    \begin{align} \label{eq:cia-general}
    \Delta Y_{g-1,t} (\overline{0}) \indep \underline{D}_{g} \mid~ & \underbrace{Pa_{\mathcal{G}(\overline{0})}(Y_{g-1}(\overline{0})) \setminus \{U,U_{Y_{g-1}}, \overline{0}\}, Pa_{\mathcal{G}(\overline{0})}(Y_{t}(\overline{0})) \setminus \{U,U_{Y_{t}}, \overline{0}\}}_{\mathbf{S}_{g,t}(\overline{0})}, \overline{D}_{g-1} = \overline{0}_{g-1}
\end{align}
Without further restrictions there is no $\mathbf{S}(\overline{0})$ such that also $\Delta Y_{g-1,t} (\overline{0}) \indep \underline{D}_{g} \mid \mathbf{S}(\overline{0}), \overline{D}_{g-1} = \overline{0}_{g-1}$ can be deduced but $\mathbf{S}_{g,t}(\overline{0}) := Pa_{\mathcal{G}(\overline{0})}(Y_{g-1}(\overline{0})) \setminus \{U,U_{Y_{g-1}}, \overline{0}\}, Pa_{\mathcal{G}(\overline{0})}(Y_{t}(\overline{0})) \setminus \{U,U_{Y_{t}}, \overline{0}\}$ is not a subset of $\mathbf{S}(\overline{0})$. The proof is provided in Appendix \ref{sec:proofs_porps_multiple}.
\end{proposition}
The proposition shows that the minimal sufficient adjustment set consists of graphical parents $Pa_{\mathcal{G}(\overline{0})}$, excluding exogenous variables and $U$, of the two level nodes in the intermediate SWIG that map to $\Delta Y_{g-1,t} (\overline{0})$ in a $\Delta$-SWIG. $\mathbf{S}_{g,t}(\overline{0})$ might contain unobservable potential covariates and we define their observable analogue as $\mathbf{S}_{g,t} := \{X_t \in \overline{X} : X_t(\overline{0}) \in \mathbf{S}_{g,t}(\overline{0})\}$, which plays an important role in characterizing feasible sets below.

Consider Figure \ref{fig:swigs_3} for illustration. It depicts the intermediate SWIGs between DAG \ref{fig:dag-2x3} and $\Delta$-SWIGs \ref{fig:d-swig-2x3} and \ref{fig:d-swig-2x3-d-x}.
For example, Proposition \ref{prop:min-z0} applied to the SWIG without treatment-covariate feedback in Figure \ref{fig:swig-2x3} leads to
\begin{enumerate}
    \item $\Delta Y_{1,2} (\overline{0}) \indep D_2 \mid  X_{0}, X_{1}, X_{2}, D_{1} = 0 \Rightarrow ATT(2,2)$ is identified; $\mathbf{S}_{2,2}(\overline{0}) = \mathbf{S}_{2,2} = \{X_{0}, X_{1}, X_{2}\}$
    \item $\Delta Y_{0,2} (\overline{0}) \indep D_1, D_2 \mid X_{0}, X_{1}, X_{2} \Rightarrow ATT(1,2)$ is identified; $\mathbf{S}_{1,2}(\overline{0}) = \mathbf{S}_{1,2} = \{X_{0}, X_{1}, X_{2}\}$
\end{enumerate}
or applied to SWIG \ref{fig:swig-2x3-d-x} with treatment-covariate feedback yields
\begin{enumerate}[start=3]
    \item $\Delta Y_{1,2} (\overline{0}) \indep D_2 \mid X_{0}, X_{1}, X_{2}(\dzero), D_{1} = \dzero \Rightarrow ATT(2,2)$ is identified; $\mathbf{S}_{2,2}(\overline{0}) = \{X_{0}, X_{1}, X_{2}(0)\}$, $\mathbf{S}_{2,2} = \{X_{0}, X_{1}, X_{2}\}$
    \item $\Delta Y_{0,2} (\overline{0}) \indep D_1, D_2 \mid X_{0}, X_{1}, X_{2}(0) \Rightarrow ATT(1,2)$ is \textit{not} identified; $\mathbf{S}_{1,2}(\overline{0}) = \{X_{0}, X_{1}, X_{2}(0)\}$, $\mathbf{S}_{1,2} = \{X_{0}, X_{1}, X_{2}\}$
\end{enumerate}
The last example is an instance where the minimal sufficient adjustment set is not enough for identification because potential covariate $X_{2}(0)$ is not observable. However, Proposition \ref{prop:min-z0} still provides a useful starting point to discuss feasible adjustment sets under different model restrictions.

\subsubsection{Valid Adjustment Sets}

\begin{figure}[t]
    \centering
    \caption{SWIGs for $T = 3$ with and without treatment-covariate feedback}
    \label{fig:swigs_3}
\begin{subfigure}[b]{0.45\textwidth}
\centering
\makebox[\linewidth][c]{%
  \resizebox{1\linewidth}{!}{%
    \begin{tikzpicture}
        \tikzset{line width=1.5pt, ell/.style={draw, fill = white, inner xsep=5pt,inner ysep=5pt, line width=1.5pt}, unobs/.style={ fill = none, inner xsep=5pt,inner ysep=5pt, line width=1.5pt}, swig vsplit={gap=5pt, inner line width right=0.5pt}};

% Nodes
    \node[name=y0, ell, shape=ellipse]{$Y_0\textcolor{gray}{(0,0)}$};
    \node[name=y1, ell, shape=ellipse] at ($(y0)+(3.5cm,0cm)$) {$Y_1(0,\textcolor{gray}{0})$};
    \node[name=y2, ell, shape=ellipse] at ($(y1)+(3.5cm,0cm)$) {$Y_2(0,0)$};
    % \node[name=y4, ell, shape=ellipse] at ($(y2)+(3.5cm,0cm)$) {$Y_4$};
    \node[name=d1, ell, shape=swig vsplit] at ($(y1)+(-0.5cm,-2cm)$) {\nodepart{left}{$D_1$} \nodepart{right}{$0$}} ;
    \node[name=d2, ell, shape=swig vsplit] at ($(y2)+(-0.5cm,-2cm)$) {\nodepart{left}{$D_2(0)$} \nodepart{right}{$0$}} ;
    % \node[name=d4, ell, shape=ellipse] at ($(y4)+(-0.5cm,-2.5cm)$) {$D_4$};
    \node[name=x0, ell, shape=ellipse] at ($(y0)+(-0.5cm,-4cm)$) {$X_0$};
    \node[name=x1, ell, shape=ellipse] at ($(y1)+(-1cm,-4cm)$) {$X_1$};
    \node[name=x2, ell, shape=ellipse] at ($(y2)+(-1cm,-4cm)$) {$X_2$};
    % \node[name=x4, ell, shape=ellipse] at ($(y4)+(-1cm,-4.5cm)$) {$X_4$};
    \node[name=u, unobs, shape=ellipse] at ($(y1)+(0cm,-6cm)$) {$U$};
    \node[name=e0, unobs, shape=ellipse] at ($(y0)+(0cm,+1.5cm)$) {$U_{Y_0}$};
    \node[name=e1, unobs, shape=ellipse] at ($(y1)+(0cm,+1.5cm)$) {$U_{Y_1}$};
    \node[name=e2, unobs, shape=ellipse] at ($(y2)+(0cm,+1.5cm)$) {$U_{Y_2}$};
    % \node[name=e4, unobs, shape=ellipse] at ($(y4)+(0cm,+1.5cm)$) {$U_{Y_4}$};

% Edges

\draw[->,line width=1pt,>=stealth]
(d1) edge (y2)
(d1) edge (d2)
(u) edge[bend left = 30] node[pos = 0.95, right] {$+ \alpha$}  (y0)
(u) edge[bend right = 30] node[pos = 0.95, right] {$+ \alpha$} (y1)
(u) edge[bend right = 60] node[pos = 0.95, right] {$+ \alpha$}  (y2.350)
(u) edge (d1)
(u) edge[bend left = 0] (d2)
(u) edge (x0)
(u) edge (x1)
(u) edge (x2)
(e0) edge (y0)
(e1) edge (y1)
(e2) edge (y2)
(x1) edge[bend left = 45] (y1)
(x2) edge[bend right = 60] (y2)
(x1) edge (d1.250)
(x2) edge (d2)
(x0) edge (d1)
(x1) edge (d2.215)
(x0) edge (x1)
(x1) edge (x2)
(x0) edge[bend left = 15] (y0)
(x0) edge[bend left = 10, dotted] (y1.195)
(x0) edge[bend left = 18, dotted] (y2)
(x0) edge (d2)
(x1) edge[dotted] (y2.230)
;

\draw[->,line width=1pt,>=stealth] (d1.40) -- (y1);
\draw[->,line width=1pt,>=stealth] (d2.40) -- (y2);

\end{tikzpicture}
    }
    }
\caption{SWIG $T=3$, $D_{1}\not\rightarrow X_2$}
        \label{fig:swig-2x3}
\end{subfigure}
\hfill
\begin{subfigure}[b]{0.45\textwidth}
\centering
\makebox[\linewidth][c]{%
  \resizebox{1\linewidth}{!}{%
    \begin{tikzpicture}
        \tikzset{line width=1.5pt, ell/.style={draw, fill = white, inner xsep=5pt,inner ysep=5pt, line width=1.5pt}, unobs/.style={ fill = none, inner xsep=5pt,inner ysep=5pt, line width=1.5pt}, swig vsplit={gap=5pt, inner line width right=0.5pt}};

% Nodes
    \node[name=y0, ell, shape=ellipse]{$Y_0\textcolor{gray}{(0,0)}$};
    \node[name=y1, ell, shape=ellipse] at ($(y0)+(3.5cm,0cm)$) {$Y_1(0,\textcolor{gray}{0})$};
    \node[name=y2, ell, shape=ellipse] at ($(y1)+(3.5cm,0cm)$) {$Y_2(0,0)$};
    % \node[name=y4, ell, shape=ellipse] at ($(y2)+(3.5cm,0cm)$) {$Y_4$};
    \node[name=d1, ell, shape=swig vsplit] at ($(y1)+(-0.5cm,-2cm)$) {\nodepart{left}{$D_1$} \nodepart{right}{$0$}} ;
    \node[name=d2, ell, shape=swig vsplit] at ($(y2)+(-0.5cm,-2cm)$) {\nodepart{left}{$D_2(0)$} \nodepart{right}{$0$}} ;
    % \node[name=d4, ell, shape=ellipse] at ($(y4)+(-0.5cm,-2.5cm)$) {$D_4$};
    \node[name=x0, ell, shape=ellipse] at ($(y0)+(-0.5cm,-4cm)$) {$X_0$};
    \node[name=x1, ell, shape=ellipse] at ($(y1)+(-1cm,-4cm)$) {$X_1$};
    \node[name=x2, ell, shape=ellipse] at ($(y2)+(-1cm,-4cm)$) {$X_2(0)$};
    % \node[name=x4, ell, shape=ellipse] at ($(y4)+(-1cm,-4.5cm)$) {$X_4$};
    \node[name=u, unobs, shape=ellipse] at ($(y1)+(0cm,-6cm)$) {$U$};
    \node[name=e0, unobs, shape=ellipse] at ($(y0)+(0cm,+1.5cm)$) {$U_{Y_0}$};
    \node[name=e1, unobs, shape=ellipse] at ($(y1)+(0cm,+1.5cm)$) {$U_{Y_1}$};
    \node[name=e2, unobs, shape=ellipse] at ($(y2)+(0cm,+1.5cm)$) {$U_{Y_2}$};
    % \node[name=e4, unobs, shape=ellipse] at ($(y4)+(0cm,+1.5cm)$) {$U_{Y_4}$};

% Edges

\draw[->,line width=1pt,>=stealth]
(d1) edge (y2)
(d1) edge (d2)
(u) edge[bend left = 30] node[pos = 0.95, right] {$+ \alpha$}  (y0)
(u) edge[bend right = 25] node[pos = 0.95, right] {$+ \alpha$} (y1)
(u) edge[bend right = 60] node[pos = 0.95, right] {$+ \alpha$}  (y2.350)
(u) edge (d1)
(u) edge[bend left = 8] (d2)
(u) edge (x0)
(u) edge (x1)
(u) edge (x2)
(e0) edge (y0)
(e1) edge (y1)
(e2) edge (y2)
(x1) edge[bend left = 45] (y1)
(x2) edge[bend right = 60] (y2)
(x1) edge (d1.250)
(x2) edge (d2)
(x0) edge (d1)
(x1) edge (d2.215)
(x0) edge (x1)
(x1) edge (x2)
(x0) edge[bend left = 15] (y0)
(x0) edge[bend left = 10, dotted] (y1.195)
(x0) edge[bend left = 18, dotted] (y2)
(x0) edge (d2)
(x1) edge[dotted] (y2.230)
(d1) edge (x2)
;

\draw[->,line width=1pt,>=stealth] (d1.40) -- (y1);
\draw[->,line width=1pt,>=stealth] (d2.40) -- (y2);

\end{tikzpicture}
    }
    }
\caption{SWIG $T=3$, $D_{1} \rightarrow X_2$}
        \label{fig:swig-2x3-d-x}
\end{subfigure}
\begin{minipage}{\linewidth}
\begin{small}
\textit{Note:} All structural equations are provided in Figure \ref{fig:SCMs_T3} in the Appendix.
\end{small}
\end{minipage}
\end{figure}
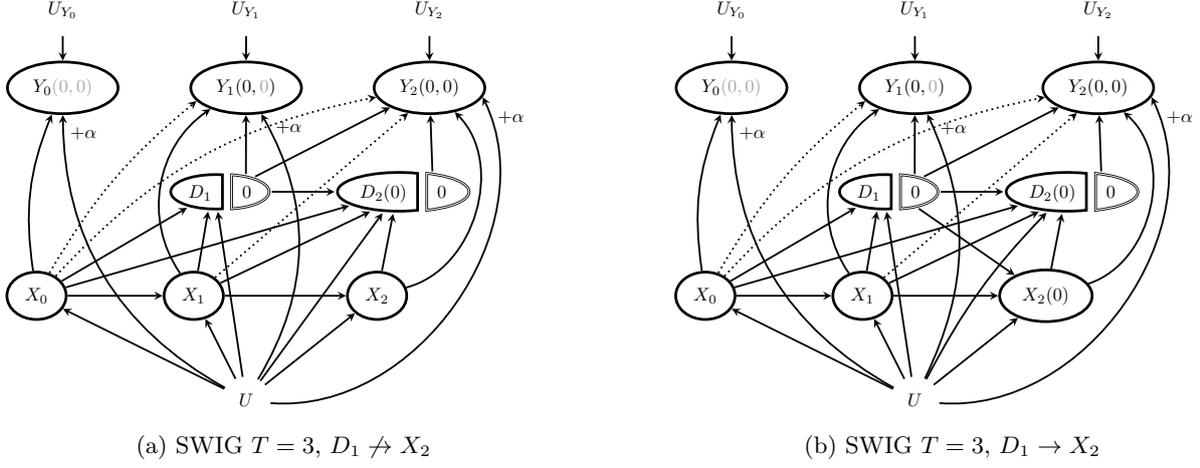

To organize the discussion, we adapt the concept of valid adjustment sets (VAS) from the causal graph literature on unconfoundedness for the DiD setting \parencite[see e.g.][]{shpitser_validity_2012}. Define VAS as all sets of time-varying covariates ensuring that the conditional DiD estimand is unbiased for $ATT(g,t)$ in all SCMs that obey restrictions $R$, i.e.~$\mathcal{Z}_{g,t}(R) := \{\mathbf{Z} \subseteq \overline{X} : ATT(g,t) = DiD_{g,t}(\mathbf{Z}) ~\forall~ \mathcal{M} \in \mathfrak{M} \text{ satisfying }R \}$. For example, the discussion in Section \ref{sec:2x2-ydyn} implies already that $\mathcal{Z}_{g,t}(\assref{ass:SWAS-staggered}) = \varnothing ~\forall~g,t$, i.e.~no VAS exists under additive separability Assumption \assref{ass:SWAS-staggered} alone. We gradually impose now additional model restrictions to study existence and form of the respective minimal VAS $\mathbf{Z}^{min}_{g,t}(R)$, which represents the minimal element of $\mathcal{Z}_{g,t}(R)$. The focus on the minimal VAS has two reasons: First, it is arguably the most interesting VAS from a practical perspective. Second, the maximal VAS and sets in between follow directly once the minimal VAS is established, as we show at the end of the section.

\textbf{Minimal valid adjustment sets:} Proposition \ref{prop:min-z0} provides the minimal sufficient adjustment set under Assumptions \assref{ass:SWAS-staggered} and \assref{ass:no-y-dyn}. It also implies the existence of minimal VAS for $g>t$ without further restrictions. This follows directly from \eqref{eq:cia-general} because in pre-treatment periods any potential variables in $\mathbf{S}_{g,t}(\overline{0})$ are observable conditional on $\overline{D}_{g-1} = \overline{0}_{g-1}$ such that the independence still holds if $\mathbf{S}_{g,t}(\overline{0})$ is replaced by observable $\mathbf{S}_{g,t}$. Thus, $\mathbf{Z}^{min}_{g,t}(\assref{ass:SWAS-staggered},\assref{ass:no-y-dyn}) = \mathbf{S}_{g,t} = \overline{X}_{g-1} ~\forall~g>t$.\footnote{For example, in both SWIGs of Figure \ref{fig:swigs_3} $\mathbf{Z}^{min}_{2,0}(\assref{ass:SWAS-staggered},\assref{ass:no-y-dyn}) = \mathbf{S}_{2,0} = \{X_0, X_1\}$.} In contrast, for all $g \leq t$, no VAS exists without further restrictions. This shows that \textit{parallel pre-trends} require fewer restrictions than identifying the treatment effects of interest.\footnote{Note that $ATT(g,t) = 0, ~\forall~g>t$ by construction (see end of the identification proof in Appendix \ref{app:attgt-ident}). Thus, $ATT(g,t) = DiD_{g,t}(\mathbf{Z}) = 0$ by definition if $\mathbf Z$ is a VAS. As $DiD_{g,t}(\mathbf{Z}) = 0, g > t$ is interpreted as parallel pre-trends, a VAS for all $g > t$ implies parallel pre-trends.} We illustrate and discuss practical implications of this finding below.

To identify the \textit{short-term effect} $ATT(g,g)$ via a VAS, it suffices to assume that $D_t$ does not affect $X_t$:
\begin{assumption}[$R^{DX}_t$] \label{ass:x-d-ordering}
    (no within-period treatment-covariate effect) $X_t \not\in Desc(D_t) ~\forall~t$ in $\mathcal{M} \in \mathfrak{M}$.
\end{assumption}
This assumption holds naturally if $X_t$ precedes $D_t$ within time period $t$. This is the case in most settings depicted above. The exception is the setting of \textcite{caetano_difference_2024} discussed in Section \ref{sec:post-treatment} where Assumption \assref{ass:x-d-ordering} holds after removing edge $D_t \rightarrow X_t$. Proposition \ref{prop:min-z0} yields $\mathbf{Z}^{min}_{g,t}(\assref{ass:SWAS-staggered},\assref{ass:no-y-dyn},\assref{ass:x-d-ordering}) = \mathbf{S}_{g,t} = \overline{X}_{g}~\forall~g \geq t$ such that $ATT(g,g)$ is identifiable under Assumptions \assref{ass:SWAS-staggered}, \assref{ass:no-y-dyn}, and \assref{ass:x-d-ordering}. Example 3 in Section \ref{sec:min-sas} illustrates this result via SWIG \ref{fig:swig-2x3-d-x} for $ATT(2,2)$ where $\mathbf{Z}^{min}_{2,2}(\assref{ass:SWAS-staggered},\assref{ass:no-y-dyn},\assref{ass:x-d-ordering}) = \mathbf{S}_{2,2} = \{X_{0}, X_{1}, X_{2}\}$.

Identification also of the \textit{dynamic effects} $ATT(g,t)$ for all $g,t$ requires further restrictions. The most common restriction is to rule out treatment-covariate feedback from $D_t$ to future covariates \parencite[e.g.][]{caetano_differenceindifferences_2024, ghanem_selection_2024}:
\begin{assumption}[$R^{DX}_{t+1}$] \label{ass:no-d-x}
    (no treatment-covariate feedback) $\underline{X}_{t+1} \cap Desc(D_t) = \varnothing ~\forall~t$ in $\mathcal{M} \in \mathfrak{M}$.
\end{assumption}
Then, the minimal sufficient adjustment set equals the minimal VAS, i.e~$\mathbf{Z}^{min}_{g,t}(\assref{ass:SWAS-staggered},\assref{ass:no-y-dyn},\assref{ass:x-d-ordering},\assref{ass:no-d-x}) = \mathbf{S}_{g,t} = \mathbf{S}_{g,t}(\overline{0}) = \overline{X}_{t} ~\forall~g,t$ because SWIG $\mathcal{G}(\overline{0})$ contains no potential covariates. Example 2 in Section \ref{sec:min-sas} illustrates this result, while Example 4 shows why dynamic effects are not identified without a further restriction.

So far, the added restrictions lead to the existence of additional minimal VAS. However, restrictions can also further shrink already existing minimal VAS. One prominent example is to rule out that $X_t$ directly affects future outcomes \parencite[e.g.][]{borusyak_revisiting_2024, ghanem_selection_2024, caetano_differenceindifferences_2024}:
\begin{assumption}[$R^{XY}_{t+1}$] \label{ass:no-x-dyn}
    (no direct covariate-outcome dynamics) $X_t \not\in Pa(Y_s) ~\forall~s > t$ in $\mathcal{M} \in \mathfrak{M}$.
\end{assumption}
The SWIGs in Figure \ref{fig:swigs_3} fulfill this assumption if all dotted edges are absent.
Then, each potential outcome node is affected by only one contemporary covariate leading to adjustment sets with two elements, one element for each level of the respective difference.\footnote{Consequently, Examples 1 and 3 in Section \ref{sec:min-sas} would not include $X_0$, and Examples 2 and 4 would not include $X_1$. Appendix \ref{app:t4} discusses also a $T=4$ example for illustration.}

Finally, also the remaining contemporary covariate-outcome effect can be ruled out:
\begin{assumption}[$R^{XY}_{t}$] \label{ass:no-x-y}
    (no within-period covariate-outcome effect) $X_t \not\in Pa(Y_t) ~\forall~ t$ in $\mathcal{M} \in \mathfrak{M}$.
\end{assumption}
Combined with \assref{ass:no-x-dyn} this would rule out any direct covariate-outcome effects and would be encoded by removing all edges from $X_t$ to outcome nodes in the graphs above. Then, the empty set is the minimal VAS regardless of the presence of treatment-covarate feedback, i.e.~$\mathbf{Z}^{min}_{g,t}(\assref{ass:SWAS-staggered},\assref{ass:no-y-dyn},\assref{ass:no-x-dyn},\assref{ass:no-x-y}) = \varnothing ~\forall~g,t.$ %, because each potential outcome $Y_t(\overline{0})$ in the respective intermediate SWIG has no other parents than $U$, $U_{Y_t}$, and $\overline{0}$. 

The remaining possible restrictions within model class $\mathfrak{M}$ are (i) covariate-covariate, (ii) covariate-treatment, (iii) treatment-treatment, and (iv) treatment-outcome. (i) and (ii) would not refine the presented results as they neither affect the parents of $Y_t(\overline{0})$ nor the presence of potential covariates. (iii) would rule out staggered treatment. (iv) would rule out the effects of interest. We therefore ignore these possibilities.
Table \ref{tab:min-vas-did} in the next section summarizes the minimal VAS under the relevant restrictions and discusses practical implications.

% \begin{table}
%     \centering
%     \begin{tabular}{l|cccccc}
%         & $X_t$ & $\underline{X}_{t+1}$ & $D_t$ & $\underline{D}_{t+1}$ & $Y_t$ & $\underline{Y}_{t+1}$ \\
%         \hline
%         $X_t \not\rightarrow$ & -- & $R_{t+1}^{XX}$ & $R_t^{XD}$ & $R_{t+1}^{XD}$  & $R_t^{XY}$ & $R_{t+1}^{XY}$ \\
%         $D_t \not\rightarrow$ & $R_t^{DX}$ & $R_{t+1}^{DX}$ & -- & $R_{t+1}^{DD}$  & $R_t^{DY}$ & $R_{t+1}^{DY}$ \\
%         $Y_t \not\rightarrow$ & $R_t^{YX}$ & $R_{t+1}^{YX}$ & $R_t^{YD}$ & $R_{t+1}^{YD}$ & -- & $R_{t+1}^{YY}$ \\
%     \end{tabular}
%     \caption{Caption}
%     \label{tab:placeholder}

%     \begin{tabular}{l|cccccc}
%         & $X_t$ & $\underline{X}_{t+1}$ & $D_t$ & $\underline{D}_{t+1}$ & $Y_t$ & $\underline{Y}_{t+1}$ \\
%         \hline
%         $X_t \not\rightarrow$ & -- & redundant & redundant & redundant & $R_t^{XY}$ & $R_{t+1}^{XY}$ \\
%         $D_t \not\rightarrow$ & $R_t^{DX}$ & $R_{t+1}^{DX}$ & -- & against A2  & effect of interest & effect of interest \\
%         $Y_t \not\rightarrow$ & $R_t^{YX}$ & $R_{t+1}^{YX}$ & $R_t^{YD}$ & $R_{t+1}^{YD}$  & -- & $R_{t+1}^{YY}$ \\
%     \end{tabular}
%     \caption{Caption}
%     \label{tab:placeholder}
% \end{table}

\textbf{All valid adjustment sets $\mathcal{Z}_{g,t}(R)$:} In the $T=3$ settings of Section \ref{sec:t3}, whenever identification holds conditional on a subsequence of $\overline{X}$, it also holds conditional on the full sequence. This result generalizes and allows to compactly characterize all VAS:
\begin{proposition} (valid adjustment sets) \label{prop:all-vas}
    Under the assumptions of Proposition \ref{prop:min-z0}. Whenever there exists a minimal valid adjustment set for model class $\mathfrak{M}$ under restrictions $R$, i.e.~$\mathbf{Z}^{min}_{g,t}(R) = \mathbf{S}_{g,t}$, then
\begin{align}
\mathcal{Z}_{g,t}(R) = \{\mathbf{Z} : \mathbf{Z}^{min}_{g,t}(R) \subseteq \mathbf{Z} \subseteq  \overline{X} \}.
\end{align}
\end{proposition}
The proof in Appendix \ref{sec:proofs_porps_multiple} uses that $\Delta Y_{g-1,t}(\overline{0}) \indep \overline{X} \setminus  \mathbf{Z}^{min}_{g,t}(R)  \mid  \mathbf{Z}^{min}_{g,t}(R), \overline{D} = \overline{0}$ holds whenever $\mathbf{Z}^{min}_{g,t}(R)$ exists and generalizes the arguments  for $T=3$ around \eqref{eq:2x3-d-x-cov-uc}.

The key insight from Proposition \ref{prop:all-vas} is that the full sequence $\overline{X}$ is always the maximal VAS if a minimal VAS exists. Also, any subset of $\overline{X}$ that contains the minimal VAS remains a VAS.

\section{Practical Considerations} \label{sec:implications}

Table \ref{tab:min-vas-did} summarizes the minimal VAS for different questions under the different restrictions. It also implicitly characterizes all VAS according to Proposition \ref{prop:all-vas}. Each cell containing a minimal VAS can be augmented with any combination of other time-varying covariates up to the full sequence and still remains a VAS.\footnote{Table \ref{tab:vas-did} in Appendix \ref{app:vas} provides such an augmented table for completeness.}
Table \ref{tab:min-vas-did} highlights how additional restrictions expand the range of identifiable parameters or shrink the size of the minimal VAS. In the following we discuss how the table can be used to justify adjustment strategies, how it assist in interpreting pre-trend diagnostics, and conclude with practical recommendations.

\begin{table}[ht]
\caption{Minimal valid adjustment sets under different model restrictions} \label{tab:min-vas-did}
    \centering
\begin{threeparttable}  \small
    \begin{tabular}{ccccc|ccc}
      \multicolumn{4}{l}{Restrictions:} & & Parallel pre-trends? & Short-term effect? & Dynamic effect? \\
       \tableref{ass:SWAS-staggered}\&\tableref{ass:no-y-dyn} & \tableref{ass:x-d-ordering} & \tableref{ass:no-d-x} & \tableref{ass:no-x-dyn} & \tableref{ass:no-x-y} & $DiD_{g,t}(\cdot) = 0 ~\forall~ g > t$ & $DiD_{g,g}(\cdot) = \tau_{g,g}$ & $DiD_{g,t}(\cdot) = \tau_{g,t} ~\forall~ g < t$ \\
         \hline
      no & n/y & n/y & n/y & n/y
        & --
        & -- 
        & --  \\
      yes & no & no & no & no
        & $\overline{X}_{g-1}$
        & -- 
        & --  \\
      yes & yes & no & no & no
        & $\overline{X}_{g-1}$
        & $\overline{X}_{g}$
        & --  \\    
      yes & yes & yes & no & no
        & $\overline{X}_{g-1}$
        & $\overline{X}_{g}$
        & $\overline{X}_{t}$ \\    
      yes & no & no & yes & no
        & $X_{t},X_{g-1}$ 
        & -- 
        & --  \\
      yes & yes & no & yes & no
        & $X_{t},X_{g-1}$ 
        & $X_{g-1},X_g$ 
        & --  \\    
      yes & yes & yes & yes & no
        & $X_{t},X_{g-1}$ 
        & $X_{g-1},X_g$ 
        & $X_{g-1},X_t$ \\
      yes & n/y & n/y & yes & yes
        & $\varnothing$ 
        & $\varnothing$ 
        & $\varnothing$ \\
        \hline
    \end{tabular}
\begin{tablenotes}[flushleft]
\footnotesize
\item \textit{Notes:} \assref{ass:SWAS-staggered}: single world additive separability \& staggered treatment; \assref{ass:no-y-dyn}: no outcome dynamics; \assref{ass:x-d-ordering}: no within-period treatment-covariate effect; \assref{ass:no-d-x}: no treatment-covariate feedback; \assref{ass:no-x-dyn}: no direct covariate-outcome dynamics; \assref{ass:no-x-y}: no within-period covariate-outcome effect. n/y means that this restriction does not affect the result. $DiD_{g,t}(\cdot)$ is defined in \eqref{eq:did-estimand}. $\tau_{g,t} = ATT(g,t)$ is defined in \eqref{eq:att-gt}. The cell entries provide the minimal valid adjustment set such that the equality in the header holds when it is inserted into the bracket of $DiD_{g,t}(\cdot)$. -- indicates that no valid adjustment set exists.
The table is representative for settings where $Y_t$ realizes at the end of period $t$. This is without loss of generality because any setting can be rearranged accordingly. The same table holds for DiD using the not-yet-treated as control group, as we show in Appendix \ref{sec:NYT_app}.
\end{tablenotes}
    \end{threeparttable}
\end{table}

\subsection{How Can Adjustment Strategies be Justified?}

Table \ref{tab:min-vas-did} allows to read off valid adjustment sets after researchers commit to a particular set of assumptions.
For illustration, consider the least restrictive potential outcome model in \textcite{caetano_differenceindifferences_2024} that reads in our notation $Y_t(\overline{0}) = \alpha(U) + g_{Y_t}(X_t) + U_{Y_t}$ such that Assumptions \assref{ass:SWAS-staggered}, \assref{ass:x-d-ordering}, \assref{ass:no-d-x} and \assref{ass:no-x-dyn} hold by construction, and \assref{ass:no-y-dyn} is usually also directly or indirectly imposed.\footnote{\label{fn:implied-r} We interpret the fact that observable $X_t$ appears in an equation of a potential outcome as ruling out treatment-covariate feedback because otherwise $X_t(0)$ would enter. Either way, \textcite{caetano_differenceindifferences_2024} rule out treatment-covariate feedback explicitly in the text such that Assumptions \assref{ass:x-d-ordering}/\assref{ass:no-d-x} hold. Also, single $X_t$ and not sequence $\overline{X}_t$ enters the equation such that Assumption \assref{ass:no-x-dyn} holds. State dependence holds because previous outcomes do not enter the equation. Outcome-covariate feedback is also ruled because $X_t$ would otherwise become a potential covariate due to a $D_{t-1} \rightarrow Y_{t-1} \rightarrow X_t$ path. Thus, only outcome-treatment feedback is not ruled out by the functional form. However, it is usually ruled out as well \parencite[see][for a detailed discussion linking it to strict exogeneity]{bonhomme_back_2025}. Assumption \assref{ass:no-y-dyn} is therefore not covered directly implied by the functional from but often implicitly assumed.} 
The second last row of Table \ref{tab:min-vas-did} shows that $\{X_{g-1},X_t\}$ is the minimal VAS to identify $ATT(g,t)$ for all $g \leq t$ under these restrictions and thus all valid adjustment sets are characterized by $\mathcal{Z}_{g,t}(\assref{ass:SWAS-staggered},\assref{ass:no-y-dyn},\assref{ass:x-d-ordering},\assref{ass:no-d-x},\assref{ass:no-x-dyn}) = \{\mathbf{Z} : \{X_{g-1},X_t\} \subseteq \mathbf{Z} \subseteq  \overline{X} \}$ by Proposition \ref{prop:all-vas}. This ``forward engineering'' of VAS has the advantage that the assumptions to be defended are the starting point of the argument and therefore very transparent. It also cleanly separates identification and estimation issues.

However, Table \ref{tab:min-vas-did} can also be consulted to ``backwards engineer'' assumptions under which a particular adjustment strategy could be justified. To this end, we determine those minimal VAS that are a subset of the adjustment strategy and read off sufficient assumptions on the left. 
We apply this for five adjustment strategies that are discussed in the literature and briefly comment on implications of the results:

\textbf{Full covariate sequence:} \textcite{caetano_differenceindifferences_2024} assume parallel trends conditional on the full covariate sequence $\overline{X}$ and discuss estimation based on the identification result $ATT(g,t) = DiD_{g,t}(\overline{X})$. As all minimal VAS in Table \ref{tab:min-vas-did} are subsets of the full sequence, the restrictions in rows 4, 7 and 8 justify controlling for $\overline{X}$ to identify $ATT(g,t)$ for all $g,t$. Row 4 imposes the weakest restrictions, ruling out treatment-covariate feedback while allowing for unrestricted covariate-outcome effects. Even then the maximal VAS $\overline{X}$ only coincides with the minimal VAS $\overline{X}_t$ for $ATT(g,\mathcal{T})$. For $t < \mathcal{T}$, the same assumptions would justify controlling for $\overline{X}_t$ instead of $\overline{X}$. Using the minimal VAS therefore partly addresses the covariate dimensionality issues discussed in \textcite{caetano_differenceindifferences_2024} Section 6 without imposing further assumptions.\footnote{The fraction of variables that can be removed without losing identification of $ATT(g,t)$ is $1 -(t+1)/T$, which is largest in early periods and increases with $T$.}

\textbf{Early level and first difference:} \textcite{caetano_differenceindifferences_2024} further discuss to condition on $\{X_{g-1},\Delta X_{g-1,t}\}$ instead of $\overline{X}$ to reduce the covariate dimension for estimation. This is equivalent to conditioning on the two levels $\{X_{g-1},X_{t}\}$. Thus, Table \ref{tab:min-vas-did} reveals that this strategy can be justified by the restrictions in the last two rows that both limit covariate-outcome dynamics/effects. The reduced dimensionality requires therefore stronger assumptions compared to using $\overline{X}$. This highlights that the decision to move from $\overline{X}$ to $X_{g-1},X_{t}$ is not only an estimation issue but also has implications for identification.\footnote{The third option discussed in \textcite{caetano_differenceindifferences_2024} is to control for the covariate average $T^{-1}\sum_t X_t$. We cannot provide a graphical way of justifying this strategy without violating the ``the future cannot affect the past'' principle and therefore do not discuss this strategy further.}

\textbf{Period 0 controls:} One common strategy is to match observations on $X_0$ and to run unconditional DiD on the matched sample \parencite[e.g.,][]{dupas_women_2024,fenizia_organized_2024,humlum_changing_2025}. This is equivalent to controlling for $X_0$. Only the empty set in the last row of Table \ref{tab:min-vas-did} is a subset of this one variable. Therefore, $X_0$ is a non-minimal VAS under this strictest set of assumptions. However, Assumption \assref{ass:no-x-y} and \assref{ass:no-x-dyn} could be relaxed to only hold for $t > 0$ making $X_0$ the minimal VAS.

\textbf{Earlier period control:} The description of function \texttt{att\_gt()} in the popular \texttt{did} R package \parencite[][]{callaway_did_man} states that \textit{"[...] in each 2x2 comparison, the covariates are taken to be the value of the covariates in the earlier time period [...]"}.\footnote{The quote is taken from version \texttt{2.3.0}.} This corresponds to single covariate $X_{min(g-1,t)}$. Again, this is a non-minimal VAS under the strictest set of assumptions in Table \ref{tab:min-vas-did}.\footnote{Similar to period 0 controls, Assumption \assref{ass:no-x-y} could technically be relaxed for one period $t'$ such that $X_{min(g-1,t)}$ is the minimal VAS for all $ATT(g,t)$ with $g-1 = t'$ or $t = t'$.}

\textbf{Pre-treatment controls:} \textcite{callaway_differenceindifferences_2021} discuss controlling for ``pre-treatment'' variables $X$. The absence of a time index suggests that $X$ is limited to time-invariant covariates. However, it might be interpreted differently as in \textcite{ghanem_when_2026} Section III that discusses controlling for $X_{g}$. Again, this can be justified by ruling out that time-varying covariates directly affect the outcome at all. \textcite{ghanem_when_2026} discuss alternative sufficient conditions without imposing $R_t^{XY}$ but adding restrictions beyond the causal structure.

% It might be counterintuitive that the VAS derived under weaker assumptions include post-treatment covariates. Even these VAS require to rule out treatment-covariate feedback unless researchers are willing to impose additional assumptions.

\textbf{Summary:} The results in Table \ref{tab:min-vas-did} enable a principled analysis of different adjustment strategies. The most compact discussion is possible if assumptions are stated first to directly derive all VAS that ensure identification. Selecting the concrete VAS for estimation requires then statistical considerations. The minimal VAS is a plausible default choice but selecting the optimal VAS in terms of statistical efficiency in a principled manner is a promising direction for future research.

The ``covariates first, assumptions second'' direction is less transparent but reveals noteworthy insights. The benefit of controlling for one pre-treatment covariate as in the final three strategies appears limited. Their inclusion can be justified under a causal structure without direct covariate-outcome effects. Then, $X_0$, $X_{min(g-1,t)}$ or $X_{g}$ are all VAS, but not the minimal VAS. If such covariate-outcome restrictions are considered implausible, the alternative is to rule out treatment-covariate feedback. Then, even the minimal VAS contains post-treatment but pre-outcome covariates. In particular, controlling for pre-outcome covariates $\overline{X}_t$ requires the least additional restrictions to identify $ATT(g,t)$ via DiD.

It is important to note that the graphical perspective is only one possibility to justify an adjustment strategy. For example, \textcite{ghanem_selection_2024} and \textcite{ghanem_when_2026} provide alternative templates based on the treatment selection mechanism that go beyond the causal structure and are thus not naturally covered in the graphical approach. One advantage of our approach is that Table \ref{tab:min-vas-did} also implies testable implications to complement graphical reasoning with statistical evidence, as we discuss in the following.

\subsection{What Can Pre-Trends Test?} \label{sec:ppt-test}

As discussed in the introduction, insignificant parallel pre-trends (PPT) often serve as the only justification for CPT in practice. However, \textcite{bonhomme_back_2025} notes that while rejections of PPT are informative, non-rejections should be interpreted with caution.
Table \ref{tab:min-vas-did} provides a basis for clarifying what exactly can be learned from such rejections and where substantive reasoning is still required. 
For a compact discussion, define five conditional DiD estimands with superscripts related to the number of involved time-varying covariates: $\delta_{g,t}^{0} := DiD_{g,t}(\varnothing)$, $\delta_{g,t}^{2} := DiD_{g,t}(X_t,X_{g-1})$, $\delta_{g,t}^g := DiD_{g,t}(\overline{X}_{g-1})$, $\delta_{g,t}^t := DiD_{g,t}(\overline{X}_{t})$, and $\delta_{g,t}^{T} := DiD_{g,t}(\overline{X})$.
% Define the null hypotheses of PPT conditional on different numbers of time-varying covariates as $H_0^g : DiD_{g,t}(\overline{X}_{g-1}) = 0 ~\forall~g > t$; $H_0^{2} : DiD_{g,t}(X_t,X_{g-1}) = 0 ~\forall~g > t$; $H_0^{0} : DiD_{g,t}(\varnothing) = 0 ~\forall~g > t$.

Rows 2, 5 and 8 of Table \ref{tab:min-vas-did} inform about which assumptions can be rejected if PPT are rejected. The restrictions in the left part of the table are sufficient for PPT with the adjustment set on the right. Therefore, rejecting the respective conditional PPT makes a violation of at least one of the assumptions on the left necessary. 
Concretely, 
\begin{align*}
\text{reject } H_0^{0} : \delta_{g,t}^0 = 0 ~\forall~g > t
&\Rightarrow \text{at least one of }
\{\assref{ass:SWAS-staggered}, \assref{ass:no-y-dyn}, \assref{ass:no-x-dyn}, \assref{ass:no-x-y}\} \text{ is violated}, \\
\text{reject } H_0^{2} : \delta_{g,t}^2 = 0 ~\forall~g > t
&\Rightarrow \text{at least one of }
\{\assref{ass:SWAS-staggered}, \assref{ass:no-y-dyn}, \assref{ass:no-x-dyn}\} \text{ is violated}, \\
\text{reject } H_0^g : \delta_{g,t}^g = 0 ~\forall~g > t
&\Rightarrow \text{at least one of }
\{\assref{ass:SWAS-staggered}, \assref{ass:no-y-dyn}\} \text{ is violated}.
\end{align*}
This reveals a nested structure where specifications with more covariates are more informative about which assumptions are violated.\footnote{Note that hypotheses adding even more covariates, e.g.~$H_0^{T} : \delta_{g,t}^T = 0 ~\forall~g > t = 0 ~\forall~g > t$, hold under all three sets of sufficient assumptions and they all contain \{{\assref{ass:SWAS-staggered}, \assref{ass:no-y-dyn}}\}. Consequently, these two assumptions alone are sufficient for $H_0^T$, implying that it does not assess additional restrictions beyond those that justify $H_0^g$.} Rejecting hypotheses corresponding to more assumptions can be informative about the violated assumptions through comparisons with nested, less restrictive specifications:
\begin{itemize}
    \item If $H_0^{0}$ is rejected, but $H_0^{2}$ is not rejected $\Rightarrow$ \assref{ass:no-x-y} is likely violated.
    \item If $H_0^{2}$ is rejected, but $H_0^{g}$ is not rejected $\Rightarrow$ \assref{ass:no-x-dyn} is likely violated.
\end{itemize}

\subsection{Practical Implications}

The practically most relevant question is to what extent different outcomes of the hypothesis tests can replace or complement economic reasoning. Rejections are highly informative and provide strong evidence against specifications that rely on assumption sets containing the rejected assumptions, potentially overruling economic arguments in favor of these specifications.
The implications of non-rejections are more ambiguous, and using them to justify the credibility of post-treatment results requires additional arguments: 

\textbf{$H_0^{0}$ not rejected:} Arguing in favor of specifications without time-varying covariates based on a non-rejection of $H_0^{0}$ is relatively credible, as identification of post-treatment effects and pre-treatment PPT relies on the same set of assumptions (see final row of Table \ref{tab:min-vas-did}).\footnote{This observations is closely related to the discussion in Section II.A in \textcite{ghanem_when_2026}.} A remaining, well-documented obstacle is statistical in nature: failure to reject may reflect insufficient power rather than validity of the null \parencite[e.g.][]{roth_pretest_2022}. Thus, arguing for credibility based on not rejected PPT still requires arguments that the tests are sufficiently powered. 

\textbf{$H_0^{0}$ rejected, $H_0^{2}$/$H_0^{g}$ not rejected:} Credibility of $\delta_{g,t}^t$ and $\delta_{g,t}^2$ is at least not rejected by the data. However, only a subset of the assumptions that are sufficient for their unbiasedness is tested. In particular, both estimands require to assume \{\assref{ass:x-d-ordering},\assref{ass:no-d-x}\} to be unbiased. These assumptions are by construction untestable using pre-treatment trends, as treatment-covariate feedback unfolds only post-treatment. Therefore, statistical arguments that violations of the tested assumptions would be detected should be complemented by economic arguments supporting the absence of treatment-covariate feedback to rule out wrong world control bias when applying $\delta_{g,t}^t$ or $\delta_{g,t}^2$.

\textbf{$H_0^{0}/H_0^{2}$ rejected, $H_0^{g}$ not rejected:} $\delta_{g,t}^t$ is the only option for which at least parts of the sufficient assumptions are not rejected by the data. However, assessing its credibility still requires reasoning about the plausibility of the untested assumptions \{{\assref{ass:x-d-ordering}, \assref{ass:no-d-x}}\}.

\textbf{$H_0^{g}$ rejected:} This is the most informative but also the most destructive testing outcome. It cleanly rejects at least one of the baseline assumptions \{\assref{ass:SWAS-staggered},\assref{ass:no-y-dyn}\} that are required to justify all VAS implied by Table \ref{tab:min-vas-did} and Proposition \ref{prop:all-vas}. Thus, no specification with or without time-varying covariates is expected to deliver credible results. Importantly, even if $H_0^{2}$ and/or $H_0^{0}$ are not rejected, this is likely due to bias cancelling and/or power differences and does not change this conclusion. Justification of (conditional) DiD requires then arguments beyond additive separability and causal structure.

\begin{remark} \label{rem:attgg-is-special}
    Recall that short-term effects $ATT(g,g)$ require only to defend $\assref{ass:x-d-ordering}$ additionally to \assref{ass:SWAS-staggered} and \assref{ass:no-y-dyn}, and might be justified by the causal ordering within period $t$ alone. Thus, defending credibility of short-term compared dynamic effects might be less demanding.
\end{remark}

\subsection{Practical Recommendations}

We conclude with some practical recommendations for researchers who consider using time-varying covariates in their conditional DiD that are informed by the previous findings and discussions:

\textbf{Graphically represent domain knowledge:} Drawing the DAG, SWIG and/or $\Delta$-SWIG based on domain knowledge facilitates communication and reasoning about the causal structure. Once a structure is established, Table \ref{tab:min-vas-did} in combination with Proposition \ref{prop:all-vas} provides the sets of good control variables that follow directly from the structure.

\textbf{Test pre-trends with all pre-treatment covariates:} Rejecting parallel pre-trends conditional on $\overline{X}_{g-1}$ rejects standard identification arguments involving additive separability and absence of outcome dynamics. Conditional DiD should then be justified by alternative arguments or discarded.

\textbf{Consider three specifications:} Running conditional DiD (i) without time-varying covariates, (ii) with $\overline{X}_t$ and (iii) with $\{X_{g-1}, X_t\}$ provides a principled way to gather statistical evidence and may help to narrow down plausible specifications. In particular, rejected parallel pre-trends for some or all specifications are informative about crucial assumptions, as discussed in the previous two sections.

\textbf{Use full covariate sequence only for computational convenience:} Controlling for full covariate sequence $\overline{X}$ instead of pre-outcome sequence $\overline{X}_t$ does not require weaker assumptions. This means that the covariate dimension would be much larger than required for early $t$. However, applying the pre-outcome sequence requires to run a different specification for each $t$, e.g. leading to $\mathcal{T}$ propensity scores to be estimated. Running one specification with $\overline{X}$ can therefore be attractive from a computational perspective.

% E.g. testing PPT with an invalid VAS is completely uniformative as it does not test components that are needed for the identification.

% How covariates are different in DiD: Covariates must not necessarily affect treatment. Suffices if both are affected by U
% Finding PPTs with covariates that are not also a VAS for effect or at least use the minimal VAS under the same R is uninformative.

\section{Concluding Remarks} \label{sec:conclusion}

This paper introduces $\Delta$-SWIGs as a graphical tool to reason about conditional parallel trends in a transparent and principled manner.
They provide a natural starting point to scrutinize and justify control variables in DiD analyses based on the causal structure but without distributional or parametric assumptions. The only functional form assumption that must be defended outside of the graph is the single world additive separability that generalizes a standard functional form assumption in the literature. A transparent tool to reason about this substantive assumption is still missing.

We apply $\Delta$-SWIGs to study time-varying control variables in a model class with unrestricted unobservables and binary staggered treatment. However, the $\Delta$-SWIG technology can similarly be applied to model classes with more variable categories, restrictions on unobservables, beyond staggered treatment, or beyond binary treatments. For the latter two, the single world additive separability most likely must be strengthened to hold for all treatment levels that serve as control group at some point, thereby restricting treatment effect heterogeneity. A detailed discussion along these lines is left for future research.

We also abstract from estimation and testing issues in this paper. One interesting follow-up question is regarding the interplay between different valid adjustment sets and overlap, which most likely interacts with the question of the most efficient valid adjustment set. Another interesting direction is an exhaustive conceptual and statistical study of the testable implications of the model structure, and the development of powerful tests building on them.

Finally, the paper reinforces recent reservations about relying solely on parallel pre-trends to justify a DiD analysis. In particular, when time-varying covariates are used in the analysis, they are not informative about all required assumptions for unbiased post-treatment estimates, even in large data sets.

% \newpage

\printbibliography

\newpage
\appendix 
\counterwithin{equation}{section}
\counterwithin{figure}{section}
\counterwithin{table}{section}

\section{Literature Review on Conditional Parallel Trends} \label{app:literature}

We conducted a systematic literature search using Google Scholar, searching for all papers published in the \emph{American Economic Review (AER)} in 2024 and 2025 that contained the words "Parallel Trends". 
The result of this search were 25 papers (after excluding 5 papers in \emph{AER: Insights}).
In addition, we excluded the following papers.

\textcite{gadenne_inkind_2024} rely on instrumental variables (IV) as their main identification strategy and refer to parallel trends to support exogeneity of their instrument.
\textcite{karahan_demographic_2024} also use IV as their identification strategy and refer to "parallel trends" only in their descriptives section.
\textcite{lyubich_role_2025} uses a "mover design" to identify place effects. "Parallel trends" is referred to as a potential method for ruling out endogenous moves but not assumed in the paper.
\textcite{kirsi_does_2025} shows up in the Google Scholar search but not in the AER. In addition, "parallel trends" is only mentioned in the descriptives section.
\textcite{goldsmith-pinkham_contamination_2024} is a methodological paper on OLS regression. Their replications of earlier studies do not implement Difference-in-Differences or assume parallel trends.
\textcite{benetton_mortgage_2025} also use IV and only mention "parallel trends"  in the descriptives section.

After excluding these, we end up with 19 papers, which assume parallel trends.
Out of these 19 papers, 14 (74\%) use covariates, either directly in their regression specifications or in a matching/reweighting step before estimating their effects of interest.
Of these papers 13 (68\%) use time-varying covariates. 
%Figure \ref{fig:jel_pt} shows, which initial letters of JEL codes (a) all 19 papers and (b) the 13 papers that use time-varying controls refer to.

% \begin{figure}[t]
%   \centering
%     \caption{Distribution of JEL code initial letters.}
%     \label{fig:jel_pt}
%   \begin{subfigure}[t]{0.49\textwidth}
%     \centering
%     \includegraphics[width=\textwidth]{figure_all.pdf}
%     \caption{All studies}
%   \end{subfigure}
%   \hfill
%   \begin{subfigure}[t]{0.49\textwidth}
%     \centering
%     \includegraphics[width=\textwidth]{figure_Xt.pdf}
%     \caption{Studies with time-varying controls}
%   \end{subfigure}
% \end{figure}

\section{Graph Details}\label{sec:steps_detail}

\subsection{$d$-Separation Definition and Example} \label{sec:app-dsep}

\begin{definition}\label{def:dep} (\textbf{$d$-separation}, following \textcite{peters_elements_2017} and \textcite{pearl_causality_2009})

    Given a DAG $\mathcal{G}$ on a set of variables $\mathbf{V}$, with $\mathbf{X},\mathbf{Y},\mathbf{Z} \subseteq \mathbf V$ and disjoint. $\mathbf{X}$ and $\mathbf{Y}$ are \textbf{$d$-separated} by (or given) $\mathbf{Z}$ if for all $X \in \mathbf{X}$ and $Y \in \mathbf{Y}$,
    all paths between $X$and $Y$ are blocked. 
    
    A path is blocked if it contains at least one node $V_j$ such that one of the following holds:
\begin{itemize}
    \item[(a)] $V_j$ is a collider on the path, i.e.,
    $V_{j-1} \rightarrow V_j \leftarrow V_{j+1}$, and $V_j \notin \mathbf Z$ and
    none of the descendants of $V_j$ are in $\mathbf Z$;
    \item[(b)] $V_j$ is a non-collider on the path appearing in one of the
    structures
    $V_{j-1} \rightarrow V_j \rightarrow V_{j+1}$,
    $V_{j-1} \leftarrow V_j \rightarrow V_{j+1}$, or
    $V_{j-1} \leftarrow V_j \leftarrow V_{j+1}$, and $V_j \in \mathbf Z$.
\end{itemize}

We denote this by $\mathbf{X} \indep_\mathcal{G} \mathbf{Y} \mid \mathbf{Z}$. If there is an unblocked path between any $X \in \mathbf{X}$ and $Y \in \mathbf{Y}$, $\mathbf{X}$ and $\mathbf{Y}$ are \textbf{d-connected} given $\mathbf{Z}$, denoted $\mathbf{X} \centernot \indep_{\!\! \mathcal{G}} \mathbf{Y} \mid \mathbf{Z}$.
\end{definition}

Given that SWIGs are a particular instance of DAGs, $d$-separation applies identically to variables $\mathbf{V}(\mathbf d)$ in SWIG $\mathcal{G}(\mathbf d)$.

Figure \ref{fig:ex-dsep} provides an instructive example to showcase arguments that are frequently but implicitly applied in the complex structures of the main text. To $d$-separate $X$ and $Y$, we need to block three paths $X \rightarrow A \rightarrow Y$, $X \leftarrow B \rightarrow Y$, and $X \rightarrow C \leftarrow D \rightarrow Y$. The first two paths can only be blocked by conditioning on $\{A\}$ and $\{B\}$, respectively, following part (b) of Definition \ref{def:dep}. However, there are three ways to block the third path that contains a collider node: the empty set $\varnothing$ blocks the path due to part (a) of Definition \ref{def:dep}, $\{C\}$ blocks the path because both part (a) and (b) hold, but also $\{C,D\}$ blocks the path due to part (b). Taken together, we can state that $X \indep_{\mathcal{G}} Y \mid A,B\textcolor{gray}{,D}$ or $X \indep_{\mathcal{G}} Y \mid A,B,D\textcolor{gray}{,C}$ where the gray variables are optional once the black variables are conditioned on. Most importantly this showcases how we can, but do not have to, condition on a collider node on a longer path as long as the path is blocked by a non-collider node.

\begin{figure}[h]
\centering
\begin{tikzpicture}[scale=0.7, every node/.style={scale=1}]
\tikzset{
  ell/.style={draw, fill=white, inner xsep=5pt, inner ysep=5pt, line width=1.5pt},
  myarrow/.style={->, >=stealth, line width=1.5pt}
}

% --- Nodes ---
\node[name=y, ell, shape=ellipse] {$X$};
\node[name=z, ell, shape=ellipse] at ($(y)+(8cm,0cm)$) {$Y$};

\node[name=a, ell, shape=ellipse] at ($(y)+(4cm,2.5cm)$) {$A$};
\node[name=b, ell, shape=ellipse] at ($(y)+(4cm,0cm)$) {$B$};

% C between Y and B, D between B and Z (same level)
\node[name=c, ell, shape=ellipse] at ($(y)+(2cm,-2.5cm)$) {$C$};
\node[name=d, ell, shape=ellipse] at ($(y)+(6cm,-2.5cm)$) {$D$};

% --- Edges ---
% Path 1: X <- A -> Y
\draw[myarrow] (y) -- (a);
\draw[myarrow] (a) -- (z);

% Path 2: X -> B -> Y
\draw[myarrow] (b) -- (y);
\draw[myarrow] (b) -- (z);

% Path 3: X -> C <- D -> Y
\draw[myarrow] (y) -- (c);
\draw[myarrow] (d) -- (c);
\draw[myarrow] (d) -- (z);

\end{tikzpicture}
\caption{Example DAG}
\label{fig:ex-dsep}
\end{figure}

\subsection{Formalization of Edge Cancelling} \label{app:edge-cancel}

The structural equation $V_j = f_{V_j}(Pa(V_j),U_{V_j})$ in DAG $\mathcal{G}$ becomes $V_j(\mathbf{d}) = f_{V_j(\mathbf{d})}(Pa_{\mathcal{G}(\mathbf{d})}(V_j(\mathbf{d})) \setminus \mathbf{d} ,U_{V_j})$ in SWIG $\mathcal{G}(\mathbf{d})$. The subscript in the parent operator emphasizes that the parents of $V_j(\mathbf{d})$ in $\mathcal{G}(\mathbf{d})$ might themselves be relabeled. Further, note that the new function $f_{V_j(\mathbf{d})}$ does not depend on $\mathbf{d}$. The structural equation underlying the new node described in step 6 of Procedure \ref{proc:d-swig} can then be expressed as $\Delta_{j,k}(\mathbf{d}) = V_j(\mathbf{d}) - V_k(\mathbf{d}) =  f_{V_j(\mathbf{d})}(Pa_{\mathcal{G}(\mathbf{d})}(V_j(\mathbf{d})) \setminus \mathbf{d},U_{V_j}) - f_{V_k(\mathbf{d})}(Pa_{\mathcal{G}(\mathbf{d})}(V_k(\mathbf{d})) \setminus \mathbf{d},U_{V_k}) =: f_{\Delta_{j,k}(\mathbf{d})}(P_{j,k} ,U_{V_j},U_{V_k})$ where $P_{j,k} := (Pa_{\mathcal{G}(\mathbf{d})}(V_j(\mathbf{d}))  \setminus \mathbf{d}) \cup (Pa_{\mathcal{G}(\mathbf{d})}(V_k(\mathbf{d}))  \setminus \mathbf{d})$ is the union of the parents of the level nodes. If there exists a function $h$ such that $f_{\Delta_{j,k}(\mathbf{d})}(P_{j,k} ,U_{V_j},U_{V_k}) = h(P_{j,k} \setminus R,U_{V_j},U_{V_k})$, i.e.~$f_{\Delta_{j,k}(\mathbf{d})}$ does not vary with $R$, then $R$ is redundant in the difference node equation and the edges from $R$ to the difference node can be omitted.

\subsection{Additional Definitions} \label{app:definitions}

\begin{definition}(Markov property, \textcite{peters_elements_2017})\label{def:markov}
    We say that a probability distribution $P(\mathbf{V})$ satisfies the Markov property with respect to a DAG $\mathcal{G}$ if the following three equivalent statements holds:
    \begin{itemize}
        \item[(a)] Each variable $V \in \mathbf{V}$ is independent of its non-descendants conditional on its parents.
        \item[(b)] The distribution $P(\mathbf{V})$ factorizes according to the graph $\mathcal{G}$:
        \begin{equation*}
            P(\mathbf{V}) = \prod_{V_j \in \mathbf{V}}P(V_j \mid Pa_{\mathcal{G}}(V_j))
        \end{equation*}

        \item[(c)] For $\mathbf{X},\mathbf{Y},\mathbf{Z}$ disjoint subsets of $\mathbf{V}$:
        \begin{equation*}
            \mathbf{X} \indep_{\mathcal{G}} \mathbf{Y} \mid \mathbf{Z} \Rightarrow \mathbf{X} \indep \mathbf{Y} \mid \mathbf{Z} 
        \end{equation*}
    \end{itemize}

The proof equivalence of the statements can be found in \textcite{lauritzen_graphical_1996}.
    
\end{definition}

\begin{definition}(Subgraph of SWIG $\mathcal{G}(\mathbf{d})$ on $\mathbf{V}(\mathbf{d})$ obtained by removing all fixed nodes $\mathbf{d}$)

    We denote the subgraph of $\mathcal{G}(\mathbf{d})$ on $\mathbf{V}(\mathbf{d})$ obtained by removing all fixed nodes $\mathbf{d}$ by $(\mathcal{G}(\mathbf{d}))_{\mathbf{V}(\mathbf{d})}$.
\end{definition}

\begin{remark} \label{rem:subgraph_markov}
    Theorem 12 of \textcite{richardson_single_2013} shows that: if the probability distribution $P(\mathbf{V}(\mathbf{d}))$ factorizes according to the SWIG $(\mathcal{G}(\mathbf{d}))$, $d$-separation, with fixed nodes included in the conditioning set, implies conditional independence.
    Section 3.5 of \textcite{richardson_single_2013} links the notion of factorization according to a SWIG to the standard Markov property as stated in Definition \ref{def:markov}.
    Factorizing according to a SWIG is equivalent to satisfying the standard Markov property with respect to the subgraph $(\mathcal{G}(\mathbf{d}))_{\mathbf{V}(\mathbf{d})}$ of the SWIG, thus, $d$-separation in  $(\mathcal{G}(\mathbf{d}))_{\mathbf{V}(\mathbf{d})}$ implies conditional independence.
    In addition, Proposition 14 of \textcite{richardson_single_2013} shows that $d$-separation among the random nodes in $(\mathcal{G}(\mathbf{d}))_{\mathbf{V}(\mathbf{d})}$ is equivalent to $d$-separation in  $\mathcal{G}(\mathbf{d})$ given the fixed nodes. 
    Thus, for $\mathbf{X},\mathbf{Y},\mathbf{Z}$ disjoint subsets of $\mathbf{V}(\mathbf{d})$ :
    \begin{equation*}
        \mathbf{X} \indep_{(\mathcal{G}(\mathbf{d}))_{\mathbf{V}(\mathbf{d})}} \mathbf{Y} \mid \mathbf{Z} \Leftrightarrow  \mathbf{X} \indep_{\mathcal{G}(\mathbf{d})} \mathbf{Y} \mid \mathbf{Z} \cup \mathbf{d} \Rightarrow \mathbf{X} \indep \mathbf{Y} \mid \mathbf{Z}.
    \end{equation*}    
Proposition 11 of \textcite{richardson_single_2013} shows that the distribution $P(\mathbf{V}(\mathbf{d}))$ factorizes according to the SWIG $\mathcal{G}(\mathbf{d})$ if the data is generated by the underlying SCM.

\end{remark}

\begin{definition}\label{def:sink}
    (Sink)
A \textbf{sink} is a variable in a structural causal model and node in the corresponding graph that is not in the structural equation for any other variable and thus the node has no descendants. A \textbf{deterministic sink} is a sink that is a deterministic function of the variables represented by its parents.  
\end{definition}

\newpage

\subsection{Proof of Theorem \ref{thm:Markov}}\label{sec:thm_markov_proof}

The structure is as follows:

We start with a probability distribution $P(\mathbf{V})$ that satisfies the Markov property with respect to a DAG $\mathcal G$ which corresponds to a causal model $\mathcal{M}$. We show that we can add a deterministic sink (Definition \ref{def:sink}) to this model and graph and 
\begin{itemize}
    \item The graph $\mathcal{G}'$ is still acyclic (Lemma \ref{lemma:acyclic}).
    \item The probability distribution on the variables including the sink satisfies the Markov property with respect to the DAG $\mathcal{G}'$ (Lemma \ref{lemma:Augmented_Markov}). Thus, $d$-separation in $\mathcal{G}'$ implies conditional independence.
    \item Finally, we show that we can prune a graph by deleting sinks that are not part of the $d$-separation relation of interest without $d$-separation among the other nodes being affected (\ref{lemma:sink}). 
\end{itemize}

The final proof of Theorem \ref{thm:Markov} links these general results on causal graphs to the steps involved in constructing a $\Delta$-SWIG as outlined in Procedure \ref{proc:d-swig}.

\begin{lemma} \label{lemma:acyclic}
    Given a DAG $\mathcal{G}$ and structural model $\mathcal{M}$ on $\mathbf{V}$, the augmented graph $\mathcal{G}'$ and model $\mathcal{M}'$ on $\mathbf{V}\cup \{S\}$, with $S$ a deterministic sink, is acyclic.
\end{lemma}

\begin{proof}
    The DAG and model are augmented by the structural equation for $S$ and the node in the graph.
    The equation is:
    \begin{equation*}
        S \defeq f_S(Pa_{\mathcal{G}'}(S)), \ \ Pa_{\mathcal{G}'}(S) \subseteq \mathbf{V}.
    \end{equation*}
    Graphically, it the node $S$ has only incoming edges. The DAG $\mathcal{G}$ is acyclic by assumption.
    $\mathcal{G}'$ contains additional edges and paths involving $S$.
    All paths without $S$ are also contained in $\mathcal{G}$ and thus acyclic.
    Suppose there is a directed cycle in $\mathcal{G}'$. It has to involve $S$.
    Such a path would be:
    \begin{equation*}
        V_j \rightarrow ... \rightarrow S \rightarrow ...\rightarrow V_j
    \end{equation*}

    This is a contradiction to $S$ being a sink, as it requires an outgoing edge/a child of $S$.
    
\end{proof}

A notable special case of this is that a full $\Delta$-SWIG, created from an acyclic SWIG is acyclic, as the $\Delta Y_{t}(\overline{0})$-node (or, in fact, any difference node) is a sink. 
A full $\Delta$-SWIG with several $\Delta$-nodes is acyclic by repeatedly invoking Lemma \ref{lemma:acyclic}.

\begin{lemma}\label{lemma:Augmented_Markov}
    Let $\mathcal{G}$ be a DAG and $\mathcal{M}$ a structural causal model with variables and corresponding nodes $\mathbf{V}$, $P(\mathbf{V})$ satisfies the Markov property with respect to $\mathcal{G}$.
    Let the augmented DAG and structural model $\mathcal{G}'$ and $\mathcal{M}'$ be the result of adding a deterministic sink $S$. The probability distribution is $P(\mathbf{V'})$ with $\mathbf{V'} \defeq \mathbf{V}\cup \{S\}$.
    $P(\mathbf{V'})$ satisfies the Markov property with respect to $\mathcal{G}'$.
    
\end{lemma}

\begin{proof}

The proof builds on the insight that the deterministic sink $S$ has the structural equation $S \defeq f_S(pa_{\mathcal{G}'}(S))$.
By assumption, $P(\mathbf{V})$ factorizes as:

\begin{equation}\label{eq:Markov_original}
  P(\mathbf{V}) = \prod_{V_j \in \mathbf{V}} P(V_j \mid Pa_\mathcal{G}(V_j))
\end{equation}

Then, the joint distribution of $\mathbf{V'} \defeq \mathbf{V} \cup S$ can be written as:
\begin{align*}
    P(\mathbf{V'}) &= P(\mathbf{V},S) = P(S \mid \mathbf{V}) P(\mathbf{V}) & (\text{factorization}) \\
                    &= P(S \mid Pa_{\mathcal{G}'}(S)) P(\mathbf{V}) &  \text{($S$ is a deterministic function of $Pa_{\mathcal{G}'}(S)) \subseteq \mathbf{V}$)} \\
                    &= P(S \mid Pa_{\mathcal{G}'}(S)) \prod_{V_j \in \mathbf{V}} P(V_j \mid Pa_\mathcal{G}(V_j)) & \text{(\ref{eq:Markov_original})} \\
                    &= P(S \mid Pa_{\mathcal{G}'}(S)) \prod_{V_j \in \mathbf{V}} P(V_j \mid Pa_{\mathcal{G}'}(V_j)) & \text{($Pa_{\mathcal{G}}(V_j)$ same as $Pa_{\mathcal{G}'}(V_j)$ for $V_j \in \mathbf{V}$)} \\
                    &= \prod_{V_j \in \mathbf{V}'} P(V_j \mid Pa_{\mathcal{G}'}(V_j))
\end{align*}
         
\end{proof}

Again, a special case is the SWIG and full $\Delta$-SWIG. The distribution $P(\mathbf{V}(\overline{0}))$ factorizes according to (is Markov to) $(\mathcal{G}(\overline{0}))_{\mathbf{V}(\overline{0})}$ by Assumption as stated in Theorem \ref{thm:Markov}.
When creating a difference node and the corresponding structural equation, the new variable is a deterministic function of (a subset of) the parents of the variables it was created from. With the edges drawn according to the functional dependencies, it is a deterministic function of its parents in the subgraph $(\mathcal{G}_\Delta^{full}(\overline{0}))_{\mathbf{V'}(\overline{0})}$ of the full $\Delta$-SWIG $\mathcal{G}_\Delta^{full}(\overline{0})$.
Following Lemma \ref{lemma:Augmented_Markov}, the distribution $P(\mathbf{V'}(\overline{0}))$ with $\mathbf{V'}(\overline{0}) = \mathbf{V}(\overline{0}) \cup \Delta Y_{t}(\overline{0})$ factorizes according to the subgraph $(\mathcal{G}_\Delta^{full}(\overline{0}))_{\mathbf{V'}(\overline{0})}$ of the full $\Delta$-SWIG $\mathcal{G}_\Delta^{full}(\overline{0})$.
If we add another $\Delta$-node $\Delta Y_{t'}(\overline{0})$, the distribution $P(\mathbf{V''}(\overline{0}))$ with $\mathbf{V''}(\overline{0}) = \mathbf{V'}(\overline{0}) \cup \Delta Y_{t'}(\overline{0})$  factorizes according to the subgraph of the augmented full $\Delta$-SWIG.
This follows from Lemma \ref{lemma:Augmented_Markov} with the $\Delta$-SWIG with one $\Delta$-node as the starting point.

To show that we can prune the full $\Delta$-SWIG, we additionally make use of the following Lemma:

\begin{lemma}\label{lemma:sink}
    Let $\mathcal{G}$ be a DAG containing nodes corresponding to variables $\mathbf{V}$ and $P(\mathbf{V})$ satisfy the Markov property with respect to $\mathcal{G}$, in particular, that $d$-separation implies conditional independence. 
    Let $\mathbf X, \mathbf Y, \mathbf Z$ be disjoint subsets of $\mathbf V$.
    In addition, $S \in \mathbf V$ is a sink and $S \notin \mathbf X \cup \mathbf Y \cup \mathbf Z $.
    Let $\Tilde{\mathcal{G}}$ be a subgraph containing $\mathbf{\Tilde{V}} = \mathbf{V}\setminus \{S\}$, obtained by removing $S$ and all its incoming edges.
    Then:
    \begin{equation*}
        \mathbf X \indep_{\Tilde{\mathcal{G}}} \mathbf Y \mid \mathbf Z \Leftrightarrow \mathbf X \indep_{\mathcal{G}} \mathbf Y \mid \mathbf Z
    \end{equation*}
    and
    \begin{equation*}
        \mathbf X \indep_{\Tilde{\mathcal{G}}} \mathbf Y \mid \mathbf Z \Rightarrow \mathbf X \indep \mathbf Y \mid \mathbf Z.
    \end{equation*}
    
\end{lemma}

\begin{proof} The proof follows from the definition of $d$-separation and the fact that $\mathcal{G}$ differs from $\Tilde{\mathcal{G}}$ only by paths that contain $S$ and its incoming edges.

\begin{itemize}
    \item $\mathbf X \indep_{\Tilde{\mathcal{G}}} \mathbf Y \mid \mathbf Z \Rightarrow \mathbf X \indep_{\mathcal{G}} \mathbf Y \mid \mathbf Z$

    For $\mathbf X \indep_{\Tilde{\mathcal{G}}} \mathbf Y \mid \mathbf Z$, (by the definition of $d$-separation), for all $X \in \mathbf X$ and $Y \in \mathbf Y$ in all paths $P = (X - ... - Y)$ that connect the two in $\Tilde{\mathcal{G}}$, there is:

\begin{itemize}
    \item[a)] a chain $V_l \rightarrow V_m \rightarrow V_n$ or a fork $V_l \leftarrow V_m \rightarrow V_n$, where the middle node $V_m \in \mathbf Z$, or:
    \item[b)] a collider $V_l \rightarrow V_m \leftarrow V_n$ where $(\{V_m\} \cup Desc_{\Tilde{\mathcal{G}}}(V_m)) \cap  \mathbf Z = \varnothing$, or both.
\end{itemize}

The graph $\mathcal{G}$ additionally contains the node $S$ and its incoming edges. All paths between $\mathbf{X}$ and $\mathbf{Y}$ from $\Tilde{\mathcal{G}}$ are also in $\mathcal{G}$.

Suppose that $\mathbf X \centernot\indep_{\!\!\! \mathcal{G}} \mathbf Y \mid \mathbf Z$.
This means that in $\mathcal{G}$, there is at least one unblocked path between some $X \in \mathbf{X}$ and $Y \in \mathbf{Y}$.
A path is unblocked if a) every collider on the path or one of its descendants $\in \mathbf{Z}$ and b) no other node (the middle node of a chain or fork) $\in \mathbf{Z}$.

$S$ is a sink, thus, for any $X \in \mathbf{X}$ and $Y \in \mathbf{Y}$ it can only be: i) a collider on a path connecting $X$ and $Y$, that is in $\mathcal{G}$ but not in $\Tilde{\mathcal{G}}$, ii) a descendant of a collider on a path connecting $X$ and $Y$ or iii) a descendant of any another variable.
As $S$ is a sink, every path between $X$ and $Y$ that contains $S$ as an interior node contains it as a collider. As $S \notin \mathbf{Z}$, and $S$ has no descendants that could be in $\mathbf{Z}$, these paths are blocked.
The paths that are not blocked by $\mathbf{Z}$ must be paths not containing $S$.
Any path that does not contain $S$ is also in $\Tilde{\mathcal{G}}$. 
The existence of an unblocked path between $X$ and $Y$ in $\mathcal{G}$ implies the existence of an unblocked path in $\Tilde{\mathcal{G}}$.
Thus, $\mathbf X \centernot\indep_{\!\!\! \mathcal{G}} \mathbf Y \mid \mathbf Z \Rightarrow \mathbf X \centernot\indep_{\!\!\!\Tilde{\mathcal{G}}} \mathbf Y \mid \mathbf Z$ which is equivalent to $\mathbf X \indep_{\Tilde{\mathcal{G}}} \mathbf Y \mid \mathbf Z \Rightarrow \mathbf X \indep_{\mathcal{G}} \mathbf Y \mid \mathbf Z$.

    \item $\mathbf X \indep_{\mathcal{G}} \mathbf Y \mid \mathbf Z \Rightarrow \mathbf X \indep_{\Tilde{\mathcal{G}}} \mathbf Y \mid \mathbf Z$

We prove that  $\mathbf X \centernot\indep_{\!\!\!\Tilde{\mathcal{G}}} \mathbf Y \mid \mathbf Z \Rightarrow \mathbf X \centernot\indep_{\!\!\! \mathcal{G}} \mathbf Y \mid \mathbf Z$.
Note that $\mathbf X \centernot\indep_{\!\!\!\Tilde{\mathcal{G}}} \mathbf Y \mid \mathbf Z$ means, there is at least one unblocked path between some $X$ and $Y$, $P_u$, such that, for all chains and forks on $P_u$, $V_m \notin \mathbf Z$ and all colliders on $P_u$,  $V_m \in \mathbf Z$ or some descendant of $V_m \in \mathbf Z$.
All paths that exist in $\Tilde{\mathcal{G}}$, also exist in $\mathcal{G}$. 
Thus, if there is a path in $\Tilde{\mathcal{G}}$ that is unblocked by $\mathbf{Z}$, it is also in $\mathcal{G}$ and $\mathbf{X}$ and $\mathbf{Y}$ are not $d$-separated given $\mathbf{Z}$.

\item By the Markov property $\mathbf X \indep_{\mathcal{G}} \mathbf Y \mid \mathbf Z \Rightarrow \mathbf X \indep \mathbf Y \mid \mathbf Z$. As $\mathbf X \indep_{\Tilde{\mathcal{G}}} \mathbf Y \mid \mathbf Z \Leftrightarrow \mathbf X \indep_{\mathcal{G}} \mathbf Y \mid \mathbf Z$ we get:
\begin{equation*}
            \mathbf X \indep_{\Tilde{\mathcal{G}}} \mathbf Y \mid \mathbf Z \Rightarrow \mathbf X \indep \mathbf Y \mid \mathbf Z.
\end{equation*}

\end{itemize}
    
\end{proof}

Now we can restate Theorem \ref{thm:Markov} from the main text.

\MarkovTheorem*

\begin{proof} Using the lemmas above. 
\begin{itemize}
    \item[(a)] If $\mathcal{G}_\Delta(\mathbf{d})$ is the full $\Delta$-SWIG:
    
    $\mathcal{G}_\Delta^{full}(\mathbf{d})$ is acyclic: We assume the original structural causal model $\mathcal{M}$ and corresponding graph $\mathcal{G}$ to be acyclic. The SWIG $\mathcal{G}(\mathbf{d})$ is acyclic, by Lemma \ref{lemma:acyclic}, $\mathcal{G}_\Delta^{full}(\mathbf{d})$ is acyclic, as the full $\Delta$-SWIG is obtained by adding several deterministic sinks. 
    The subgraph $(\mathcal{G}_\Delta^{full}(\mathbf{d}))_{\mathbf{V'}(\mathbf{d})}$ is also acyclic as a subgraph of an acyclic graph is always acyclic. To see this, suppose $(\mathcal{G}_\Delta^{full}(\mathbf{d}))_{\mathbf{V'}(\mathbf{d})}$ contained a cyclic graph. As it is a subgraph of $\mathcal{G}_\Delta^{full}(\mathbf{d})$, i.e., it contains only a subset of the nodes and edges, the cyclic graph would also exist in $\mathcal{G}_\Delta^{full}(\mathbf{d})$.

    The distribution $P(\mathbf{V}_\Delta(\mathbf{d}))$ satisfies the Markov property with respect to $(\mathcal{G}_\Delta^{full}(\mathbf{d}))_{\mathbf{V}_\Delta(\mathbf{d})}$: 
    $P(\mathbf{V}(\mathbf{d}))$ satisfies the Markov property with respect to $(\mathcal{G}(\mathbf{d}))_{\mathbf{V}(\mathbf{d})}$, with $U_{V_j}, U_{V_k}$ explicit, by Assumption.  
    By Lemma \ref{lemma:Augmented_Markov}, $P(\mathbf{V}_\Delta(\mathbf{d}))$ satisfies the Markov property with respect to $(\mathcal{G}_\Delta^{full}(\mathbf{d}))_{\mathbf{V}_\Delta(\mathbf{d})}$, with $\mathbf{V}_\Delta(\mathbf{d}) = \mathbf{V}(\mathbf{d}) \cup \Delta_{j,k}(\mathbf{d})$ as $\Delta_{j,k}(\mathbf{d})$ is a sink that is deterministically determined by its parents in $(\mathcal{G}_\Delta^{full}(\mathbf{d}))_{\mathbf{V}_{\Delta}(\mathbf{d})}$.
    With several $\Delta$-nodes, repeatedly invoke Lemma \ref{lemma:Augmented_Markov}.

    So, we can conclude that $d$-separation in $(\mathcal{G}_\Delta^{full}(\mathbf{d}))_{\mathbf{V}_\Delta(\mathbf{d})}$ implies conditional independence. The statement in the Theorem follows from the insight that $d$-separation in $(\mathcal{G}_\Delta^{full}(\mathbf{d}))_{\mathbf{V}_\Delta(\mathbf{d})}$ is equivalent to $d$-separation in $\mathcal{G}_\Delta^{full}(\mathbf{d})$ given in addition the sequence of fixed nodes $\mathbf{d}$ \parencite[Proposition 14 of ][]{richardson_single_2013}.

\item[(b)] If $\mathcal{G}(\mathbf{d})$ is the pruned $\Delta$-SWIG, i.e., $\mathbf{V}_\Delta(\mathbf{d})$ does not contain the sinks $S$ that are removed in the pruning step: 

    We know from Part (a), that $d$-separation in $\mathcal{G}_\Delta^{full}(\mathbf{d})$ implies conditional independence in $P(\mathbf{V}_\Delta(\mathbf{d}),\mathbf{S})$.
    The setting is thus the same as in Lemma \ref{lemma:sink}.
    The statement of the Theorem follows from (repeatedly) using Lemma \ref{lemma:sink}, i.e., $d$-separation in  $\mathcal{G}_\Delta^{prune}(\mathbf{d})$, which results from removing one or several sinks is equivalent to $d$-separation in $\mathcal{G}_\Delta^{full}(\mathbf{d})$. $\mathbf{X,Y,Z}$ do not contain any of the deleted sinks by definition, as the sinks are not contained in $\mathbf{V}_{\Delta}(\mathbf{d})$.
\end{itemize}

\end{proof}

\section{Post-Treatment Covariates} \label{app:post-treatment}

We show here that the causal structure with post-treatment covariates depicted in Figure \ref{fig:DAG_DSWIG_post} can not be modified by adding nodes or edges such that the resulting $\Delta$-SWIG implies the two independencies
\begin{align}
    \Delta Y_1(0) &\indep D \mid X_0, X_1(0) \label{eq:caetano_CPT_post} \\
    X_1(0) & \indep D \mid X_0,Y_0 \label{eq:caetano_cov_unconf_post}
\end{align}
without implying at the same time
\begin{equation}
     X_1(0)  \indep D \mid X_0. \label{eq:caetano_cov_unconf}
\end{equation}

First, note that obtaining \eqref{eq:caetano_cov_unconf_post} but not \eqref{eq:caetano_cov_unconf} would require $d$-separation of $D$ and $X_1(0)$ given $X_0$ and $Y_0$, but no $d$-separation given only $X_0$ (and the fixed node $0$):
\begin{align}
    X_1(0) & \indep_{\mathcal G_{\Delta} (0)} D \mid X_0,Y_0,0 \label{eq:caetano_cov_unconf_post_dsep} \\
    X_1(0) & \centernot \indep_{\! \mathcal G_{\Delta} (0)} D \mid X_0,0 \label{eq:caetano_cov_unconf_post_nodsep}
\end{align}

Both can only be true if there is at least one path connecting $D$ and $X_1(0)$ that is unblocked by $X_0,0$, but blocked by $X_0, Y_0,0$. This requires that $Y_0$ appears on \textit{all} such unblocked paths as a non-collider.

The following Lemma is useful for the rest of the argument:

\begin{lemma}\label{lemma:post_treatment_controls}
    Consider the path $P_u = D - ... -Y_0 - ... - X_1(0)$ with $Y_0$ an non-collider.\footnote{We use the notational convention that $A-...-B$ denotes the paths connecting nodes $A$ and $B$, where $-$ denotes an edge that can point towards either $A$ or $B$ and $...$ means that there could be additional nodes and edges on the path.}
    Let $0$ denote a fixed node (i.e., it has no incoming edges).
    If $P_u$ is unblocked given $X_0,0$ then \emph{at least one of the following} is true:
    \begin{itemize}
    \item[(a)]  $D-...\leftarrow Y_0$ is unblocked given $X_0, X_1(0),0$
    \item[(b)] $X_0 \in Desc(Y_0)$
    \item[(c)] $X_1(0) \in Desc(Y_0)$
\end{itemize}
\end{lemma}

\begin{proof}
    $Y_0$ is a non-collider on $P_u$. $P_u$ thus contains one of the local patterns $\leftarrow Y_0 \rightarrow$, $\leftarrow Y_0 \leftarrow$ or $\rightarrow Y_0 \rightarrow$. We proceed by discussing each case noting that a path is unblocked only if every subpath (part) of it is unblocked, as blocking a path \textit{once} suffices to block the path:
    \begin{itemize}
        \item $\leftarrow Y_0 \rightarrow$:    
    Consider the unblocked subpath $D-...\leftarrow Y_0$ of $P_u$ that does not contain $X_1(0)$ by construction. Thus, $D-...\leftarrow Y_0$ is also unblocked by $X_0,X_1(0),0$.

    For unblocked subpath $Y_0 \rightarrow ... -X_1(0)$ there are two scenarios: (i) This subpath could have all arrows pointing in the same direction. Then, $X_1(0) \in Desc(Y_0)$.
    (ii) If not all arrows point in the same direction, there has to be a collider that is a descendant of $Y_0$ because of the initial $Y_0 \rightarrow$. The subpath is only unblocked given $X_0,0$ if either $X_0$ is the collider or $X_0$ is a descendant of the collider (as the fixed node has no incoming edges, it can never be the descendant of any other node). Both imply that $X_0 \in Desc(Y_0)$. 

    $\Rightarrow$ (a) and (b) or (a) and (c) of Lemma \ref{lemma:post_treatment_controls} are true.

    \item $\leftarrow Y_0 \leftarrow$:
    By the same argument as in the first case, $D-...\leftarrow Y_0$ is unblocked by $X_0, X_1(0),0$. 
    
    $\Rightarrow$ (a) of Lemma \ref{lemma:post_treatment_controls} is true.

    \item $\rightarrow Y_0 \rightarrow$:
    By the same argument as in the first case, $Y_0 \rightarrow ... -X_1(0)$ unblocked by $X_0,0$ implies either  $X_0 \in Desc(Y_0)$ or $X_1(0) \in Desc(Y_0)$. 
    
    $\Rightarrow$ Either (b) or (c) of Lemma \ref{lemma:post_treatment_controls} is true.
    \end{itemize}
For each of the three possible local patterns at least one of (a), (b), or (c) is true, which concludes the proof.
\end{proof}

(\ref{eq:caetano_CPT_post}) would be implied by:
\begin{equation}
    \Delta Y_1(0) \indep_{\mathcal G_{\Delta} (0)} D \mid X_0, X_1(0),0\label{eq:caetano_CPT_post_dsep}
\end{equation}

First, note that a $\Delta$-SWIG in two time periods with unobserved confounding always contains the path $P = D \leftarrow U \rightarrow Y_0 \leftarrow U_{Y_0} \rightarrow \Delta Y_1(0)$.
The subpath $D \leftarrow U \rightarrow Y_0$ represents the unobserved confounding between $D$ and $Y_0$. The subpath $Y_0 \leftarrow U_{Y_0} \rightarrow \Delta Y_1(0)$ exists by construction, as exogenous variable $U_{Y_0}$ is inherited by $\Delta Y_1(0)$.

(\ref{eq:caetano_cov_unconf_post_dsep}) and (\ref{eq:caetano_cov_unconf_post_nodsep}) imply the existence of $P_u$ as in Lemma \ref{lemma:post_treatment_controls} and thus at least one of (a), (b) or (c) in the Lemma. 
This leads to a contradiction with (\ref{eq:caetano_CPT_post_dsep}).

If (a) is true, then due to the path $Y_0 \leftarrow U_{Y_0} \rightarrow \Delta Y_1(0)$ (subpath of $P$ above), there is a path between $D$ and $\Delta Y_1(0)$ that is unblocked given $X_0,X_1(0),0$. This path is obtained by concatenating $D - ... \leftarrow Y_0$ and $Y_0 \leftarrow U_{Y_0} \rightarrow \Delta Y_1(0)$ and noting that concatenating two unblocked paths by a non-collider creates an unblocked path, i.e.~\eqref{eq:caetano_CPT_post_dsep} is not true.

Alternatively, consider (b) or (c) being true. In either case, the path $P$ between $D$ and $\Delta Y_1(0)$ would be unblocked given $X_0,X_1(0),0$, as this would condition on descendants of the collider $Y_0$. So, (\ref{eq:caetano_CPT_post_dsep}) is again not true.

This shows that (\ref{eq:caetano_cov_unconf_post_dsep}), (\ref{eq:caetano_cov_unconf_post_nodsep}) and (\ref{eq:caetano_CPT_post_dsep}) can not be true at the same time.
Either $Y_0$ is required for $d$-separation of $D$ and $X_1(0)$, implying covariate unconfoundedness conditional on $X_0, Y_0$, but then $\Delta Y_1(0)$ and $D$ are not $d$-separated given $X_0,X_1(0),0$, i.e.~CPT is not fulfilled.
Or, $Y_0$ is a neutral control for covariate unconfoundedness meaning that $X_1(0)$ and $D$ are $d$-separated given \textit{only} $X_0,0$ such that infeasible CPT \eqref{eq:caetano_CPT_post} still holds and collapses to feasible $\Delta Y_1(0) \indep D \mid X_0$ as in the main text.

% \newpage

\section{Difference-in-Differences}
\subsection{Single World Additive Separability in the Literature} \label{app:swas-lit}

Earlier references assume ``all worlds'' additive separability (see, e.g.~Equation 18 in \textcite{angrist_empirical_1999} or Equation 34 in \textcite{imbens_recent_2009}) with constant individual treatment effects instead of leaving them unrestricted. The first reference we could find that leaves individual treatment effects unrestricted and only assumes single world additive separability is \textcite{blundell_alternative_2009}. The more recent references directly state the untreated potential outcomes and usually leave the treated potential outcomes unrestricted by not imposing a model for them.

Table \ref{tab:SWAS_lit} collects examples of Assumption \assref{ass:SWAS-staggered} in the literature and shows how their functions map to  $\alpha(U,C)$ and $g_{Y_t}(\cdot)$ in our notation in the final two columns. The third column shows which assumptions are either explicitly discussed in the reference or implied by its stated assumptions. Assumption \assref{ass:no-y-dyn} is not deduced from the functional forms following the discussion in Footnote \ref{fn:implied-r}.

\begin{landscape}

\begin{table}[h!]
\centering \small
\caption{Examples of Single World Additive Separability in the DiD literature}
\label{tab:SWAS_lit}
\begin{tabular}{@{}p{7cm} p{5cm} p{3.5cm} p{1.5cm} p{6cm}@{}}
\hline
\textbf{Reference} & \textbf{Equation} & \textbf{Assumptions} & $\alpha(U,C)$ & $g_{Y_t}(\cdot)$ \\ % &\textbf{Exact place in the reference} \\
\hline
\textcite{blundell_alternative_2009} & $y_{it} = \beta + \alpha_i d_{it} + u_{it}, \ \E[u_{it}\mid d_i,t] = \E[\eta_i \mid d_i] + m_t$ & \assref{ass:SWAS-staggered}, \assref{ass:no-x-y}, \assref{ass:no-x-dyn} & $\beta + \eta_i$ & $ \underbrace{\E[g_{Y_t}(U_{Y_t})]}_{m_t} + g_{Y_t}(U_{Y_t}) - E[g_{Y_t}(U_{Y_t})] $\\ %  & Equation (13)  \\
\textcite{arkhangelsky_causal_2024} & $Y_{it}(0) = \alpha_i + \beta_t + \varepsilon_{it}$ & \assref{ass:SWAS-staggered}, \assref{ass:no-x-y}, \assref{ass:no-x-dyn}   & $\alpha_i$  &  $\underbrace{\E[g_{Y_t}(U_{Y_t})]}_{\beta_t} + \underbrace{g_{Y_t}(U_{Y_t}) - \E[g_{Y_t}(U_{Y_t})]}_{\varepsilon_{it}}$\\ % & Assumption 5.1\\
[1em]
\textcite{dechaisemartin_credible_2023} & $Y_{g,t}(\mathbf{0}_t) = \alpha_g + \gamma_t + \varepsilon_{g,t}$ & \assref{ass:SWAS-staggered}, \assref{ass:no-x-y}, \assref{ass:no-x-dyn} & $\alpha_g$  &  $\underbrace{\E[g_{Y_t}(U_{Y_t})]}_{\gamma_t} + \underbrace{g_{Y_t}(U_{Y_t}) - \E[g_{Y_t}(U_{Y_t})]}_{\varepsilon_{g,t}}$\\ %  & Equation (2.4) \\
[1em]
\textcite{roth_whats_2023} & $Y_{i,t}(0) = \alpha_i + \phi_t + \epsilon_{i,t}$ & \assref{ass:SWAS-staggered}, \assref{ass:no-x-y}, \assref{ass:no-x-dyn} & $\alpha_i$ & $ \underbrace{\E[g_{Y_t}(U_{Y_t})]}_{\phi_t} + \underbrace{g_{Y_t}(U_{Y_t}) - E[g_{Y_t}(U_{Y_t})]}_{\epsilon_{i,t}}$\\ %  & Discussion after Assumption 1 \\
[1em]
\textcite{caetano_differenceindifferences_2024} & $Y_{it}(0) = g_t(Z_i, X_{it}) + \eta_i  + e_{it}$ & \assref{ass:SWAS-staggered}, \assref{ass:x-d-ordering}, \assref{ass:no-d-x}, \assref{ass:no-x-dyn} & $\eta_i$ & $   \underbrace{g_{Y_t}(Z, X_{t}, U_{Y_t})}_{g_t(Z_i, X_{it}) + e_{it}}$\\ %  & Section on "Models that rationalize parallel trends assumptions" \\
[1em]
\textcite{ghanem_selection_2024} & $Y_{it} (0) = \alpha_i + \lambda_t + \varepsilon_{it}$ & \assref{ass:SWAS-staggered}, \assref{ass:no-x-y}, \assref{ass:no-x-dyn} &  $\alpha_i$  & $ \underbrace{\E[g_{Y_t}(U_{Y_t})]}_{\lambda_t} + \underbrace{g_{Y_t}(U_{Y_t}) - \E[g_{Y_t}(U_{Y_t})]}_{\varepsilon_{it}}$\\ %  & Equation (5) and section 3.3.5 (02/2026 version) \\
& $Y_{it} (0) = \alpha_i + \chi_t(X_{it}) +\lambda_t + \varepsilon_{it}$ & \assref{ass:SWAS-staggered}, \assref{ass:x-d-ordering}, \assref{ass:no-d-x}, \assref{ass:no-x-dyn} &  $\alpha_i$ & $\underbrace{g_{Y_t}(X_t, U_{Y_t}) }_{ \chi_t(X_{it}) + \lambda_t + \varepsilon_{it}}$\\ %   & Appendix D.2 \\
% & $Y_{it}(0) = \mu(X_{it}^\mu, \alpha_i^\mu, \varepsilon_{it}^\mu) + \lambda_t(X_{it}^\lambda, \alpha_i^\lambda, \varepsilon_{it}^\lambda)$  &  & $\mu(X_{it}^\mu, \alpha_i^\mu, \varepsilon_{it}^\mu)$, w/ $(X_{i1}^\mu,\varepsilon_{i1}^\mu) = (X_{i2}^\mu,\varepsilon_{i2}^\mu)$ & $\lambda_t(X_{it}^\lambda, \alpha_i^\lambda, \varepsilon_{it}^\lambda)$\\ % &   Assumption NSP-X \\
% \textcite{abadie_semiparametric_2005} & $Y_{i,t} = \delta_t + \alpha D_{i,t} + \eta_i + v_{i,t}$ & Equation (1)\\
%      & $Y_{i,t} = \mu + X_i' \pi_t + \tau \cdot D_{i,1} + \delta \cdot t + \alpha \cdot D_{i,t} + \varepsilon_{i,t}$ & Equation (8)\\
% \textcite{lechner_estimation_2011} & $\E[Y_t(0) \mid X = x, D = d] = \alpha + t \delta^0 + d \gamma + x \beta + t x \lambda^0 + dx \pi$ &  & $\underbrace{\E[\alpha(U,X)\mid X = x, D = d]}_{\alpha + d \gamma + x (\beta + d \pi)x} + \varepsilon^\alpha  $ & $ \underbrace{\E[g_{Y_t}(X,U_{Y_t})\mid X = x, D = d]}_{t\delta^0  + t x \lambda^0 } + \varepsilon^{g_{Y_t}} $\\ %  & Equation (3.8) \\
% & $\E_{U \mid D = d}\E[Y_t(0) \mid D = d, U = u]= \E[Y_t(0) \mid D = d] + \E_{U \mid D = d}f^d(U)$ & Equation (3.5) &  &  \\
 [1em]
 \textcite{borusyak_revisiting_2024} & $\E[Y_{it}(0)] = \alpha_i + \beta_t$ & \assref{ass:SWAS-staggered}, \assref{ass:no-x-y}, \assref{ass:no-x-dyn}  & $\alpha_i$ & $ \underbrace{\E[g_{Y_t}(U_{Y_t})]}_{\beta_t} + \underbrace{g_{Y_t}(U_{Y_t}) - \E[g_{Y_t}(U_{Y_t})]}_{\varepsilon_{it}}$\\ %  & Assumption 1 (Parallel Trends)  \\
 [1em]
 \textcite{deng_doubly_2026} & $Y_t(\infty) = f_t(Z,Z_t) + U + \varepsilon_t$ & \assref{ass:SWAS-staggered} & $U$ & $f_t(Z,Z_t) + \varepsilon_t$\\ %  & Motivation for Assumption 2 (Parallel Trends) \\
    % \textcite{gardner_twostage_2025} & $\E[Y_{it}(0) \mid \{ X_{it} \}_{t = 1}^T, G_i = g] = \lambda_i + X_{it}' \gamma$ & \assref{ass:SWAS-staggered}, \assref{ass:x-d-ordering}, \assref{ass:no-d-x}, \assref{ass:no-x-dyn} & $\underbrace{\E[\alpha(U) \mid \{ X_{t} \}_{t = 1}^T, G = g]}_{\lambda_i} + \varepsilon_{t}^\alpha$ & $\underbrace{\E[g_{Y_t}(X_{t}, U_{Y_t})\mid X_t]}_{X_{it}'\gamma} + \varepsilon_{t}^{g_{Y_t}}$\\ %  &Assumption 1   \\
\hline
\end{tabular}
\end{table}

\end{landscape}
\ref{ass:SWAS-staggered}
\newpage

\subsection{Beyond Staggered Treatment} \label{app:general-swas}

While it is not relevant in this paper, we state a single world additive separability for the general case without staggered treatment for potential future reference. The complication is now that the individual treatment effect may be different for every possible treatment sequence, including those starting as being treated in period 0 or those switching out of treatment.

\begin{assumption} \label{ass:SWAS-general}
    (single world additive separability - general) \\
For all $t = 0,...,\mathcal{T}$, the structural equation for $Y_t$ has the form
\begin{align}
    Y_t \defeq  f_{Y_t}(Pa(Y_t), U_{Y_t}) & =\alpha (U,C) + g_{Y_t}(Pa(Y_t) \setminus \{ U, \overline{D}_t \}, U_{Y_t}) + \sum_{s \in \overline{\mathcal{D}}_t \setminus \{\overline{0}_t\}} \mathbf{1}[\overline{D}_t = s] \cdot \tau_{t,s}(Pa(Y_t)\setminus \overline{D}_t, U_{Y_t})
\end{align}
where $U$ is a set of unobserved, time-invariant confounders, $C \defeq \left(\bigcap_{t = 0}^\mathcal{T} Pa(Y_t) \right) \setminus \{\overline{D}, U \}$ are all time-invariant parents of $Y_t$, and $\overline{\mathcal{D}}_t := \{0,1\}^{t+1}$ is the set of all possible treatment paths until $t$.
\end{assumption}

\subsection{Identification Proof of $ATT(g,t)$ in General Setting} \label{app:attgt-ident}

\subsubsection{Identification Proof of $ATT(g,t)$ - Never-Treated}

Assume that conditional independence (CIA) $\Delta Y_{g-1,t} (\overline{0}) \indep \underline{D}_{g} \mid \mathbf{Z}, \overline{D}_{g-1} = \overline{0}_{g-1}$ holds. Then,    \begin{align*}
       DiD_{g,t}(\mathbf{Z}) &  = \E[Y_t - Y_{g-1}\mid \overline{D}_{g-1} = \overline{0}_{g-1}, \underline{D}_g = \underline{1}_g]  \\
       & \quad - \E[\E[Y_t - Y_{g-1}\mid \mathbf{Z}, \overline{D} = \overline{0}]\mid \overline{D}_{g-1} = \overline{0}_{g-1}, \underline{D}_g = \underline{1}_g ] \\
       &  = \E[Y_t(\overline{0}_{g-1},\underline{1}_g) - Y_{g-1} (\overline{0}_{g-1}\textcolor{gray}{,\underline{1}_g})\mid \overline{D}_{g-1} = \overline{0}_{g-1}, \underline{D}_g = \underline{1}_g] \\
       & \quad - \E[\E[Y_t(\overline{0}) - Y_{g-1}(\overline{0}) \mid \mathbf{Z}, \overline{D} = \overline{0}]\mid \overline{D}_{g-1} = \overline{0}_{g-1}, \underline{D}_g = \underline{1}_g ] & (consistency) \\
       &  = \E[Y_t(\overline{0}_{g-1},\underline{1}_g) - Y_{g-1} (\overline{0}_{g-1}\textcolor{gray}{,\underline{1}_g})\mid \overline{D}_{g-1} = \overline{0}_{g-1}, \underline{D}_g = \underline{1}_g]  \\
       & \quad - \E[\E[Y_t(\overline{0}) - Y_{g-1}(\overline{0}) \mid \mathbf{Z}, \overline{D}_{g-1} = \overline{0}_{g-1}, \underline{D}_g = \underline{1}_g]\mid \overline{D}_{g-1} = \overline{0}_{g-1}, \underline{D}_g = \underline{1}_g ]  & (CIA) \\
       &  = \E[Y_t(\overline{0}_{g-1},\underline{1}_g) - Y_{g-1} (\overline{0}_{g-1}\textcolor{gray}{,\underline{1}_g})\mid \overline{D}_{g-1} = \overline{0}_{g-1}, \underline{D}_g = \underline{1}_g]  \\
       & \quad - \E[Y_t(\overline{0}) - Y_{g-1}(\overline{0}) \mid \overline{D}_{g-1} = \overline{0}_{g-1}, \underline{D}_g = \underline{1}_g ]  \\
       &  = \E[Y_t(\overline{0}_{g-1},\underline{1}_g) - Y_{g-1} (\overline{0}_{g-1}\textcolor{gray}{,\underline{1}_g})\mid \overline{D}_{g-1} = \overline{0}_{g-1}, \underline{D}_g = \underline{1}_g]  \\
       & \quad - \E[Y_t(\overline{0}) - Y_{g-1}(\overline{0}_{g-1}\textcolor{gray}{,\underline{0}_g}) \mid \overline{D}_{g-1} = \overline{0}_{g-1}, \underline{D}_g = \underline{1}_g ]  \\
       &  = \E[Y_t(\overline{0}_{g-1},\underline{1}_g) - Y_t(\overline{0}) \mid \overline{D}_{g-1} = \overline{0}_{g-1}, \underline{D}_g = \underline{1}_g ] \\
       & = ATT(g,t) \\
       if~ t<g & = \E[Y_t(\overline{0}_{g-1}\textcolor{gray}{,\underline{1}_g}) - Y_t(\overline{0}_{g-1}\textcolor{gray}{,\underline{0}_g}) \mid \overline{D}_{g-1} = \overline{0}_{g-1}, \underline{D}_g = \underline{1}_g ] = 0
    \end{align*}

\subsubsection{Identification Proof of $ATT(g,t)$ - Not-Yet-Treated} \label{app:attgt-ident-nyt}

Let $max(g,t) \leq s$ and assume that conditional independence (CIA-NYT)~$\Delta Y_{g-1,t} (\overline{0}) \indep \overline{D}_{g,s} \mid \mathbf{Z}, \overline{D}_{g-1} = \overline{0}_{g-1}$ holds. Then,    \begin{align*}
       DiD_{g,t}^{NYT}(\mathbf{Z}) &  = \E[Y_t - Y_{g-1}\mid \overline{D}_{g-1} = \overline{0}_{g-1}, \underline{D}_g = \underline{1}_g]  \\
       & \quad - \E[\E[Y_t - Y_{g-1}\mid \mathbf{Z}, \overline{D}_{s} = \overline{0}_{s}]\mid \overline{D}_{g-1} = \overline{0}_{g-1}, \underline{D}_g = \underline{1}_g ] \\
       &  = \E[Y_t(\overline{0}_{g-1},\underline{1}_g) - Y_{g-1} (\overline{0}_{g-1}\textcolor{gray}{,\underline{1}_g})\mid \overline{D}_{g-1} = \overline{0}_{g-1}, \underline{D}_g = \underline{1}_g] \\
       & \quad - \E[\E[Y_t(\overline{0}_{s},\textcolor{gray}{\underline{0}_{s+1}}) - Y_{g-1}(\overline{0}_{s},\textcolor{gray}{\underline{0}_{s+1}}) \mid \mathbf{Z}, \overline{D}_{s} = \overline{0}_{s}]\mid \overline{D}_{g-1} = \overline{0}_{g-1}, \underline{D}_g = \underline{1}_g ] & (consistency) \\
       &  = \E[Y_t(\overline{0}_{g-1},\underline{1}_g) - Y_{g-1} (\overline{0}_{g-1}\textcolor{gray}{,\underline{1}_g})\mid \overline{D}_{g-1} = \overline{0}_{g-1}, \underline{D}_g = \underline{1}_g]  \\
       & \quad - \E[\E[Y_t(\overline{0}) - Y_{g-1}(\overline{0}) \mid \mathbf{Z}, \overline{D}_{g-1} = \overline{0}_{g-1}, \overline{D}_{g,s} = \overline{1}_{g,s}]\mid \overline{D}_{g-1} = \overline{0}_{g-1}, \underline{D}_g = \underline{1}_g ]  & (CIA-NYT) \\
        &  = \E[Y_t(\overline{0}_{g-1},\underline{1}_g) - Y_{g-1} (\overline{0}_{g-1}\textcolor{gray}{,\underline{1}_g})\mid \overline{D}_{g-1} = \overline{0}_{g-1}, \underline{D}_g = \underline{1}_g]  \\
       & \quad - \E[\E[Y_t(\overline{0}) - Y_{g-1}(\overline{0}) \mid \mathbf{Z}, \overline{D}_{g-1} = \overline{0}_{g-1}, \underline{D}_g = \underline{1}_g]\mid \overline{D}_{g-1} = \overline{0}_{g-1}, \underline{D}_g = \underline{1}_g ]  & (\assref{ass:staggered-trmnt}) \\
       &  = \E[Y_t(\overline{0}_{g-1},\underline{1}_g) - Y_{g-1} (\overline{0}_{g-1}\textcolor{gray}{,\underline{1}_g})\mid \overline{D}_{g-1} = \overline{0}_{g-1}, \underline{D}_g = \underline{1}_g]  \\
       & \quad - \E[Y_t(\overline{0}) - Y_{g-1}(\overline{0}) \mid \overline{D}_{g-1} = \overline{0}_{g-1}, \underline{D}_g = \underline{1}_g ]  \\
       &  = \E[Y_t(\overline{0}_{g-1},\underline{1}_g) - Y_{g-1} (\overline{0}_{g-1}\textcolor{gray}{,\underline{1}_g})\mid \overline{D}_{g-1} = \overline{0}_{g-1}, \underline{D}_g = \underline{1}_g]  \\
       & \quad - \E[Y_t(\overline{0}) - Y_{g-1}(\overline{0}_{g-1}\textcolor{gray}{,\underline{0}_g}) \mid \overline{D}_{g-1} = \overline{0}_{g-1}, \underline{D}_g = \underline{1}_g ]  \\
       &  = \E[Y_t(\overline{0}_{g-1},\underline{1}_g) - Y_t(\overline{0}) \mid \overline{D}_{g-1} = \overline{0}_{g-1}, \underline{D}_g = \underline{1}_g ] \\
       & = ATT(g,t) \\
       if~ t<g & = \E[Y_t(\overline{0}_{g-1}\textcolor{gray}{,\underline{1}_g}) - Y_t(\overline{0}_{g-1}\textcolor{gray}{,\underline{0}_g}) \mid \overline{D}_{g-1} = \overline{0}_{g-1}, \underline{D}_g = \underline{1}_g ] = 0.
    \end{align*}

\newpage

\section{Proofs in Multiple Time Periods}

\subsection{Additional Notation Convention}

We make use of the following notational shortcut for convenience: $\overline{D}_t(\overline{0}_{t-1}) = \{D_1,D_2(0),...,D_t(\overline{0}_{t-1}) \}$ with $\overline{D}_1(\overline{0}_{0}) = D_1$ and $\overline{D}_0(\overline{0}_{0-1}) = \varnothing$.

\subsection{Auxiliary Lemmata}

\begin{lemma} \label{lemma:d_sep_neutral_controls}
    Let $\mathcal{G}(\mathbf{d})$ be a SWIG containing nodes $\mathbf{V}(\mathbf{d}) \cup \mathbf{d}$.
    $\mathbf{Z}, \mathbf{R}, \mathbf{N}$ are pairwise disjoint subsets of $\mathbf{V}(\mathbf{d})$. 
    Let $V \in \mathbf{V}(\mathbf{d})$ be a random node, $V \notin \mathbf{Z} \cup \mathbf{R} \cup \mathbf{N}$. $V$ is a sink. The set of its parents fulfills $Pa_{\mathcal{G}(\mathbf{d})}(V) \cap \mathbf{R} = \varnothing$ and all of its parents $V_p \in Pa_{\mathcal{G}(\mathbf{d})}(V)$ fulfill one of the following criteria:
\begin{itemize}
    \item[(a)]  $V_p \in \mathbf{Z}$. 
    \item[(b)]  $V_p$ is a fixed (split) node: $V_p \in \mathbf{d}$
    \item[(c)]  $V_p$ has no parents and all its children are sinks $S_j \notin \mathbf{R}$. In addition:
    \begin{equation*}
        S_j  \in \mathbf{Z} \cup \mathbf{N} \Rightarrow Pa_{\mathcal{G}(\mathbf{d})}(S_j) \setminus \{V_p\} \subseteq (\mathbf{Z} \cup \mathbf{d})
    \end{equation*}
\end{itemize}
    Then, 
\begin{equation}\label{eq:d_sep_min}
    V \indep_{\mathcal{G}(\mathbf{d})} \mathbf{R} \mid \mathbf{Z} , \mathbf{d}
\end{equation}
and
\begin{equation}\label{eq:d_sep_min_plus_neutral}
    V \indep_{\mathcal{G}(\mathbf{d})} \mathbf{R} \mid \mathbf{Z} ,\mathbf{d} , \mathbf{N}.
\end{equation}
\end{lemma}

\begin{proof}
    To show that (\ref{eq:d_sep_min}) holds, check all paths that may connect $V$ and $\mathbf{R}$. As $V$ is a sink all paths are $... V_p \rightarrow V$, i.e., $V_p \in Pa_{\mathcal{G}(\mathbf{d})}(V)$ are non-colliders. As $V_p \notin \mathbf{R}$, any path has to continue past $V_p$ and thus contain it as a non-collider.

    If $V_p$ is as in (a), the paths are blocked as $\mathbf{Z}$ is in the adjustment set.

    If $V_p$ as in (b), blocked as $V_p \in \mathbf{d}$ and $\mathbf{d}$ is in the conditioning set for the $d$-separation statement.

    If $V_p$ as in (c), blocked due to the sink $S_j$ if $S_j \notin \mathbf{Z}$. 
    To see this, note that any path connecting $V$ and $\mathbf{R}$ with $V_p$ as in (c) has to contain: 
    \begin{equation}\label{eq:path_parents_sinks}
        V \leftarrow V_p \rightarrow S_j \leftarrow V_j
    \end{equation}
    This is because (1) $V$ is a sink, (2) $V_p$ has no parents, (3) any child of $V_p$ is a sink $S_j \notin \mathbf{R}$, so the path has to continue with a collider $S_j$. 
    With $S_j \notin \mathbf{Z}$, the path is blocked.\footnote{Note that a collider $S_j$ can not be in the set $\mathbf{d}$, as fixed node contain no incoming edges by construction. In addition, as $S_j$ is a sink, it has no descendants that could be conditioned on.}
    If $S_j \in \mathbf{Z}$, then $V_j \in \mathbf{Z} \cup \mathbf{d}$. As $\mathbf{Z}$ and $\mathbf{R}$ are disjoint, $V_j \notin   \mathbf{R}$. The path has to continue beyond (\ref{eq:path_parents_sinks}) but is blocked due to the non-collider $V_j \in (Pa_{\mathcal{G}(\mathbf{d})}(S_j) \setminus \{V_p\}) \subseteq (\mathbf{Z} \cup \mathbf{d})$.

For (\ref{eq:d_sep_min_plus_neutral}), recall that (\ref{eq:d_sep_min}) holds but now suppose:
\begin{equation*}
     V \centernot \indep_{\! \! \mathcal{G}(\mathbf{d})} \mathbf{R} \mid \mathbf{Z} \cup \mathbf{d} \cup \mathbf{N}.
\end{equation*}
This can only hold if there is at least one path that is blocked given $\mathbf{Z} \cup \mathbf{d}$ but unblocked given $\mathbf{Z} \cup \mathbf{d} \cup \mathbf{N}$.
All paths contain the non-collider $V_p$, where $V_p$ fulfills one of (a), (b) or (c).
Suppose the unblocked path contains $V_p$ as in (a) or (b). $V_p$ is a non-collider and $V_p \in \mathbf{Z,d}$. The path is blocked. Contradiction.

Now, suppose the unblocked path contains the non-collider $V_p$ as in (c). The path is blocked given $\mathbf{Z,d}$.
If $S_j \notin \mathbf{Z} \cup \mathbf{R} \cup \mathbf{N}$, the path is blocked by the collider $S_j$.
If $S_j  \in \mathbf{Z} \cup \mathbf{N}$ then $(Pa_{\mathcal{G}(\mathbf{d})}(S_j) \setminus \{V_p\}) \subseteq (\mathbf{Z} \cup \mathbf{d})$. The path is blocked due to the non-collider $V_j \in \mathbf{Z} \cup \mathbf{d}$.
Hence every path remains blocked after adding $\mathbf{N}$, contradicting the supposition.

\end{proof}

\begin{lemma} \label{lemma:D_sequence_contraction}
    Under consistency and Assumption \assref{ass:staggered-trmnt}, for $t''\geq t'$:
\begin{align*}
    \Delta Y_{g-1,t}(\overline{0}) & \indep D_{s}(\overline{0}_{s-1}) \mid Z, \overline{D}_{s-1}(\overline{0}_{s-2}), \ \forall \  1 \leq s \leq \mathcal{T}  \\
    &\Rightarrow \Delta Y_{g-1,t}(\overline{0}) \indep \overline{D}_{t',t''} \mid Z, \overline{D}_{t'-1} = \overline{0}_{t'-1}. 
\end{align*}   

and as a special case:

\begin{equation*}
    \Delta Y_{g-1,t}(\overline{0}) \indep \underline{D}_{t'} \mid Z, \overline{D}_{t'-1} = \overline{0}_{t'-1}.
\end{equation*}

\end{lemma}

\begin{proof}
    Proof by induction. For the base step, start with:
\begin{align}
     \Delta Y_{g-1,t}(\overline{0})  &\indep D_{t'}(\overline{0}_{t'-1}) \mid Z, \overline{D}_{t'-1}(\overline{0}_{t'-2}) & \notag \\
    \Rightarrow \Delta Y_{g-1,t}(\overline{0})  &\indep D_{t'}(\overline{0}_{t'-1}) \mid Z, \overline{D}_{t'-1} = \overline{0}_{t'-1} & (cons.) \notag \\
    \Rightarrow \Delta Y_{g-1,t}(\overline{0})  &\indep D_{t'} \mid Z, \overline{D}_{t'-1} = \overline{0}_{t'-1} & (cons.) \label{eq:induc_base_t}
\end{align}
and
\begin{align}
     \Delta Y_{g-1,t}(\overline{0})  &\indep D_{t'+1}(\overline{0}_{t'}) \mid Z, \overline{D}_{t'}(\overline{0}_{t'-1}) & \notag \\
    \Rightarrow \Delta Y_{g-1,t}(\overline{0})  &\indep D_{t'+1}(\overline{0}_{t'}) \mid Z, \overline{D}_{t'} = \overline{0}_{t'} & (cons.) \notag \\
    \Rightarrow \Delta Y_{g-1,t}(\overline{0})  &\indep D_{t'+1} \mid Z, \overline{D}_{t'} = \overline{0}_{t'} & (cons.)
\end{align}

In addition, for $D_{t'} = 1$, $D_{t'+1}$ is a constant ($=1$). Thus,

\begin{equation*}
    \Delta Y_{g-1,t}(\overline{0})  \indep D_{t'+1} \mid Z, \overline{D}_{t'-1} = \overline{0}_{t'-1}, D_{t'} = 1
\end{equation*}

also holds. This implies

\begin{equation}
    \Delta Y_{g-1,t}(\overline{0})  \indep D_{t'+1} \mid Z, \overline{D}_{t'-1} = \overline{0}_{t'-1}, D_{t'} \label{eq:induc_base_tp1}
\end{equation}

and by contraction: 

\begin{equation}
    (\ref{eq:induc_base_t}) \land (\ref{eq:induc_base_tp1}) \Rightarrow \Delta Y_{g-1,t}(\overline{0})  \indep \overline{D}_{t',t'+1} \mid Z, \overline{D}_{t'-1} = \overline{0}_{t'-1}
\end{equation}

Induction Step. For $t'' > t'$, take: 

\begin{equation}
    \Delta Y_{g-1,t}(\overline{0})  \indep \overline{D}_{t',t''} \mid Z, \overline{D}_{t'-1} = \overline{0}_{t'-1} \label{eq:induc_first}
\end{equation}

and, by consistency:

\begin{align}
     \Delta Y_{g-1,t}(\overline{0}) & \indep D_{t''+1}(\overline{0}_{t''}) \mid Z, \overline{D}_{t''}(\overline{0}_{t''-1}) \notag \\
     \Rightarrow & \Delta Y_{g-1,t}(\overline{0}) \indep D_{t''+1} \mid Z, \overline{D}_{t''} = \overline{0}_{t''}
\end{align}

Now, write $\overline{D}_{t''} = \overline{0}_{t''}$ as $\overline{D}_{t'-1} = \overline{0}_{t'-1}, \overline{D}_{t',t''} = \overline{0}_{t',t''}$ and note that if any of $D_{t'},...,D_{t''}$ is unequal to zero, then $D_{t''+1} = 1$ by staggered adoption, and thus independent of $\Delta Y_{g-1,t}(\overline{0})$. Thus:

\begin{equation}
    \Delta Y_{g-1,t}(\overline{0}) \indep D_{t''+1} \mid Z, \overline{D}_{t'-1} = \overline{0}_{t'-1}, \overline{D}_{t',t''} .\label{eq:induc_sec}
\end{equation}

By contraction: 

\begin{equation}
    (\ref{eq:induc_first}) \land (\ref{eq:induc_sec}) \Rightarrow  \Delta Y_{g-1,t}(\overline{0})  \indep \overline{D}_{t',t''+1} \mid Z, \overline{D}_{t'-1} = \overline{0}_{t'-1}
\end{equation}

for arbitrary $t''+1\geq t'$.

So, in particular:

\begin{equation*}
    \Delta Y_{g-1,t}(\overline{0})  \indep \underline{D}_{t'} \mid Z, \overline{D}_{t'-1} = \overline{0}_{t'-1}
\end{equation*}

\end{proof}

\subsection{Proofs of the Results in the Main Text}\label{sec:proofs_porps_multiple}

We state and proof the following Lemma, which can be used to show the Propositions in the main text:

\begin{lemma}\label{lemma:parent_indep}
    Consider a SCM $\mathcal{M} \in \mathfrak{M}$ and corresponding DAG $\mathcal{G}$. Let $\mathcal{G}(\overline{0}_m)$ be the SWIG for the intervention $\overline{D}_m = \overline{0}_m$ constructed in step 3 of Procedure \ref{proc:d-swig}. 
    % Additionally, let the following statements hold:
    % \begin{itemize}
    %     \item No outcome dynamics: $Y_t$ are sinks, i.e., $Desc_{\mathcal{G}}(Y_t) = \varnothing$
    %     \item Single world additive separability: Assumption XXX
    %     \item The distribution $P(\mathbf{V}(\overline{0}_m))$) satisfies the Markov property \parencite[][or Definition \ref{def:markov} in the Appendix]{kiiveri_recursive_1984} with respect to the subgraph $(\mathcal{G}(\overline{0}_m))_{\mathbf{V}(\overline{0}_m)}$ of $\mathcal{G}(\overline{0}_m)$ that is obtained by removing all fixed nodes $\overline{0}_m$.
    % \end{itemize}
%
Let $\mathbf{B} \subseteq (\overline{X}(\overline{0}_m) \cup \overline{D} (\overline{0}_m))$ and $\mathbf{N} \subseteq (\overline{X}(\overline{0}_m)) \cup \overline{D} (\overline{0}_m)))$.
In addition: $\mathbf{B}$, $\mathbf{N}$ and $Pa_{\mathcal{G}(\overline{0}_m))}(Y_{g-1}(\overline{0}_m))) \cup Pa_{\mathcal{G}(\overline{0}_m))}(Y_{t}(\overline{0}_m)))$ are pairwise disjoint. Under Assumptions \assref{ass:SWAS-staggered} and \assref{ass:no-y-dyn} the following conditional independence holds for all $m\geq max( g-1,t)$ and $g,t \geq 1$:
\begin{equation}
    \Delta Y_{g-1,t} (\overline{0}_m) \indep \mathbf{B} \mid Pa_{\mathcal{G}(\overline{0}_m)}(Y_{g-1}(\overline{0}_m)) \setminus \{U,U_{Y_{g-1}}, \overline{0}_m\}, Pa_{\mathcal{G}(\overline{0}_m)}(Y_{t}(\overline{0}_m)) \setminus \{U,U_{Y_{t}}, \overline{0}_m\}, \textcolor{gray}{\mathbf{N}}.
\end{equation}

where the gray $\textcolor{gray}{\mathbf{N}}$ in the conditioning set indicates that it is optional. Without further restrictions, no member in $Pa_{\mathcal{G}(\overline{0}_m)}(Y_{g-1}(\overline{0}_m)) \setminus \{U,U_{Y_{g-1}}, \overline{0}_m\}, Pa_{\mathcal{G}(\overline{0}_m)}(Y_{t}(\overline{0}_m)) \setminus \{U,U_{Y_{t}}, \overline{0}_m\}$ can be removed from the conditioning set. 
    
\end{lemma}

\begin{proof}
    Let $\mathcal{G}_{\Delta}(\overline{0}_m)$ denote a $\Delta$-SWIG obtained via Procedure \ref{proc:d-swig} containing the Difference node $\Delta Y_{g-1,t}(\overline{0}_m) \defeq Y_{t}(\overline{0}_m) - Y_{g-1}(\overline{0}_m)$.

    We show that Lemma \ref{lemma:d_sep_neutral_controls} applies with $V = \Delta Y_{g-1,t}(\overline{0})$, $\mathbf{R} = \mathbf{B}$, $\mathbf{Z} = Pa_{\mathcal{G}(\overline{0}_m)}(Y_{g-1}(\overline{0}_m)) \setminus \{U,U_{Y_{g-1}}, \overline{0}_m\} \cup Pa_{\mathcal{G}(\overline{0}_m)}(Y_{t}(\overline{0}_m)) \setminus \{U,U_{Y_{t}}, \overline{0}_m\}$, $\mathbf{d} = \overline{0}_m$, $\mathbf{N} = \mathbf{N}$.

    First, note that the following properties stated in Lemma \ref{lemma:d_sep_neutral_controls} are all fulfilled:

    \begin{itemize}
        \item $\mathbf{Z} = Pa_{\mathcal{G}(\overline{0}_m)}(Y_{g-1}(\overline{0}_m)) \setminus \{U,U_{Y_{g-1}}, \overline{0}_m\} \cup Pa_{\mathcal{G}(\overline{0}_m)}(Y_{t}(\overline{0}_m)) \setminus \{U,U_{Y_{t}}, \overline{0}_m\}$, $\mathbf{R} = \mathbf{B}$,  $\mathbf{N} = \mathbf{N}$ are pairwise disjoint subsets as stated in the Lemma.
        \item By Procedure \ref{proc:d-swig}, $\Delta Y_{g-1,t}(\overline{0}_m)$ is a sink, it is not in $\mathbf{Z}$ (it can not be as it is not in $\mathcal{G} (\overline{0}_m)$). It is not in $\mathbf{N}$ or $\mathbf{B}$ (as the latter two are by definition $\subseteq (\overline{X}(\overline{0}_m) \cup \overline{D} (\overline{0}_m))$.
        \item $\mathbf{R} = \mathbf{B} \cap  Pa_{\mathcal{G}_{\Delta}(\overline{0}_m)} (\Delta Y_{g-1,t}(\overline{0}_m)) = \varnothing$ as $Pa_{\mathcal{G}_{\Delta}(\overline{0}_m)} (\Delta Y_{g-1,t}(\overline{0}_m)) \subseteq Pa_{\mathcal{G}(\overline{0}_m)}(Y_{g-1}(\overline{0}_m)) \cup Pa_{\mathcal{G}(\overline{0}_m)}(Y_{t}(\overline{0}_m))$ and the Lemma text states $\mathbf{B} \cap (Pa_{\mathcal{G}(\overline{0}_m)}(Y_{g-1}(\overline{0}_m)) \cup Pa_{\mathcal{G}(\overline{0}_m)}(Y_{t}(\overline{0}_m))) = \varnothing$.
    \end{itemize}

    Under Assumption \assref{ass:SWAS-staggered}, as $\mathcal{G}_{\Delta}(\overline{0}_m)$ is obtained via Procedure \ref{proc:d-swig} and $m \geq t,g-1$, the functional dependency of $\Delta Y_{g-1,t}(\overline{0}_m)$ on $U$ cancels and thus the corresponding edge can be removed:
    \begin{align*}
        Pa_{\mathcal{G}_{\Delta}(\overline{0}_m)} (\Delta Y_{g-1,t}(\overline{0}_m)) = &(Pa_{\mathcal{G}(\overline{0}_m)}(Y_{g-1}(\overline{0}_m)) \cup Pa_{\mathcal{G}(\overline{0}_m)}(Y_{t}(\overline{0}_m)) ) \setminus \{U\} \\
         = &\underbrace{Pa_{\mathcal{G}(\overline{0}_m)}(Y_{g-1}(\overline{0}_m)) \setminus \{U,U_{Y_{g-1}}, \overline{0}_m\} \cup Pa_{\mathcal{G}(\overline{0}_m)}(Y_{t}(\overline{0}_m)) \setminus \{U,U_{Y_{t}}, \overline{0}_m\}}_{\textbf{(a)}} \\
         &\cup \underbrace{(\overline{0}_m \cap (Pa_{\mathcal{G}(\overline{0}_m)}(Y_{g-1}(\overline{0}_m)) \cup Pa_{\mathcal{G}(\overline{0}_m)}(Y_{t}(\overline{0}_m)) ))}_{(b), \subseteq \overline{0}_m} \cup \underbrace{\{U_{Y_{g-1}}, U_{Y_{t}}\}}_{(c)}
    \end{align*}    

It suffices to show that all these parents fulfill one of the conditions in Lemma \ref{lemma:d_sep_neutral_controls}.

For parts (a) and (b) the mapping to Lemma \ref{lemma:d_sep_neutral_controls} is immediate. 

For (c), note that by Procedure \ref{proc:d-swig} and the facts that $Pa_\mathcal{G}(U_{Y_t}) = \varnothing$ and $Children_{\mathcal{G}}(U_{Y_t}) = \{Y_t\}$, all descendants of $U_{Y_{g-1}}, U_{Y_{t}}$ in $\mathcal{G}_{\Delta}(\overline{0}_m)$ are $Y_{g-1}(\overline{0}_m), Y_{t}(\overline{0}_m)$ or additional difference nodes. 
Under Assumption \assref{ass:no-y-dyn}, all of these nodes are sinks. Let $S$ denote any of these nodes. $S \notin \mathbf{B}$, as $\mathbf{B} \subseteq (\overline{X}(\overline{0}) \cup \overline{D} (\overline{0}))$.
In addition, note that:
\begin{align*}
    \mathbf{Z} \cup \mathbf{N} = Pa_{\mathcal{G}(\overline{0}_m)}(Y_{g-1}(\overline{0}_m)) \setminus \{U,U_{Y_{g-1}}, \overline{0}_m\} \cup Pa_{\mathcal{G}(\overline{0}_m)}(Y_{t}(\overline{0}_m)) \setminus \{U,U_{Y_{t}}, \overline{0}_m\} \cup \mathbf{N} \subseteq (\overline{X}(\overline{0}) \cup \overline{D} (\overline{0})).
\end{align*}

which follows from noticing that the sinks $Y_t(\overline{0}_m)$ can never be parents of any other node. Thus, the only nodes that can be parents of $Y_{g-1}(\overline{0}_m), Y_{t}(\overline{0}_m)$ after excluding $U, U_{Y_{t}}, U_{Y_{g-1}}, \overline{0}_m$, are $\overline{X}(\overline{0}_m)$.

This implies $S \notin \mathbf{Z} \cup \mathbf{N}$. So, all parents of $\Delta Y_{g-1,t}(\overline{0}_m)$ fulfill one of the conditions in Lemma \ref{lemma:d_sep_neutral_controls}.

Thus, by Lemma \ref{lemma:d_sep_neutral_controls}:
\begin{equation}\label{eq:d-sep_parents_res}
     \Delta Y_{g-1,t} (\overline{0}_m) \indep_{\mathcal{G}_{\Delta}(\overline{0}_m)} \mathbf{B} \mid Pa_{\mathcal{G}(\overline{0}_m)}(Y_{g-1}(\overline{0}_m)) \setminus \{U,U_{Y_{g-1}}, \overline{0}_m\}, Pa_{\mathcal{G}(\overline{0}_m)}(Y_{t}(\overline{0}_m)) \setminus \{U,U_{Y_{t}}, \overline{0}_m\}, \overline{0}_m, \textcolor{gray}{\mathbf{N}}.
\end{equation}

To see that, without further restrictions, no node in $Pa_{\mathcal{G}(\overline{0}_m)}(Y_{g-1}(\overline{0}_m)) \setminus \{U,U_{Y_{g-1}}, \overline{0}_m\}, Pa_{\mathcal{G}(\overline{0}_m)}(Y_{t}(\overline{0}_m)) \setminus \{U,U_{Y_{t}}, \overline{0}_m\}$ can be left out of the conditioning set, suppose that a node $V_p \in (Pa_{\mathcal{G}(\overline{0}_m)}(Y_{g-1}(\overline{0}_m)) \setminus \{U,U_{Y_{g-1}}, \overline{0}_m\} \cup Pa_{\mathcal{G}(\overline{0}_m)}(Y_{t}(\overline{0}_m)) \setminus \{U,U_{Y_{t}}, \overline{0}_m\}$ is removed from the conditioning set. Without further restrictions, $V_p$ might be $d$-connected to $\mathbf{B}$ given the remaining conditioning set, i.e., one can construct a SCM $\mathcal{M} \in \mathfrak{M}$ under \assref{ass:SWAS-staggered} and \assref{ass:no-y-dyn} such that $V_p$ is $d$-connected to some node in $\mathbf{B}$. The fact that $V_p \in  Pa_{\mathcal{G}_{\Delta}(\overline{0}_m)} (\Delta Y_{g-1,t}(\overline{0}_m))$ implies $\Delta Y_{g-1,t}(\overline{0}_m)$ is then not $d$-separated from $\mathbf{B}$.

By Theorem \ref{thm:Markov}, (\ref{eq:d-sep_parents_res}) implies: 

\begin{equation*}
         \Delta Y_{g-1,t} (\overline{0}_m) \indep \mathbf{B} \mid Pa_{\mathcal{G}(\overline{0}_m)}(Y_{g-1}(\overline{0}_m)) \setminus \{U,U_{Y_{g-1}}, \overline{0}_m\}, Pa_{\mathcal{G}(\overline{0}_m)}(Y_{t}(\overline{0}_m)) \setminus \{U,U_{Y_{t}}, \overline{0}_m\}, \textcolor{gray}{\mathbf{N}}.
\end{equation*}

\end{proof}

\begin{lemma}\label{lemma:min-z0_general}
    Consider SCM $\mathcal{M} \in \mathfrak{M}$ and let $\mathcal{G}(\overline{0})$ be the SWIG constructed in step 3 of Procedure \ref{proc:d-swig}.
Then, under Assumptions \assref{ass:SWAS-staggered}, \assref{ass:no-y-dyn} and consistency the following conditional independence holds for all $g,t$ and $t'\geq g$:
    \begin{align} \label{eq:cia-general-generalized-sequence-version}
    \Delta Y_{g-1,t} (\overline{0}) \indep \overline{D}_{g,t'} \mid~ & \underbrace{Pa_{\mathcal{G}(\overline{0})}(Y_{g-1}(\overline{0})) \setminus \{U,U_{Y_{g-1}}, \overline{0}\}, Pa_{\mathcal{G}(\overline{0})}(Y_{t}(\overline{0})) \setminus \{U,U_{Y_{t}}, \overline{0}\}}_{\mathbf{S}_{g,t}(\overline{0})}, \overline{D}_{g-1} = \overline{0}_{g-1}
\end{align}
Without further restrictions there is no $\mathbf{S}(\overline{0})$ such that also $\Delta Y_{g-1,t} (\overline{0}) \indep  \overline{D}_{g,t'} \mid \mathbf{S}(\overline{0}), \overline{D}_{g-1} = \overline{0}_{g-1}$ can be deduced but $\mathbf{S}_{g,t}(\overline{0}) := Pa_{\mathcal{G}(\overline{0})}(Y_{g-1}(\overline{0})) \setminus \{U,U_{Y_{g-1}}, \overline{0}\}, Pa_{\mathcal{G}(\overline{0})}(Y_{t}(\overline{0})) \setminus \{U,U_{Y_{t}}, \overline{0}\}$ is not a subset of $\mathbf{S}(\overline{0})$.
\end{lemma}

\begin{proof}
    By Lemma \ref{lemma:parent_indep} with $m\geq max( g-1,t)$, $\mathbf{B} = D_{s}(\overline{0}_{s-1})$ and $\textcolor{gray}{N = \overline{D}_{s-1}(\overline{0}_{s-2})}$ we get, for all $1 \leq g,t,s \leq \mathcal{T}$:\footnote{Thos expression contains a slight abuse of notation. $D_s(\overline{0}_{s-1}) = D_s(\overline{0}_{m})$ for $s>m+1$.}
\begin{equation*}
    \Delta Y_{g-1,t}(\overline{0}_m) \indep D_{s}(\overline{0}_{s-1}) \mid Pa_{\mathcal{G}(\overline{0}_m)}(Y_{g-1}(\overline{0}_m)) \setminus \{U,U_{Y_{g-1}}, \overline{0}_m\}, Pa_{\mathcal{G}(\overline{0}_m)}(Y_{t}(\overline{0}_m)) \setminus \{U,U_{Y_{t}}, \overline{0}_m\}, \textcolor{gray}{\overline{D}_{s-1}(\overline{0}_{s-2})}, 
\end{equation*}

where without further restrictions, no node in $Pa_{\mathcal{G}(\overline{0}_m)}(Y_{g-1}(\overline{0}_m)) \setminus \{U,U_{Y_{g-1}}, \overline{0}_m\}, Pa_{\mathcal{G}(\overline{0}_m)}(Y_{t}(\overline{0}_m)) \setminus \{U,U_{Y_{t}}, \overline{0}_m\}$ can be left out of the conditioning set (see Lemma \ref{lemma:parent_indep}). 
Let $m = \mathcal{T}$. By Lemma \ref{lemma:D_sequence_contraction}, as assumption \assref{ass:SWAS-staggered} includes Assumption \assref{ass:staggered-trmnt} we get:
\begin{equation*}
    \Delta Y_{g-1,t}(\overline{0}) \indep \overline{D}_{g,t'} \mid Pa_{\mathcal{G}(\overline{0})}(Y_{g-1}(\overline{0})) \setminus \{U,U_{Y_{g-1}}, \overline{0}\}, Pa_{\mathcal{G}(\overline{0})}(Y_{t}(\overline{0})) \setminus \{U,U_{Y_{t}}, \overline{0}\}, \overline{D}_{g-1} = \overline{0}_{g-1},  \ \ 1 \leq g,t\leq \mathcal{T} .
\end{equation*}

\end{proof}

\textbf{Proof of Proposition \ref{prop:min-z0}:}

\begin{proof}
    Special case of Lemma \ref{lemma:min-z0_general} with $t' = \mathcal{T}$. 
\end{proof}

\begin{lemma}\label{lemma:all-vas}
    Consider SCM $\mathcal{M} \in \mathfrak{M}$ and and let $\mathcal{G}(\overline{0})$ be the SWIG constructed in step 3 of Procedure \ref{proc:d-swig}. Under Assumptions \assref{ass:SWAS-staggered} and \assref{ass:no-y-dyn} and consistency we get, for all $g,t$ and $t'\geq max(g-1,t)$:

    \begin{equation*}
        \Delta Y_{g-1,t}(\overline{0}) \indep \overline{X} \setminus \mathbf{S}_{g,t} \mid \mathbf{S}_{g,t}, \overline{D}_{t'} = \overline{0}_{t'}.
    \end{equation*}
\end{lemma}

\begin{proof}
    Let $max(g-1,t) \leq m \leq t'$. Define:
    \begin{align*}
  \mathbf{S}_{g,t}(\overline{0}_m)  &:=  Pa_{\mathcal{G}(\overline{0}_m)}(Y_{g-1}(\overline{0}_m)) \setminus \{U,U_{Y_{g-1}}, \overline{0}_m\}, Pa_{\mathcal{G}(\overline{0}_m)}(Y_{t}(\overline{0}_m)) \setminus \{U,U_{Y_{t}}, \overline{0}_m\},
\end{align*}
    By Lemma \ref{lemma:parent_indep}, with $\mathbf{B} = \overline{X}(\overline{0}_m) \setminus  \mathbf{S}_{g,t}(\overline{0}_m) $,  $\textcolor{gray}{\mathbf{N}} = \overline{D}_{t'}(\overline{0}_m)$ we get:

\begin{equation*}
     \Delta Y_{g-1,t} (\overline{0}_m) \indep \overline{X}(\overline{0}_m) \setminus  \mathbf{S}_{g,t}(\overline{0}_m)  \mid  \mathbf{S}_{g,t}(\overline{0}_m), \overline{D}_{t'}(\overline{0}_m),
\end{equation*}
This implies:
\begin{equation*}
         \Delta Y_{g-1,t} (\overline{0}_m) \indep \overline{X}(\overline{0}_m) \setminus  \mathbf{S}_{g,t}(\overline{0}_m)  \mid  \mathbf{S}_{g,t}(\overline{0}_m) , \overline{D}_{t'}(\overline{0}_m) = \overline{0}_{t'}.
\end{equation*}

by consistency, since $t' \geq m$:
\begin{equation*}
     \Delta Y_{g-1,t} (\overline{0}_m) \indep \overline{X}\setminus \mathbf{S}_{g,t} \mid \mathbf{S}_{g,t}, \overline{D}_{t'} = \overline{0}_{t'}.
\end{equation*}

The statement in the Lemma uses: $\Delta Y_{g-1,t}(\overline{0})  =  \Delta Y_{g-1,t}(\overline{0}_m)$

\end{proof}

\textbf{Proof of Proposition \ref{prop:all-vas}:}

\begin{proof}
    As stated, we assume existence of a valid minimal adjustment set $\mathbf{Z}^{min}_{g,t}(R) = \mathbf{S}_{g,t}$. This set can not be reduced by removing nodes/variables. Let $\mathbf{Z}$ be a set that satisfies $\mathbf{S}_{g,t} \subseteq \mathbf{Z} \subseteq \overline{X}$. In conjunction with Lemma \ref{lemma:all-vas} we can see:
    \begin{align*}
        ATT(g,t) &= DiD_{g,t}(\mathbf{S}_{g,t}) \\
        &= \E[\Delta Y_{g-1,t}\mid \overline{D}_{g-1} = \overline{0}_{g-1}, \underline{D}_g = \underline{1}_g]  - \E[\E[\Delta Y_{g-1,t} \mid \mathbf{S}_{g,t}, \overline{D} = \overline{0}]\mid \overline{D}_{g-1} = \overline{0}_{g-1}, \underline{D}_g = \underline{1}_g] \\
        &= \E[\Delta Y_{g-1,t}\mid \overline{D}_{g-1} = \overline{0}_{g-1}, \underline{D}_g = \underline{1}_g]  - \E[\E[\Delta Y_{g-1,t} \mid \mathbf{Z}, \overline{D} = \overline{0}]\mid \overline{D}_{g-1} = \overline{0}_{g-1}, \underline{D}_g = \underline{1}_g].
    \end{align*}

Meaning that $\mathbf{Z}$ is a valid adjustment set. The first and second equality follow by the definition of a valid adjustment set and $DiD_{g,t}(\mathbf{S}_{g,t})$. The third equality uses that Lemma \ref{lemma:all-vas} with $t' = \mathcal{T}$ and consistency imply:
\begin{align*}
    \Delta Y_{g-1,t}(\overline{0}) &\indep \overline{X} \setminus \mathbf{S}_{g,t} \mid \mathbf{S}_{g,t}, \overline{D} = \overline{0} & \\
    \Rightarrow \Delta Y_{g-1,t} &\indep \overline{X} \setminus \mathbf{S}_{g,t} \mid \mathbf{S}_{g,t}, \overline{D} = \overline{0} & \text{(cons.)} \\
    \Rightarrow \Delta Y_{g-1,t} &\indep \mathbf{Z} \setminus \mathbf{S}_{g,t} \mid \mathbf{S}_{g,t}, \overline{D} = \overline{0} & \text{(decomposition)}
\end{align*}
    
\end{proof}

\subsection{Not-Yet-Treated}\label{sec:NYT_app}

The following is similar to Proposition \ref{prop:min-z0} but adapted for the not-yet-treated.

\begin{proposition} (minimal sufficient adjustment set - not-yet-treated) \label{prop:min-z0-nyt}
    Consider SCM $\mathcal{M} \in \mathfrak{M}$ and let $\mathcal{G}(\overline{0})$ be the corresponding SWIG constructed in step 3 of Procedure \ref{proc:d-swig}.
Then, under Assumptions \assref{ass:SWAS-staggered} and \assref{ass:no-y-dyn} the following conditional independence holds for all $g,t$ and $s \geq g$:
    \begin{align} \label{eq:cia-general-nyt}
    \Delta Y_{g-1,t} (\overline{0}) \indep \overline{D}_{g,s} \mid~ & \underbrace{Pa_{\mathcal{G}(\overline{0})}(Y_{g-1}(\overline{0})) \setminus \{U,U_{Y_{g-1}}, \overline{0}\}, Pa_{\mathcal{G}(\overline{0})}(Y_{t}(\overline{0})) \setminus \{U,U_{Y_{t}}, \overline{0}\}}_{\mathbf{S}_{g,t}(\overline{0})}, \overline{D}_{g-1} = \overline{0}_{g-1}
\end{align}
Without further restrictions there is no $\mathbf{S}(\overline{0})$ such that also $\Delta Y_{g-1,t} (\overline{0}) \indep \underline{D}_{g} \mid \mathbf{S}(\overline{0}), \overline{D}_{g-1} = \overline{0}_{g-1}$ can be deduced but $\mathbf{S}_{g,t}(\overline{0}) := Pa_{\mathcal{G}(\overline{0})}(Y_{g-1}(\overline{0})) \setminus \{U,U_{Y_{g-1}}, \overline{0}\}, Pa_{\mathcal{G}(\overline{0})}(Y_{t}(\overline{0})) \setminus \{U,U_{Y_{t}}, \overline{0}\}$ is not a subset of $\mathbf{S}(\overline{0})$. 
\end{proposition}

\begin{proof}
     Special case of Lemma \ref{lemma:min-z0_general} with $t' = s$.
\end{proof}

For the conditional difference-in-differences estimand based on the not-yet-treated, with $s \geq max(g-1,t)$:
{\small
\begin{align} \label{eq:did-estimand-nyt}
     DiD_{g,t}^{NYT}(\mathbf{Z}) & := \underbrace{\E[\Delta Y_{g-1,t}\mid \overline{D}_{g-1} = \overline{0}_{g-1}, \underline{D}_g = \underline{1}_g]}_{\text{observable trend of group g}}  - \underbrace{\E[\E[\Delta Y_{g-1,t} \mid \mathbf{Z}, \overline{D}_s = \overline{0}_s]\mid \overline{D}_{g-1} = \overline{0}_{g-1}, \underline{D}_g = \underline{1}_g]}_{\text{covariate adjusted not-yet-treated trend}},
\end{align} }

we can now define VAS as all sets of time-varying covariates ensuring that the conditional DiD estimand based on the not-yet-treated is unbiased for $ATT(g,t)$ in all SCMs that obey restrictions $R$, i.e.~$\mathcal{Z}_{g,t}^{NYT}(R) := \{\mathbf{Z} \subseteq \overline{X} : ATT(g,t) = DiD_{g,t}^{NYT}(\mathbf{Z}) ~\forall~ \mathcal{M} \in \mathfrak{M} \text{ satisfying }R \}$.

The fact that Proposition \ref{prop:min-z0-nyt} uses the same conditioning set as Proposition \ref{prop:min-z0}, means that we can use the same reasoning as in the main text for the minimal VAS to use in $DiD_{g,t}^{NYT}(\mathbf{Z})$. We can thus compactly denote $\mathbf{Z}^{min}_{g,t}(R)$ as the minimal valid adjustment set for for $DiD_{g,t}(\mathbf{Z})$ and $DiD_{g,t}^{NYT}(\mathbf{Z})$ and Table \ref{tab:min-vas-did} also holds for DiD based on the not-yet-treated.

In addition:

\begin{proposition} (valid adjustment sets - not-yet-treated) \label{prop:all-vas-nyt}
    Under the assumptions of Proposition \ref{prop:min-z0-nyt}. Whenever there exists a minimal valid adjustment set for model class $\mathfrak{M}$ under restrictions $R$, i.e.~$\mathbf{Z}^{min}_{g,t}(R) = \mathbf{S}_{g,t}$, then
\begin{align}
\mathcal{Z}_{g,t}^{NYT}(R) = \{\mathbf{Z} : \mathbf{Z}^{min}_{g,t}(R) \subseteq \mathbf{Z} \subseteq  \overline{X} \}.
\end{align}
\end{proposition}

\begin{proof}
    As stated, we assume existence of a valid minimal adjustment set $\mathbf{Z}^{min}_{g,t}(R) = \mathbf{S}_{g,t}$. Let $\mathbf{Z}$ be a set that satisfies $\mathbf{S}_{g,t} \subseteq \mathbf{Z} \subseteq \overline{X}$. In conjunction with Lemma \ref{lemma:all-vas} we can see:
    \begin{align*}
        ATT(g,t) &= DiD_{g,t}^{NYT}(\mathbf{S}_{g,t}) \\
        &= \E[\Delta Y_{g-1,t}\mid \overline{D}_{g-1} = \overline{0}_{g-1}, \underline{D}_g = \underline{1}_g]  - \E[\E[\Delta Y_{g-1,t} \mid \mathbf{S}_{g,t}, \overline{D}_s = \overline{0}_s]\mid \overline{D}_{g-1} = \overline{0}_{g-1}, \underline{D}_g = \underline{1}_g] \\
        &= \E[\Delta Y_{g-1,t}\mid \overline{D}_{g-1} = \overline{0}_{g-1}, \underline{D}_g = \underline{1}_g]  - \E[\E[\Delta Y_{g-1,t} \mid \mathbf{Z}, \overline{D}_s = \overline{0}_s]\mid \overline{D}_{g-1} = \overline{0}_{g-1}, \underline{D}_g = \underline{1}_g].
    \end{align*}

Meaning that $\mathbf{Z}$ is a valid adjustment set. The first and second equality follow by the definition of a valid adjustment set and $DiD_{g,t}^{NYT}(\mathbf{S}_{g,t})$. The third equality uses that Lemma \ref{lemma:all-vas} with $t' = s$ and consistency imply:
\begin{align*}
    \Delta Y_{g-1,t}(\overline{0}) &\indep \overline{X} \setminus \mathbf{S}_{g,t} \mid \mathbf{S}_{g,t}, \overline{D}_s = \overline{0}_s & \\
    \Rightarrow \Delta Y_{g-1,t} &\indep \overline{X} \setminus \mathbf{S}_{g,t} \mid \mathbf{S}_{g,t}, \overline{D}_s = \overline{0}_s & \text{(cons.)} \\
    \Rightarrow \Delta Y_{g-1,t} &\indep \mathbf{Z} \setminus \mathbf{S}_{g,t} \mid \mathbf{S}_{g,t}, \overline{D}_s = \overline{0}_s. & \text{(decomposition)}
\end{align*}
    
\end{proof}

\section{SWIGs and SCMs}

This section contains the intermediate SWIGs and SCMs for the DAGs and $\Delta$-SWIGs used in the main text:
\begin{itemize}
    \item The SWIG and structural equations in Figure \ref{fig:SWIG_time_varying_T2} correspond to Figure \ref{fig:DAG_DSWIG_time_varying_T2}.
    \item Figure \ref{fig:SWIG_outcome_dyn} shows the intermediate SWIG and SCM equations for Figure \ref{fig:DAG_DSWIG_outcome_dyn}.
    \item Figure \ref{fig:SWIG_post} shows the intermediate SWIG and structural equations for Figure \ref{fig:DAG_DSWIG_post}.
    \item Figure \ref{fig:SCMs_T3} shows the structural equations for Figures \ref{fig:2x3} and \ref{fig:swigs_3}.
\end{itemize}

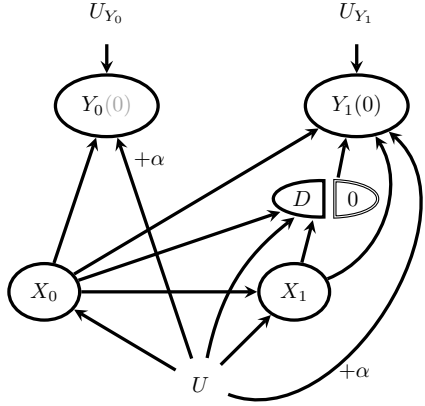
\begin{figure}[h!]
    \centering
    \caption{SWIG and SCM with time-varying covariates, T = 2}
    \label{fig:SWIG_time_varying_T2}
\begin{subfigure}[b]{0.44\textwidth}
\vspace{0pt}
\centering
\makebox[\linewidth][c]{%
  \resizebox{0.9\linewidth}{!}{%
        \begin{tikzpicture}
        \tikzset{line width=1.5pt, ell/.style={draw, fill = white, inner xsep=5pt,inner ysep=5pt, line width=1.5pt}, unobs/.style={ fill = none, inner xsep=5pt,inner ysep=5pt, line width=1.5pt}, swig vsplit={gap=5pt, inner line width right=0.5pt}};

% Nodes
    \node[name=y0, ell, shape=ellipse]{$Y_0\textcolor{gray}{(0)}$};
    \node[name=y1, ell, shape=ellipse] at ($(y0)+(4cm,0cm)$) {$Y_1(0)$};
    % \node[name=d, ell, shape=swig vsplit] at ($(y2)+(-0.5cm,-1.5cm)$) {\nodepart{left}{$D$} \nodepart{right}{$0$}};
    \node[name=d, ell, shape=swig vsplit] at ($(y1)+(-0.5cm,-1.5cm)$) {\nodepart{left}{$D$} \nodepart{right}{$0$}};
    \node[name=x0, ell, shape=ellipse] at ($(y0)+(-1cm,-3cm)$) {$X_0$};
    \node[name=x1, ell, shape=ellipse] at ($(y1)+(-1cm,-3cm)$) {$X_1$};
    \node[name=u_dyx, unobs, shape=ellipse] at ($(y1)+(-2.5cm,-4.5cm)$) {$\UDYX$};
    \node[name=u_y0, unobs, shape=ellipse] at ($(y0)+(0cm,+1.5cm)$) {$U_{Y_0}$};
    \node[name=u_y1, unobs, shape=ellipse] at ($(y1)+(0cm,+1.5cm)$) {$U_{Y_1}$};

% Edges

\draw[->,line width=1.5pt,>=stealth]
(d) edge (y1)
% (u_dyx) edge[color = blue] node[midway, right] {$+ \alpha$} (y0)
% (u_dyx) edge[bend right = 60,color = blue] node[midway, right] {$+ \alpha$} (y1)
(u_dyx) edge node[pos = 0.9, right] {$+ \alpha$} (y0)
(u_dyx) edge[bend right = 80] node[pos = 0.3, right] {$+ \alpha$} (y1)
(u_dyx) edge[bend left = 20] (d)
(x0) edge (y0)
(x0) edge (y1)
(x1) edge[bend right= 50] (y1)
(x0) edge (x1)
(x0) edge (d)
% (x1) edge (d.left south)
(x1) edge (d)
(u_dyx) edge (x0)
(u_dyx) edge (x1)
(u_y0) edge (y0)
(u_y1) edge (y1);
    \end{tikzpicture}
    }
    }
\end{subfigure}
\hfill
\begin{subfigure}[b]{0.55\textwidth}
\begin{minipage}[b]{\linewidth}
{\small
\centering
    \begin{align}
    Y_0 \textcolor{gray}{(0)} = Y_0  &\defeq f_{Y_0}(\UDYX,X_0,U_{Y_0})     \label{eq:T2_basic_Y0} \tag{M1.$Y_0$}  \\
    &= \alpha(\UDYX,X_0) + g_{Y_0}(X_0,U_{Y_0}) \notag \\
    Y_1 &\defeq f_{Y_1}(\UDYX,X_0,X_1,D,U_{Y_1}) \notag \\
    & =\alpha(\UDYX,X_0) + g_{Y_1}(X_0,X_1,U_{Y_1})  \label{eq:T2_basic_Y1} \tag{M1.$Y_1$} \\
    & \quad + D \cdot \tau_1(\UDYX, X_0,X_1, U_{Y_1}) \notag \\
    Y_1(0) &\defeq f_{Y_1}(\UDYX,X_0,X_1,0,U_{Y_1}) \notag \\
    & =\alpha(\UDYX,X_0) + g_{Y_1}(X_0,X_1,U_{Y_1})  \label{eq:T2_basic_Y1_0} \tag{M1.$Y_1(0)$} \\
    \Delta Y_1(0) &\defeq f_{\Delta Y_1(0)}(X_0,X_1,U_{Y_0},U_{Y_1}) \label{eq:T2_basic_DY1} \tag{M1.$\Delta Y_1(0)$} \\
    & \defeq g_{Y_1}(X_0, X_1,U_{Y_1}) - g_{Y_0}(X_0,U_{Y_0}) \notag \\
    D   &\defeq f_{D}(\UDYX,X_0,X_1,U_D) \label{eq:T2_basic_D} \tag{M1.$D$} \\
    X_0  &\defeq f_{X_0}(\UDYX,U_{X_0}) \label{eq:T2_basic_X0} \tag{M1.$X_0$} \\
    X_1  &\defeq f_{X_1}(\UDYX,X_0,U_{X_1}) \label{eq:T2_basic_X1} \tag{M1.$X_1$} \\
    U & \defeq f_U (U_U) \label{eq:T2_basic_U} \tag{M1.$U$}
\end{align}
}
\end{minipage}
\end{subfigure}
\end{figure}

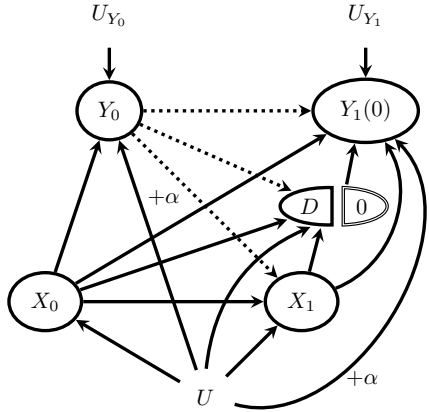
\begin{figure}[h!]
    \centering
    \caption{SWIG and SCM with time varying control variables, T = 2}
    \label{fig:SWIG_outcome_dyn}
\begin{subfigure}[b]{0.44\textwidth}
\vspace{0pt}
\centering
\makebox[\linewidth][c]{%
  \resizebox{0.9\linewidth}{!}{%
        \begin{tikzpicture}
        \tikzset{line width=1.5pt, ell/.style={draw, fill = white, inner xsep=5pt,inner ysep=5pt, line width=1.5pt}, unobs/.style={ fill = none, inner xsep=5pt,inner ysep=5pt, line width=1.5pt}, swig vsplit={gap=5pt, inner line width right=0.5pt}};

% Nodes
    \node[name=y0, ell, shape=ellipse]{$Y_0$};
    \node[name=y1, ell, shape=ellipse] at ($(y0)+(4cm,0cm)$) {$Y_1(0)$};
    % \node[name=d, ell, shape=swig vsplit] at ($(y2)+(-0.5cm,-1.5cm)$) {\nodepart{left}{$D$} \nodepart{right}{$0$}};
    \node[name=d, ell, shape=swig vsplit] at ($(y1)+(-0.5cm,-1.5cm)$)  {\nodepart{left}{$D$} \nodepart{right}{$0$}};
    \node[name=x0, ell, shape=ellipse] at ($(y0)+(-1cm,-3cm)$) {$X_0$};
    \node[name=x1, ell, shape=ellipse] at ($(y1)+(-1cm,-3cm)$) {$X_1$};
    \node[name=u_dyx, unobs, shape=ellipse] at ($(y1)+(-2.5cm,-4.5cm)$) {$\UDYX$};
    \node[name=u_y0, unobs, shape=ellipse] at ($(y0)+(0cm,+1.5cm)$) {$U_{Y_0}$};
    \node[name=u_y1, unobs, shape=ellipse] at ($(y1)+(0cm,+1.5cm)$) {$U_{Y_1}$};

% Edges

\draw[->,line width=1.5pt,>=stealth]
(d) edge (y1)
% (u_dyx) edge[color = blue] node[midway, right] {$+ \alpha$} (y0)
% (u_dyx) edge[bend right = 60,color = blue] node[midway, right] {$+ \alpha$} (y1)
(u_dyx) edge node[pos = 0.75, right] {$+ \alpha$} (y0)
(u_dyx) edge[bend right = 75]node[pos = 0.3, right] {$+ \alpha$} (y1)
(u_dyx) edge[bend left = 30] (d.230)
(x0) edge (y0)
(x0) edge (y1)
(x1) edge[bend right= 50] (y1)
(x0) edge (x1)
(x0) edge (d)
% (x1) edge (d.left south)
(x1) edge (d)
(u_dyx) edge (x0)
(u_dyx) edge (x1)
(u_y0) edge (y0)
(u_y1) edge (y1)
(y0) edge[dotted]  (y1)
(y0) edge[dotted]  (d)
(y0) edge[dotted]  (x1);
    \end{tikzpicture}
    }
    }
    \caption{2x2 SWIG with outcome dynamics}
        \label{fig:swig-2x2-dyn}
\end{subfigure}
\hfill
\begin{subfigure}[b]{0.55\textwidth}
\begin{minipage}[t]{\linewidth}
{\small
\centering
\eqref{eq:T2_basic_Y0} \& \eqref{eq:T2_basic_D} \& \eqref{eq:T2_basic_X0} \& \eqref{eq:T2_basic_X1} \& \eqref{eq:T2_basic_U}
    \begin{align*}
    Y_1 &\defeq f_{Y_1}(\UDYX,Y_0,X_0,X_1,D,U_{Y_1}) \\
    & =\alpha(\UDYX,X_0) + g_{Y_1}(Y_0,X_0,X_1,U_{Y_1})  \\
    & \quad + D \cdot \tau_1(\UDYX, Y_0,X_0,X_1, U_{Y_1}) \\
    Y_1(0) &\defeq f_{Y_1}(\UDYX,Y_0,X_0,X_1,0,U_{Y_1})  \\
    & =\alpha(\UDYX,X_0) + g_{Y_1}(Y_0,X_0,X_1,U_{Y_1})  \\
   \Delta Y_1(0) &\defeq f_{\Delta Y_1(0)}(Y_0,X_0,X_1,U_{Y_0},U_{Y_1})  \\
    & \defeq g_{Y_1}(Y_0,X_0, X_1,U_{Y_1}) - g_{Y_0}(X_0,U_{Y_0})  
\end{align*}
}
\vspace{0.2cm}
\end{minipage}
    \caption{SCM state dependence}
        \label{fig:scm-2x2-y-y}
\end{subfigure}
\begin{subfigure}{0.45\textwidth}
\vspace{0.5cm}
\begin{minipage}{\linewidth}
{\small
\centering
\eqref{eq:T2_basic_Y0} \& \eqref{eq:T2_basic_Y1} \& \eqref{eq:T2_basic_Y1_0} \& \eqref{eq:T2_basic_DY1}  \& \eqref{eq:T2_basic_X0} \& \eqref{eq:T2_basic_X1} \& \eqref{eq:T2_basic_U} 
    \begin{align*}
    D   &\defeq f_{D}(\UDYX,Y_0,X_0,X_1,U_D) 
\end{align*}
}
\end{minipage}
\vspace{0.2cm}
    \caption{SCM outcome-treatment feedback}
        \label{fig:scm-2x2-y-d}
\end{subfigure}
\hfill
\begin{subfigure}{0.45\textwidth}
\vspace{0.5cm}
\begin{minipage}{\linewidth}
{\small
\centering
\eqref{eq:T2_basic_Y0} \& \eqref{eq:T2_basic_Y1} \& \eqref{eq:T2_basic_Y1_0} \& \eqref{eq:T2_basic_DY1}  \& \eqref{eq:T2_basic_D} \& \eqref{eq:T2_basic_X0}  \& \eqref{eq:T2_basic_U} 
    \begin{align*}
    X_1   &\defeq f_{D}(\UDYX,Y_0,X_0,X_1,U_D) 
\end{align*}
}
\end{minipage}
\vspace{0.2cm}
    \caption{SCM outcome-covariate feedback}
        \label{fig:scm-2x2-y-x}
\end{subfigure}
\end{figure}

\begin{figure}[h!]
    \centering
    \caption{SWIG and SCM with post-treatment control variables, T = 2}
    \label{fig:SWIG_post}
\begin{subfigure}[b]{0.48\textwidth}
\centering
\makebox[\linewidth][c]{%
  \resizebox{0.8\linewidth}{!}{%
        \begin{tikzpicture}
        \tikzset{line width=1.5pt, ell/.style={draw, fill = white, inner xsep=5pt,inner ysep=5pt, line width=1.5pt}, unobs/.style={ fill = none, inner xsep=5pt,inner ysep=5pt, line width=1.5pt}, swig vsplit={gap=5pt, inner line width right=0.5pt}};

% Nodes
    \node[name=y0, ell, shape=ellipse]{$Y_0$};
    \node[name=y1, ell, shape=ellipse] at ($(y0)+(4cm,0cm)$) {$Y_1(0)$};
    % \node[name=d, ell, shape=swig vsplit] at ($(y2)+(-0.5cm,-1.5cm)$) {\nodepart{left}{$D$} \nodepart{right}{$0$}};
    \node[name=d, ell, shape=swig vsplit] at ($(y1)+(-1cm,-1.5cm)$) {\nodepart{left}{$D$} \nodepart{right}{$0$}};
    \node[name=x0, ell, shape=ellipse] at ($(y0)+(-1cm,-3cm)$) {$X_0$};
    \node[name=x1, ell, shape=ellipse] at ($(y1)+(-0.5cm,-3cm)$) {$X_1(0)$};
    \node[name=u_dyx, unobs, shape=ellipse] at ($(y1)+(-2.5cm,-5cm)$) {$\UDYX$};
    \node[name=u_y0, unobs, shape=ellipse] at ($(y0)+(0cm,+1.5cm)$) {$U_{Y_0}$};
    \node[name=u_y1, unobs, shape=ellipse] at ($(y1)+(0cm,+1.5cm)$) {$U_{Y_1}$};

% Edges

\draw[->,line width=1.5pt,>=stealth]
(d) edge (y1)
% (u_dyx) edge[color = blue] node[midway, right] {$+ \alpha$} (y0)
% (u_dyx) edge[bend right = 60,color = blue] node[midway, right] {$+ \alpha$} (y1)
(u_dyx) edge node[pos = 0.9, right] {$+ \alpha$} (y0)
(u_dyx) edge[bend right = 70] node[pos = 0.3, right] {$+ \alpha$} (y1)
(u_dyx) edge (d)
(x0) edge (y0)
(x0) edge (y1)
(x1) edge[bend right = 20]  (y1)
(x0) edge (x1)
(x0) edge (d)
% (x1) edge (d.left south)
(d) edge (x1)
(u_dyx) edge (x0)
(u_dyx) edge[color = orange, dotted]  (x1)
(u_y0) edge (y0)
(u_y1) edge (y1);
    \end{tikzpicture}
    }
    }
\caption{SWIG}
        \label{fig:swig-2x2-d-x}
\end{subfigure}
\hfill
\begin{subfigure}[b]{0.48\textwidth}
\centering
\begin{minipage}{\linewidth}
{\small
\centering \eqref{eq:T2_basic_Y0} \& \eqref{eq:T2_basic_Y1}  \& \eqref{eq:T2_basic_X0} \& \eqref{eq:T2_basic_U} 
\begin{align*}
Y_1(0) &\defeq f_{Y_1}(\UDYX,X_0,X_1(0),0,U_{Y_1}) \\
    & =\alpha(\UDYX,X_0) + g_{Y_1}(X_0,X_1(0),U_{Y_1}) \\
\Delta Y_1(0) &\defeq f_{\Delta Y_1(0)}(X_0,X_1(0),U_{Y_0},U_{Y_1}) \\
    & \defeq g_{Y_1}(X_0, X_1(0),U_{Y_1}) - g_{Y_0}(X_0,U_{Y_0}) \\
D & \defeq f_{D}(\UDYX, X_0, U_D) \\ 
X_1  &\defeq f_{X_1}(\textcolor{orange}{\UDYX},X_0,D,U_{X_1}) \\
X_1(0)  &\defeq f_{X_1}(\textcolor{orange}{\UDYX},X_0,0,U_{X_1}) 
\end{align*}
}
\vspace{0.2 cm}
\end{minipage}
\caption{SCM}
        \label{fig:SCM-2x2-d-x}
\end{subfigure}
\end{figure}
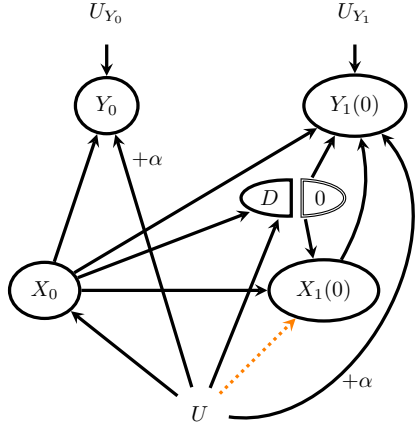

\begin{figure}[h]
    \centering
    \caption{SCMs with $T = 3$}
    \label{fig:SCMs_T3}
\begin{subfigure}[b]{0.99\textwidth}
\centering
\begin{minipage}{\linewidth}
{\small
\begin{align*}
 Y_0 \textcolor{gray}{(0,0)} & \defeq f_{Y_0}(U,X_0,U_{Y_0}) = \alpha(U, X_0) + g_{Y_0}(X_0,U_{Y_0}) \label{eq:T3_Y0} \tag{M2.$Y_0$} \\
 Y_1 &\defeq f_{Y_1}(\UDYX,X_0,X_1,D_1,U_{Y_1}) =\alpha(\UDYX,X_0) + g_{Y_1}(X_0,X_1,U_{Y_1})   \quad + D_1 \cdot \tau_1(\UDYX, X_0,X_1, U_{Y_1}) \label{eq:T3_Y1} \tag{M2.$Y_1$}  \\
Y_1(0\textcolor{gray}{,0}) & \defeq f_{Y_1}(\UDYX,X_0,X_1,0,U_{Y_1}) =\alpha(\UDYX,X_0) + g_{Y_1}(X_0,X_1,U_{Y_1})  \label{eq:T3_Y1_0} \tag{M2.$Y_1(0)$} \\
    \Delta Y_1(0\textcolor{gray}{,0}) &\defeq f_{\Delta Y_1(0)}(X_0,X_1,U_{Y_0},U_{Y_1}) \defeq g_{Y_1}(X_0, X_1,U_{Y_1}) - g_{Y_0}(X_0,U_{Y_0}) \label{eq:T3_DY1} \tag{M2.$\Delta Y_1(0)$}   \\
Y_2 & \defeq f_{Y_2}(\UDYX,X_0,X_1,X_2,D_1,D_2,U_{Y_2}) \notag \\
    & =\alpha(\UDYX,X_0) + g_{Y_2}(X_0,X_1,X_2,U_{Y_2})  \quad + D_2 \cdot \tau_2(\UDYX, X_0,X_1,X_2,D_1, U_{Y_2})   \label{eq:T3_Y2} \tag{M2.$Y_2$} \\
Y_2(0,0) & \defeq f_{Y_2}(\UDYX,X_0,X_1,X_2,0,0,U_{Y_2}) \\
& =\alpha(\UDYX,X_0) + g_{Y_2}(X_0,X_1,X_2,U_{Y_2})  \label{eq:T3_Y2_0} \tag{M2.$Y_2(0,0)$} \\
 \Delta Y_2(0,0) &\defeq f_{\Delta Y_2(0,0)}(X_0,X_1,X_2,U_{Y_1},U_{Y_2})  \defeq g_{Y_2}(X_0,X_1,X_2,U_{Y_2}) - g_{Y_1}(X_0,X_1,U_{Y_1}) \label{eq:T3_DY2} \tag{M2.$\Delta Y_2(0,0)$} \\
    D_1 & \defeq f_{D_1}(\UDYX, X_0, X_1, U_{D_1}) \label{eq:T3_D1} \tag{M2.$D_1$} \\
     D_2 & \defeq f_{D_2}(\UDYX,X_0, X_1, X_2,D_1,U_{D_t}) = (1-D_{1}) g_{D_2}( \UDYX,X_0, X_1, X_2, U_{D_2}) + D_{1} \label{eq:T3_D2} \tag{M2.$D_2$} \\
      D_2(0) & \defeq f_{D_2}(\UDYX,X_0, X_1, X_2,0,U_{D_t}) = g_{D_2}( \UDYX,X_0, X_1, X_2, U_{D_2}) \label{eq:T3_D2_0} \tag{M2.$D_2(0)$} \\
     X_0  &\defeq f_{X_0}(\UDYX,U_{X_0}) \label{eq:T3_X0} \tag{M2.$X_0$} \\
    X_1  &\defeq f_{X_1}(\UDYX,X_0,U_{X_1}) \label{eq:T3_X1} \tag{M2.$X_1$} \\
    X_2  &\defeq f_{X_2}(\UDYX,X_1,U_{X_2}) \label{eq:T3_X2} \tag{M2.$X_2$} \\
    U & \defeq f_U (U_U) \label{eq:T3_U} \tag{M2.$U$}
\end{align*}
}
%\vspace{0.2 cm}
\end{minipage}
\caption{SCM, $T = 3$, $D_{1}\not\rightarrow X_2$}
        \label{fig:SCM-3}
\end{subfigure}
\hfill
\begin{subfigure}[b]{0.99\textwidth}
\centering
\begin{minipage}{\linewidth}
\vspace{0.5cm}
\centering \eqref{eq:T3_Y0} \& \eqref{eq:T3_Y1}  \& \eqref{eq:T3_Y1_0} \& \eqref{eq:T3_DY1} \& \eqref{eq:T3_Y2} \&  \eqref{eq:T3_D1} \& \eqref{eq:T3_D2}  \& \eqref{eq:T3_X0} \& \eqref{eq:T3_X1} \& \eqref{eq:T3_U}
{\small
\begin{align*}
Y_2(0,0) & \defeq f_{Y_2}(\UDYX,X_0,X_1,X_2(0),0,0,U_{Y_2}) \\
& =\alpha(\UDYX,X_0) + g_{Y_2}(X_0,X_1,X_2,U_{Y_2}) \\
 \Delta Y_2(0,0) &\defeq f_{\Delta Y_2(0,0)}(X_0,X_1,X_2(0),U_{Y_1},U_{Y_2})  \defeq g_{Y_2}(X_0,X_1,X_2(0),U_{Y_2}) - g_{Y_1}(X_0,X_1,U_{Y_1})  \\
      D_2(0) & \defeq f_{D_2}(\UDYX,X_0, X_1, X_2(0),0,U_{D_t}) = g_{D_2}( \UDYX,X_0, X_1, X_2(0), U_{D_2}) \\
    X_2  &\defeq f_{X_2}(\UDYX,X_1,D_1,U_{X_2}) \\
        X_2(0)  &\defeq f_{X_2}(\UDYX,X_1,0,U_{X_2})
\end{align*}
}
%\vspace{0.2 cm}
\end{minipage}
\caption{SCM, $T = 3$, $D_{1}\rightarrow X_2$}
        \label{fig:fig:SCM-3-d-x}
\end{subfigure}
\end{figure}

\FloatBarrier

\section{Simulations}\label{app:simulations}

The simulation is an implementation of the DAG with multiple time periods, variation in treatment timing, and three different unobserved confounders, i.e., $U_{XY}$, $U_{XD}$ and $U_{DY}$, which we draw as:

{\singlespacing
\begin{align*}
(U_{DY}, U_{XD}, U_{XY})' &= \mathcal{N}((0,0,0)', \mathbf{\Sigma}) \\
\mathbf{\Sigma} &= \left[\begin{array}
{rrr}
1 & \rho & \rho \\
\rho & 1 & \rho  \\
\rho & \rho & 1 \\
\end{array}\right]
\end{align*}
}

with $\rho = 0.9$.

Furthermore initialize:
\begin{align*}
    X_{0,i} &= X_{1,i} (\mathbf{0}) = \mathbf{1}[0.3\cdot U_{XD,i} + 0.3 \cdot U_{XY,i} + \varepsilon_{X_1,i} > 0] ,\quad  \varepsilon_{X_1} \sim \mathcal{N}(0,1.5) \\
    D_{0,i} &= 0
\end{align*}

Then, for the covariates $X_t$ in $t = 1,...\mathcal{T}$:
\begin{align*}
    v_{t,i} &= 0.3 \cdot U_{XD,i} + 0.3 \cdot U_{XY,i} + 0.1 \cdot \left(X_{t-1,i} - 0.6 \right)+ \beta_{XD} \cdot D_{t-1,i} + \varepsilon_{X_t,i} , \quad  \varepsilon_{X_t,i} \sim \mathcal{N}(0,1.5) \\
X_{t,i} &= \mathbf{1}[v_{t,i} > 0]
\end{align*}

The treatment sequence for $t = 1,...,\mathcal{T}$:
\begin{align*}
    D_{t,i} &=D_{t-1,i} + (1-D_{t-1,i}) \cdot \mathbf 1 \big[ U_{DY, i} + U_{XD, i} + 0.15 \cdot t \cdot \left (X_{t-1,i}- 0.6 \right) \\
    &\qquad \qquad + 0.35 \cdot t \cdot \left(X_{t,i} - 0.6 \right) + \varepsilon_{D_t,i} > 0.7 - 0.3t\big],\\
    \varepsilon_{D_t,i} &\sim \mathcal{N}(0,1.5)
\end{align*}
This encodes staggered adoption.
And for the outcomes:
\begin{align*}
    Y_{t,i} &= 0.6 \cdot \sqrt{t} + U_{DY,i} + 0.8\cdot U_{XY,i} + (0.5 + 0.1 t)\cdot X_{t,i} + \varepsilon_{Y_t,i}  \\
    & \qquad + D_{t,i} \cdot (0.2 + 0.2 \cdot t \cdot U_{DY,i} +0.05 \cdot  U_{XY,i} \cdot X_{t,i}), \\ 
    \varepsilon_{Y_t,i} &\sim \mathcal{N}(0,0.1)
\end{align*}

This satisfies Assumptions \assref{ass:SWAS-staggered}, \assref{ass:no-y-dyn}, \assref{ass:no-x-dyn} and \assref{ass:x-d-ordering}. 
$\beta_{XD}$ can be used to specify whether $D_{t-1}$ has any effect on $X_t$, i.e., changes its distribution. For the simulations with $D_{t-1} \rightarrow X_t$ feedback we set $\beta_{XD} = 0.5$, and for the simulations without $D_{t-1} \rightarrow X_t$ feedback we set $\beta_{XD} = 0$, thus making Assumption \assref{ass:no-d-x} hold.
The different estimation strategies correspond to $DiD_{g,t}(\mathbf{Z})$ as defined in \eqref{eq:did-estimand}, where we use $\mathbf{Z} = \overline{X}_g$ for "pre-treatment", $\mathbf{Z} = \overline{X}_{max(g-1,t)}$ for "pre-outcome", and $\mathbf{Z} = \overline{X}$ for "full sequence." 
As all $X_t$ are binary, we use saturated regressions to estimate the conditional expectations.
The sample size is set to 10,000,000 to limit finite-sample distortions.

\section{More Results}

\subsection{All Valid Adjustment Sets} \label{app:vas}

Table \ref{tab:vas-did} compactly provides all valid adjustment sets for $\mathfrak{M}$ model class. The gray sequences or any of their subsequences are optional.

\begin{table}[ht] \small
\caption{Valid adjustment sets under different model restrictions} \label{tab:vas-did}
\hspace*{-1.2cm}
    \centering
\begin{threeparttable} \small
    \begin{tabular}{cccccc|ccc}
      \multicolumn{5}{l}{Restrictions:} & & Parallel pre-trends? & Short-term effect? & Dynamic effect? \\
       \assref{ass:SWAS-staggered} & \assref{ass:no-y-dyn} & \assref{ass:x-d-ordering} & \assref{ass:no-d-x} & \assref{ass:no-x-dyn} & \assref{ass:no-x-y} & $DiD_{g,t}(\cdot) = 0 ~\forall~ g > t$ & $DiD_{g,g}(\cdot) = \tau_{g,g}$ & $DiD_{g,t}(\cdot) = \tau_{g,t} ~\forall~ g < t$ \\
         \hline
      no & n/y & n/y & n/y & n/y & n/y
        & --
        & -- 
        & --  \\
      n/y & no & n/y & n/y & n/y & n/y
        & --
        & -- 
        & --  \\
      yes & yes & no & no & no & no
        & $\overline{X}_{g-1}\textcolor{gray}{,\underline{X}_{g}}$
        & -- 
        & --  \\
      yes & yes & yes & no & no & no
        & $\overline{X}_{g-1}\textcolor{gray}{,\underline{X}_{g}}$
        & $\overline{X}_{g}\textcolor{gray}{,\underline{X}_{g+1}}$
        & --  \\    
      yes & yes & yes & yes & no & no
        & $\overline{X}_{g-1}\textcolor{gray}{,\underline{X}_{g}}$
        & $\overline{X}_{g}\textcolor{gray}{,\underline{X}_{g+1}}$
        & $\overline{X}_{t}\textcolor{gray}{,\underline{X}_{t+1}}$ \\    
      yes & yes & no & no & yes & no
        & $\textcolor{gray}{\overline{X}_{t-1},} X_{t}\textcolor{gray}{,\overline{X}_{t+1,g-2}},X_{g-1}\textcolor{gray}{,\underline{X}_{g}}$ 
        & -- 
        & --  \\
      yes & yes & yes & no & yes & no
        & $\textcolor{gray}{\overline{X}_{t-1},} X_{t}\textcolor{gray}{,\overline{X}_{t+1,g-2}},X_{g-1}\textcolor{gray}{,\underline{X}_{g}}$ 
        & $\textcolor{gray}{\overline{X}_{g-2},} X_{g-1},X_g \textcolor{gray}{,\underline{X}_{g+1}}$  
        & --  \\    
      yes & yes & yes & yes & yes & no
        & $\textcolor{gray}{\overline{X}_{t-1},} X_{t}\textcolor{gray}{,\overline{X}_{t+1,g-2}},X_{g-1}\textcolor{gray}{,\underline{X}_{g}}$ 
        & $\textcolor{gray}{\overline{X}_{g-2},} X_{g-1},X_g \textcolor{gray}{,\underline{X}_{g+1}}$  
        &  $\textcolor{gray}{\overline{X}_{g-2}}, X_{g-1}\textcolor{gray}{,\overline{X}_{g,t-1}},X_t\textcolor{gray}{,\underline{X}_{t+1}}$ \\
        yes & yes & n/y & n/y & yes & yes
        & $\textcolor{gray}{\overline{X}}$ 
        & $\textcolor{gray}{\overline{X}}$
        & $\textcolor{gray}{\overline{X}}$ \\
        \hline
    \end{tabular}
\begin{tablenotes}[flushleft]
\footnotesize
\item \textit{Notes:} \assref{ass:SWAS-staggered}: single world additive separability and staggered treatment; \assref{ass:no-y-dyn}: no outcome dynamics; \assref{ass:x-d-ordering}: no within-period treatment-covariate effect; \assref{ass:no-d-x}: no treatment-covariate feedback; \assref{ass:no-x-dyn}: no direct covariate-outcome dynamics; \assref{ass:no-x-y}: no within-period covariate-outcome effect. n/y indicates that this restriction does not affect the result. $DiD_{g,t}(\cdot)$ is defined in \eqref{eq:did-estimand}. $\tau_{g,t} = ATT(g,t)$ is defined in \eqref{eq:att-gt}. The cell entries provide the minimal valid adjustment set such that the equality in the header holds when it is inserted into the bracket of $DiD_{g,t}(\cdot)$. -- indicates that no valid adjust set exists.
The table is representative for settings where $Y_t$ realizes at the end of period $t$. This is without loss of generality because any setting can be rearranged accordingly.
\end{tablenotes}
\end{threeparttable}
\end{table}

\FloatBarrier

\subsection{$T=4$} \label{app:t4}

Consider Figure \ref{fig:SWIG_T4} for illustration. It depicts a setting with $T=4$ and limited covariate dynamics, i.e.~no edges from past covariates to future treatments or outcomes. This keeps the graph readable and captures causal structures that assume that $Y_t(\overline{0})$ only depends on $X_t$, as is common in the literature (see Appendix \ref{app:swas-lit} for a detailed discussion).
Proposition \ref{prop:min-z0} applied to the SWIG without treatment-covariate feedback in Figure \ref{fig:swig-2x4} leads for example to
\begin{enumerate}
    \item $\Delta Y_{1,2} (\overline{0}) \indep D_2,D_3 \mid X_{1}, X_{2}, D_{1} = 0 \Rightarrow ATT(2,2)$ is identified; $\mathbf{S}_{2,2}(\overline{0}) = \mathbf{S}_{2,2} = \{X_{1}, X_{2}\}$
    \item $\Delta Y_{1,3} (\overline{0}) \indep D_2,D_3 \mid X_{1}, X_{3}, D_{1} = 0 \Rightarrow ATT(2,3)$ is identified; $\mathbf{S}_{2,3}(\overline{0}) = \mathbf{S}_{2,3} = \{X_{1}, X_{3}\}$
\end{enumerate}
or applied to SWIG \ref{fig:swig-2x4-d-x} with treatment-covariate feedback yields
\begin{enumerate}[start=3]
    \item $\Delta Y_{1,2} (\overline{0}) \indep D_2,D_3 \mid X_{1}, X_{2}(\dzero), D_{1} = \dzero \Rightarrow ATT(2,2)$ is identified; $\mathbf{S}_{2,2}(\overline{0}) = \{X_{1}, X_{2}(0)\}$, $\mathbf{S}_{2,2} = \{X_{1}, X_{2}\}$
    \item $\Delta Y_{1,3} (\overline{0}) \indep D_2,D_3 \mid X_{1}, X_{3}(0,0), D_{1} = 0 \Rightarrow ATT(2,3)$ is \textit{not} identified; $\mathbf{S}_{2,3}(\overline{0}) = \{X_{1}, X_{3}(0,0)\}$, $\mathbf{S}_{2,3} = \{X_{1}, X_{3}\}$
\end{enumerate}

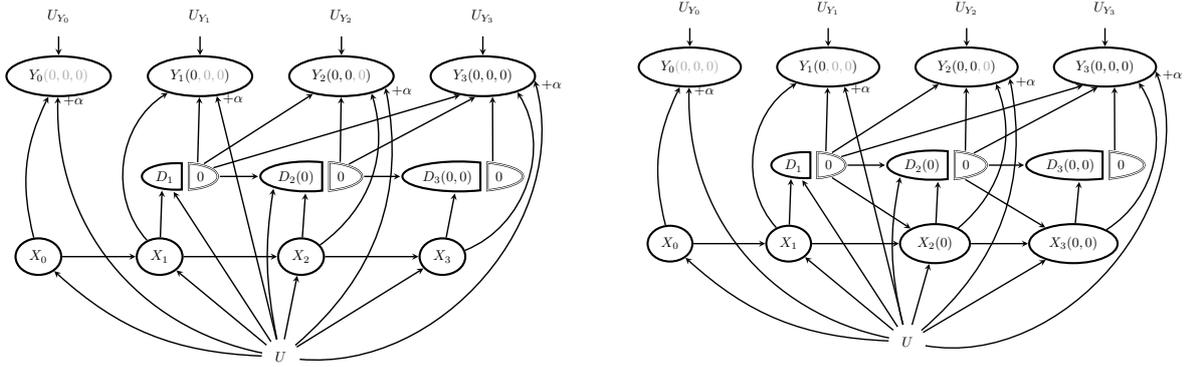
\begin{figure}[h]
    \centering
    \caption{SWIGs with $T = 4$ and no covariate dynamics}
    \label{fig:SWIG_T4}
\begin{subfigure}[b]{0.48\textwidth}
\centering
\makebox[\linewidth][c]{%
  \resizebox{0.98\linewidth}{!}{%
    \begin{tikzpicture}
        \tikzset{line width=1.5pt, ell/.style={draw, fill = white, inner xsep=5pt,inner ysep=5pt, line width=1.5pt}, unobs/.style={ fill = none, inner xsep=5pt,inner ysep=5pt, line width=1.5pt}, swig vsplit={gap=5pt, inner line width right=0.5pt}};

% Nodes
    \node[name=y0, ell, shape=ellipse]{$Y_0\textcolor{gray}{(0,0,0)}$};
    \node[name=y1, ell, shape=ellipse] at ($(y0)+(3.5cm,0cm)$) {$Y_1(0\textcolor{gray}{,0,0})$};
    \node[name=y2, ell, shape=ellipse] at ($(y1)+(3.5cm,0cm)$) {$Y_2(0,0\textcolor{gray}{,0})$};
    \node[name=y3, ell, shape=ellipse] at ($(y2)+(3.5cm,0cm)$) {$Y_3(0,0,0)$};
    \node[name=d1, ell, shape=swig vsplit] at ($(y1)+(-0.5cm,-2.5cm)$){\nodepart{left}{$D_1$} \nodepart{right}{$0$}};
    \node[name=d2, ell, shape=swig vsplit] at ($(y2)+(-0.75cm,-2.5cm)$) {\nodepart{left}{$D_2(0)$} \nodepart{right}{$0$}};
    \node[name=d3, ell, shape=swig vsplit] at ($(y3)+(-0.5cm,-2.5cm)$) {\nodepart{left}{$D_3(0,0)$} \nodepart{right}{$0$}};
    \node[name=x0, ell, shape=ellipse] at ($(y0)+(-0.5cm,-4.5cm)$) {$X_0$};
    \node[name=x1, ell, shape=ellipse] at ($(y1)+(-1cm,-4.5cm)$) {$X_1$};
    \node[name=x2, ell, shape=ellipse] at ($(y2)+(-1cm,-4.5cm)$) {$X_2$};
    \node[name=x3, ell, shape=ellipse] at ($(y3)+(-1cm,-4.5cm)$) {$X_3$};
    \node[name=u_dyx, unobs, shape=ellipse] at ($(y1)+(2cm,-7cm)$) {$\UDYX$};
    \node[name=u0, unobs, shape=ellipse] at ($(y0)+(0cm,+1.5cm)$) {$U_{Y_0}$};
    \node[name=u1, unobs, shape=ellipse] at ($(y1)+(0cm,+1.5cm)$) {$U_{Y_1}$};
    \node[name=u2, unobs, shape=ellipse] at ($(y2)+(0cm,+1.5cm)$) {$U_{Y_2}$};
    \node[name=u3, unobs, shape=ellipse] at ($(y3)+(0cm,+1.5cm)$) {$U_{Y_3}$};

% Edges

\draw[->,line width=1pt,>=stealth]
(d1) edge (y2)
(d1) edge (y3.220)
(d1) edge (d2)
(d2) edge (d3)
(u_dyx) edge[bend left = 40] node[pos = 0.99, right] {$+ \alpha$}  (y0)
(u_dyx) edge node[pos = 0.99, right] {$+ \alpha$} (y1.310)
(u_dyx) edge[bend right = 30] node[pos = 0.99, right] {$+ \alpha$} (y2.345)
(u_dyx) edge[bend right = 55] node[pos = 0.99, right] {$+ \alpha$} (y3.355)
(u_dyx) edge (d1.260)
(u_dyx) edge[bend left = 10] (d2.200)
(u_dyx) edge[bend left = 20] (x0)
(u_dyx) edge (x1)
(u_dyx) edge (x2)
(u_dyx) edge (x3)
(u0) edge (y0)
(u1) edge (y1)
(u2) edge (y2)
(u3) edge (y3)
(x0) edge[bend left = 20] (y0)
(x1) edge[bend left = 50] (y1)
(x2) edge[bend right = 35] (y2.330)
(x3) edge[bend right = 50] (y3.335)
(x1) edge (d1.225)
(x2) edge (d2)
(x3) edge (d3)
(x0) edge (x1)
(x1) edge (x2)
(x2) edge (x3)
;

\draw[->,line width=1pt,>=stealth] (d2.30) -- (y2);
\draw[->,line width=1pt,>=stealth] (d1.45) -- (y1);
\draw[->,line width=1pt,>=stealth] (d3.30) -- (y3.295);
\draw[->,line width=1pt,>=stealth] (d2.20) -- (y3.250);

    \end{tikzpicture}
    }
    }
\begin{minipage}{\linewidth}
{\footnotesize
\begin{align*}
    Y_0  &\defeq f_{Y_0}(\UDYX,X_0,U_{Y_0})        \\
    &= \alpha(\UDYX) + g_{Y_0}(X_0,U_{Y_0})  \\
    Y_t(\overline{0}_t) & \defeq \alpha(\UDYX) + g_{Y_t}(X_t,U_{Y_t}), \ \ t \geq 1  \\
    D_1   &\defeq f_{D_1}(\UDYX,X_1,U_{D_1}),  \\
    D_t(\overline{0}_{t-1})   &\defeq f_{D_t}(\UDYX,X_t(\overline{0}_{t-1}),\overline{0}_{t-1},U_{D_t}), \ t \geq 2  \\
    X_0  &\defeq f_{X_0}(\UDYX,U_{X_0})   \\
    X_1  &\defeq f_{X_1}(\UDYX,U_{X_1})  \\
    X_t &\defeq f_{X_t}(\UDYX,U_{X_t}),  \ \ t \geq 2 
\end{align*} 
}
\end{minipage}
\caption{Without treatment-covariate feedback}
        \label{fig:swig-2x4}
\end{subfigure}
\hfill
\begin{subfigure}[b]{0.48\textwidth}
\centering
\makebox[\linewidth][c]{%
  \resizebox{0.98\linewidth}{!}{%
    \begin{tikzpicture}
        \tikzset{line width=1.5pt, ell/.style={draw, fill = white, inner xsep=5pt,inner ysep=5pt, line width=1.5pt}, unobs/.style={ fill = none, inner xsep=5pt,inner ysep=5pt, line width=1.5pt}, swig vsplit={gap=5pt, inner line width right=0.5pt}};

% Nodes
    \node[name=y0, ell, shape=ellipse]{$Y_0\textcolor{gray}{(0,0,0)}$};
    \node[name=y1, ell, shape=ellipse] at ($(y0)+(3.5cm,0cm)$) {$Y_1(0\textcolor{gray}{,0,0})$};
    \node[name=y2, ell, shape=ellipse] at ($(y1)+(3.5cm,0cm)$) {$Y_2(0,0\textcolor{gray}{,0})$};
    \node[name=y3, ell, shape=ellipse] at ($(y2)+(3.5cm,0cm)$) {$Y_3(0,0,0)$};
    \node[name=d1, ell, shape=swig vsplit] at ($(y1)+(-0.5cm,-2.5cm)$){\nodepart{left}{$D_1$} \nodepart{right}{$0$}};
    \node[name=d2, ell, shape=swig vsplit] at ($(y2)+(-0.75cm,-2.5cm)$) {\nodepart{left}{$D_2(0)$} \nodepart{right}{$0$}};
    \node[name=d3, ell, shape=swig vsplit] at ($(y3)+(-0.5cm,-2.5cm)$) {\nodepart{left}{$D_3(0,0)$} \nodepart{right}{$0$}};
    \node[name=x0, ell, shape=ellipse] at ($(y0)+(-0.5cm,-4.5cm)$) {$X_0$};
    \node[name=x1, ell, shape=ellipse] at ($(y1)+(-1cm,-4.5cm)$) {$X_1$};
    \node[name=x2, ell, shape=ellipse] at ($(y2)+(-0.8cm,-4.5cm)$) {$X_2(0)$};
    \node[name=x3, ell, shape=ellipse] at ($(y3)+(-0.8cm,-4.5cm)$) {$X_3(0,0)$};
    \node[name=u_dyx, unobs, shape=ellipse] at ($(y1)+(2cm,-7cm)$) {$\UDYX$};
    \node[name=u0, unobs, shape=ellipse] at ($(y0)+(0cm,+1.5cm)$) {$U_{Y_0}$};
    \node[name=u1, unobs, shape=ellipse] at ($(y1)+(0cm,+1.5cm)$) {$U_{Y_1}$};
    \node[name=u2, unobs, shape=ellipse] at ($(y2)+(0cm,+1.5cm)$) {$U_{Y_2}$};
    \node[name=u3, unobs, shape=ellipse] at ($(y3)+(0cm,+1.5cm)$) {$U_{Y_3}$};

% Edges

\draw[->,line width=1pt,>=stealth]
(d1) edge (y2)
(d1) edge (y3.220)
(d1) edge (d2)
(d2) edge (d3)
(u_dyx) edge[bend left = 40] node[pos = 0.99, right] {$+ \alpha$}  (y0)
(u_dyx) edge node[pos = 0.99, right] {$+ \alpha$} (y1.310)
(u_dyx) edge[bend right = 30] node[pos = 0.99, right] {$+ \alpha$} (y2.345)
(u_dyx) edge[bend right = 60] node[pos = 0.99, right] {$+ \alpha$} (y3.355)
(u_dyx) edge (d1.260)
(u_dyx) edge[bend left = 12] (d2.200)
(u_dyx) edge[bend left = 20] (x0)
(u_dyx) edge (x1)
(u_dyx) edge (x2)
(u_dyx) edge (x3)
(u0) edge (y0)
(u1) edge (y1)
(u2) edge (y2)
(u3) edge (y3)
(x0) edge[bend left = 20] (y0)
(x1) edge[bend left = 50] (y1)
(x2) edge[bend right = 35] (y2.330)
(x3) edge[bend right = 45] (y3.335)
(x1) edge (d1.225)
(x2) edge (d2.285)
(d1) edge (x2)
(x3) edge (d3)
(d2) edge (x3)
(x0) edge (x1)
(x1) edge (x2)
(x2) edge (x3)
;

\draw[->,line width=1pt,>=stealth] (d2.30) -- (y2);
\draw[->,line width=1pt,>=stealth] (d1.45) -- (y1);
\draw[->,line width=1pt,>=stealth] (d3.30) -- (y3.295);
\draw[->,line width=1pt,>=stealth] (d2.20) -- (y3.250);

    \end{tikzpicture}
    }
    }
\begin{minipage}{\linewidth}
{\footnotesize
\begin{align*}
    Y_0  &\defeq f_{Y_0}(\UDYX,X_0,U_{Y_0})        \\
    &= \alpha(\UDYX) + g_{Y_0}(X_0,U_{Y_0})  \\
    Y_t(\overline{0}_t) & \defeq \alpha(\UDYX) + g_{Y_t}(X_t(\overline{0}_{t-1}),U_{Y_t}), \ \ t \geq 1  \\
    D_1   &\defeq f_{D_1}(\UDYX,X_1,U_{D_1}),  \\
    D_t(\overline{0}_{t-1})   &\defeq f_{D_t}(\UDYX,X_t(\overline{0}_{t-1}),\overline{0}_{t-1},U_{D_t}), \ t \geq 2  \\
    X_0  &\defeq f_{X_0}(\UDYX,U_{X_0})   \\
    X_1  &\defeq f_{X_1}(\UDYX,U_{X_1})  \\
    X_t(\overline{0}_{t-1}) &\defeq f_{X_t}(\UDYX,\overline{0}_{t-1},U_{X_t}),  \ \ t \geq 2 
\end{align*} 
}
\end{minipage}
\caption{With treatment-covariate feedback}
        \label{fig:swig-2x4-d-x}
\end{subfigure}
\end{figure}

\end{document}